\documentclass[aps,10pt,twocolumn,nofootinbib,showpacs,pra,amsmath,amssymb,mathtools,floatfix,superscriptaddress,dvipsnames]{revtex4-2}

\usepackage{soul}
\usepackage[normalem]{ulem}

\newcommand{\blk}{\color{black}}

\newcommand{\ulmshort}{Institute of Theoretical Physics, Ulm University, Ulm, Germany}
\newcommand{\gdanskshort}{ICTQT -- International Centre for Theory of Quantum Technologies, Gda\'nsk, Poland} 
\newcommand{\perimetershort}{Perimeter Institute for Theoretical Physics, Waterloo, Canada} 
\newcommand{\waterlooshort}{Department of Physics and Astronomy, University of Waterloo, Waterloo, Canada} 

\usepackage{tikz-cd}
\usepackage{tikzfig}

\usepackage{booktabs}
\usepackage{amsmath,amsthm,amssymb}
\usepackage{dsfont}
\usepackage{mathtools,amssymb,amsfonts,physics,graphicx,multirow,float,bm}
\usepackage{bbold}
\usepackage{xspace,enumerate,epsfig}
\usepackage{graphicx}
\usepackage{marvosym}

\usepackage{color,xcolor}
\usepackage[colorlinks=true,linktocpage=true]{hyperref}
\hypersetup{
    allcolors  = {blue!90},
}
\usepackage{cleveref}

\usepackage{docmute}
\usepackage{keycommand}

\usetikzlibrary{shapes.multipart}

\tikzset{->-/.style={decoration={
  markings,
  mark=at position .5 with {\arrow{>}}},postaction={decorate}}}
\tikzset{-<-/.style={decoration={
  markings,
  mark=at position .5 with {\arrow{<}}},postaction={decorate}}}

\tikzstyle{bwSpider}=[
       rectangle split,
       rectangle split parts=2,
       rectangle split part fill={black,white},
 minimum size=3.6 mm, inner sep=-2mm, draw=black,scale=0.5,rounded corners=0.8 mm
       ]
 \tikzstyle{wbSpider}=[
       rectangle split,
       rectangle split parts=2,
       rectangle split part fill={white,black},
 minimum size=3.6 mm, inner sep=-2mm, draw=black,scale=0.5,rounded corners=0.8 mm
       ]

\tikzstyle{epiCopoint}=[regular polygon,regular polygon sides=3,draw,scale=0.75,inner sep=-0.5pt,minimum width=5mm,fill=white,regular polygon rotate=0,line width=1pt]
\tikzstyle{epiPoint}=[regular polygon,regular polygon sides=3,draw,scale=0.75,inner sep=-0.5pt,minimum width=5mm,fill=white,regular polygon rotate=180,line width=1pt]
\tikzstyle{epiPointWide}=[regular polygon,regular polygon sides=3,draw,scale=0.75,inner sep=-0.5pt,minimum width=8mm,fill=white,regular polygon rotate=180,line width=1pt]
\tikzstyle{epiBox}=[fill=white,draw, line width = 1pt,inner sep=0.6mm,font=\footnotesize,minimum height=3mm,minimum width=3mm]
\tikzstyle{epiBoxWide}=[fill=white,draw, line width = 1pt,inner sep=0.6mm,font=\footnotesize,minimum height=3mm,minimum width=5mm]
\tikzstyle{epiBoxVeryWide}=[fill=white,draw, line width = 1pt,inner sep=0.6mm,font=\footnotesize,minimum height=3mm,minimum width=7mm]
\tikzstyle{qWire}=[line width = 1pt, color=black]
\tikzstyle{cWire}=[color=gray,line width = .75pt]
\tikzstyle{CqWire}=[color=gray,line width = .75pt,->-]
\tikzstyle{CcWire}=[color=gray,line width = .75pt,->-]
\tikzstyle{RqWire}=[line width = 1pt, color=black,-<-]
\tikzstyle{RcWire}=[color=gray,line width = .75pt,-<-]
\tikzstyle{env}=[copoint,regular polygon rotate=0,minimum width=0.2cm, fill=black]

\tikzstyle{probs}=[shape=semicircle,fill=white,draw=black,shape border rotate=180,minimum width=1.2cm]

\tikzstyle{every picture}=[baseline=-0.25em,scale=0.5]
\tikzstyle{dotpic}=[] 
\tikzstyle{diredges}=[every to/.style={diredge}]
\tikzstyle{math matrix}=[matrix of math nodes,left delimiter=(,right delimiter=),inner sep=2pt,column sep=1em,row sep=0.5em,nodes={inner sep=0pt},text height=1.5ex, text depth=0.25ex]

\tikzstyle{inline text}=[text height=1.5ex, text depth=0.25ex,yshift=0.5mm]
\tikzstyle{label}=[font=\footnotesize,text height=1.5ex, text depth=0.25ex,yshift=0.5mm]
\tikzstyle{left label}=[label,anchor=east,xshift=1.5mm]
\tikzstyle{right label}=[label,anchor=west,xshift=-1mm]
\tikzstyle{up label}=[label,anchor=south,yshift=-1mm]

\tikzstyle{braceedge}=[decorate,decoration={brace,amplitude=2mm,raise=-1mm}]
\tikzstyle{small braceedge}=[decorate,decoration={brace,amplitude=1mm,raise=-1mm}]

\tikzstyle{doubled}=[line width=1.6pt] 
\tikzstyle{boldedge}=[doubled,shorten <=-0.17mm,shorten >=-0.17mm]
\tikzstyle{boldedgegray}=[doubled,gray,shorten <=-0.17mm,shorten >=-0.17mm]
\tikzstyle{singleedgegray}=[gray]

\tikzstyle{semidoubled}=[line width=1.4pt] 
\tikzstyle{semiboldedgegray}=[semidoubled,gray,shorten <=-0.17mm,shorten >=-0.17mm]

\tikzstyle{boxedge}=[semiboldedgegray]

\tikzstyle{boldedgedashed}=[very thick,dashed,shorten <=-0.17mm,shorten >=-0.17mm]
\tikzstyle{vboldedgedashed}=[doubled,dashed,shorten <=-0.17mm,shorten >=-0.17mm]
\tikzstyle{left hook arrow}=[left hook-latex]
\tikzstyle{right hook arrow}=[right hook-latex]
\tikzstyle{sembracket}=[line width=0.5pt,shorten <=-0.07mm,shorten >=-0.07mm]

\tikzstyle{causal edge}=[->,thick,gray]
\tikzstyle{causal nondir}=[thick,gray]
\tikzstyle{timeline}=[thick,gray, dashed]

\tikzstyle{cedge}=[<->,thick,gray!70!white]

\tikzstyle{empty diagram}=[draw=gray!40!white,dashed,shape=rectangle,minimum width=1cm,minimum height=1cm]
\tikzstyle{empty diagram small}=[draw=gray!50!white,dashed,shape=rectangle,minimum width=0.6cm,minimum height=0.5cm]

\tikzstyle{dot}=[inner sep=0mm,minimum width=2mm,minimum height=2mm,draw,shape=circle]
\tikzstyle{bigdot}=[inner sep=0mm,minimum width=5mm,minimum height=5mm,draw,shape=circle]
\tikzstyle{leak}=[white dot, shape=regular polygon, minimum size=3.3 mm, regular polygon sides=3, outer sep=-0.2mm, regular polygon rotate=270]
\tikzstyle{proj}=[regular polygon,regular polygon sides=4,draw,scale=0.75,inner sep=-0.5pt,minimum width=6mm,fill=white]
\tikzstyle{projOut}=[regular polygon,regular polygon sides=3,draw,scale=0.75,inner sep=-0.5pt,minimum width=7.5mm,fill=white,regular polygon rotate=180]
\tikzstyle{projIn}=[regular polygon,regular polygon sides=3,draw,scale=0.75,inner sep=-0.5pt,minimum width=7.5mm,fill=white]
\tikzstyle{Vleak}=[white dot, shape=regular polygon, minimum size=3.3 mm, regular polygon sides=3, outer sep=-0.2mm, regular polygon rotate=90]
\tikzstyle{dleak}=[white dot, line width=1.6pt, shape=regular polygon, minimum size=3.3 mm, regular polygon sides=3, outer sep=-0.2mm, regular polygon rotate=270]

\tikzstyle{Wsquare}=[white dot, shape=regular polygon, rounded corners=0.8 mm, minimum size=3.3 mm, regular polygon sides=3, outer sep=-0.2mm]
\tikzstyle{Wsquareadj}=[white dot, shape=regular polygon, rounded corners=0.8 mm, minimum size=3.3 mm, regular polygon sides=3, outer sep=-0.2mm, regular polygon rotate=180]
\tikzstyle{ddot}=[inner sep=0mm, doubled, minimum width=2.5mm,minimum height=2.5mm,draw,shape=circle]

\tikzstyle{clear dot}=[dot,fill=none,text depth=-0.2mm,draw=gray, line width = .75pt]
\tikzstyle{tall clear dot}=[dot,fill=none,text depth=-0.2mm,draw=gray, line width = .75pt,shape=ellipse, minimum height=5mm]
\tikzstyle{wide clear dot}=[dot,fill=none,text depth=-0.2mm,draw=gray, line width = .75pt, shape=ellipse, minimum width = 5mm]
\tikzstyle{very wide clear dot}=[dot,fill=none,text depth=-0.2mm,draw=gray, line width = .75pt, shape=ellipse, minimum width = 7mm ]

\tikzstyle{black dot}=[dot,fill=black]
\tikzstyle{white dot}=[dot,fill=white,,text depth=-0.2mm]
\tikzstyle{white Wsquare}=[Wsquare,fill=gray,,text depth=-0.2mm]
\tikzstyle{white Wsquareadj}=[Wsquareadj,fill=white,,text depth=-0.2mm]
\tikzstyle{green dot}=[white dot] 
\tikzstyle{gray dot}=[dot,fill=gray!40!white,,text depth=-0.2mm]
\tikzstyle{red dot}=[gray dot] 

\tikzstyle{black ddot}=[ddot,fill=black]
\tikzstyle{white ddot}=[ddot,fill=white]
\tikzstyle{gray ddot}=[ddot,fill=gray!40!white]

\tikzstyle{gray edge}=[gray!60!white]

\tikzstyle{small dot}=[inner sep=0.2mm,minimum width=0pt,minimum height=0pt,draw,shape=circle]

\tikzstyle{small black dot}=[small dot,fill=black]
\tikzstyle{small white dot}=[small dot,fill=white]
\tikzstyle{small gray dot}=[small dot,fill=gray,draw=gray]

\tikzstyle{causal dot}=[inner sep=0.4mm,minimum width=0pt,minimum height=0pt,draw=white,shape=circle,fill=gray!40!white]

\tikzstyle{phase dimensions}=[minimum size=5mm,font=\footnotesize,rectangle,rounded corners=2.5mm,inner sep=0.2mm,outer sep=-2mm]
\tikzstyle{dphase dimensions}=[minimum size=5mm,font=\footnotesize,rectangle,rounded corners=2.5mm,inner sep=0.2mm,outer sep=-2mm]

\tikzstyle{white phase dot}=[dot,fill=white,phase dimensions]
\tikzstyle{white phase ddot}=[ddot,fill=white,dphase dimensions]

\tikzstyle{white rect ddot}=[draw=black,fill=white,doubled,minimum size=5mm,font=\footnotesize,rectangle,rounded corners=2.5mm,inner sep=0.2mm]
\tikzstyle{gray rect ddot}=[draw=black,fill=gray!40!white,doubled,minimum size=6mm,font=\footnotesize,rectangle,rounded corners=3mm]

\tikzstyle{gray phase dot}=[dot,fill=gray!40!white,phase dimensions]
\tikzstyle{gray phase ddot}=[ddot,fill=gray!40!white,dphase dimensions]
\tikzstyle{grey phase dot}=[gray phase dot]
\tikzstyle{grey phase ddot}=[gray phase ddot]

\tikzstyle{small phase dimensions}=[minimum size=4mm,font=\tiny,rectangle,rounded corners=2mm,inner sep=0.2mm,outer sep=-2mm]
\tikzstyle{small dphase dimensions}=[minimum size=4mm,font=\tiny,rectangle,rounded corners=2mm,inner sep=0.2mm,outer sep=-2mm]

\tikzstyle{small gray phase dot}=[dot,fill=gray!40!white,small phase dimensions]
\tikzstyle{small gray phase ddot}=[ddot,fill=gray!40!white,small dphase dimensions]

\tikzstyle{small map}=[draw,shape=rectangle,minimum height=4mm,minimum width=4mm,fill=white]

\tikzstyle{cnot}=[fill=white,shape=circle,inner sep=-1.4pt]

\tikzstyle{asym hadamard}=[fill=white,draw,shape=NEbox,inner sep=0.6mm,font=\footnotesize,minimum height=4mm]
\tikzstyle{asym hadamard conj}=[fill=white,draw,shape=NWbox,inner sep=0.6mm,font=\footnotesize,minimum height=4mm]
\tikzstyle{asym hadamard dag}=[fill=white,draw,shape=SEbox,inner sep=0.6mm,font=\footnotesize,minimum height=4mm]

\tikzstyle{hadamard}=[fill=white,draw,inner sep=0.6mm,font=\footnotesize,minimum height=4mm,minimum width=4mm]
\tikzstyle{small hadamard}=[fill=white,draw,inner sep=0.6mm,minimum height=1.5mm,minimum width=1.5mm]
\tikzstyle{small hadamard rotate}=[small hadamard,rotate=45]
\tikzstyle{dhadamard}=[hadamard,doubled]
\tikzstyle{small dhadamard}=[small hadamard,doubled]
\tikzstyle{small dhadamard rotate}=[small hadamard rotate,doubled]
\tikzstyle{antipode}=[white dot,inner sep=0.3mm,font=\footnotesize]

\tikzstyle{scalar}=[diamond,draw,inner sep=0.5pt,font=\small]
\tikzstyle{dscalar}=[diamond,doubled, draw,inner sep=0.5pt,font=\small]

\tikzstyle{small box}=[rectangle,inline text,fill=white,draw,minimum height=5mm,yshift=-0.5mm,minimum width=5mm,font=\small]
\tikzstyle{small gray box}=[small box,fill=gray!30]
\tikzstyle{medium box}=[rectangle,inline text,fill=white,draw,minimum height=5mm,yshift=-0.5mm,minimum width=10mm,font=\small]
\tikzstyle{square box}=[small box] 
\tikzstyle{medium gray box}=[small box,fill=gray!30]
\tikzstyle{semilarge box}=[rectangle,inline text,fill=white,draw,minimum height=5mm,yshift=-0.5mm,minimum width=12.5mm,font=\small]
\tikzstyle{large box}=[rectangle,inline text,fill=white,draw,minimum height=5mm,yshift=-0.5mm,minimum width=15mm,font=\small]
\tikzstyle{large gray box}=[small box,fill=gray!30]

\tikzstyle{Bayes box}=[rectangle,fill=black,draw, minimum height=3mm, minimum width=3mm]

\tikzstyle{gray square point}=[small box,fill=gray!50]

\tikzstyle{dphase box white}=[dhadamard]
\tikzstyle{dphase box gray}=[dhadamard,fill=gray!50!white]
\tikzstyle{phase box white}=[hadamard]
\tikzstyle{phase box gray}=[hadamard,fill=gray!50!white]

\tikzstyle{point}=[regular polygon,regular polygon sides=3,draw,scale=0.75,inner sep=-0.5pt,minimum width=9mm,fill=white,regular polygon rotate=180]
\tikzstyle{infpoint}=[regular polygon,regular polygon sides=3,draw,scale=0.75,inner sep=-0.5pt,minimum width=9mm,fill=white,regular polygon rotate=90]
\tikzstyle{point nosep}=[regular polygon,regular polygon sides=3,draw,scale=0.75,inner sep=-2pt,minimum width=9mm,fill=white,regular polygon rotate=180]
\tikzstyle{infcopoint}=[regular polygon,regular polygon sides=3,draw,scale=0.75,inner sep=-0.5pt,minimum width=9mm,fill=white,regular polygon rotate=270]
\tikzstyle{copoint}=[regular polygon,regular polygon sides=3,draw,scale=0.75,inner sep=-0.5pt,minimum width=9mm,fill=white]
\tikzstyle{dpoint}=[point,doubled]
\tikzstyle{dcopoint}=[copoint,doubled]

\tikzstyle{pointgrow}=[shape=cornerpoint,kpoint common,scale=0.75,inner sep=3pt]
\tikzstyle{pointgrow dag}=[shape=cornercopoint,kpoint common,scale=0.75,inner sep=3pt]

\tikzstyle{wide copoint}=[fill=white,draw,shape=isosceles triangle,shape border rotate=90,isosceles triangle stretches=true,inner sep=0pt,minimum width=1.5cm,minimum height=6.12mm]
\tikzstyle{wide point}=[fill=white,draw,shape=isosceles triangle,shape border rotate=-90,isosceles triangle stretches=true,inner sep=0pt,minimum width=1.5cm,minimum height=6.12mm,yshift=-0.0mm]
\tikzstyle{wide point plus}=[fill=white,draw,shape=isosceles triangle,shape border rotate=-90,isosceles triangle stretches=true,inner sep=0pt,minimum width=1.74cm,minimum height=7mm,yshift=-0.0mm]

\tikzstyle{wide dpoint}=[fill=white,doubled,draw,shape=isosceles triangle,shape border rotate=-90,isosceles triangle stretches=true,inner sep=0pt,minimum width=1.5cm,minimum height=6.12mm,yshift=-0.0mm]

\tikzstyle{tinypoint}=[regular polygon,regular polygon sides=3,draw,scale=0.55,inner sep=-0.15pt,minimum width=6mm,fill=white,regular polygon rotate=180]

\tikzstyle{white point}=[point]
\tikzstyle{white dpoint}=[dpoint]
\tikzstyle{green point}=[white point] 
\tikzstyle{white copoint}=[copoint]
\tikzstyle{gray point}=[point,fill=gray!40!white]
\tikzstyle{gray dpoint}=[gray point,doubled]
\tikzstyle{red point}=[gray point] 
\tikzstyle{gray copoint}=[copoint,fill=gray!40!white]
\tikzstyle{gray dcopoint}=[gray copoint,doubled]

\tikzstyle{white point guide}=[regular polygon,regular polygon sides=3,font=\scriptsize,draw,scale=0.65,inner sep=-0.5pt,minimum width=9mm,fill=white,regular polygon rotate=180]

\tikzstyle{black point}=[point,fill=black,font=\color{white}]
\tikzstyle{black copoint}=[copoint,fill=black,font=\color{white}]

\tikzstyle{tiny gray point}=[tinypoint,fill=gray!40!white]

\tikzstyle{diredge}=[->]
\tikzstyle{ddiredge}=[<->]
\tikzstyle{rdiredge}=[<-]
\tikzstyle{thickdiredge}=[->, very thick]
\tikzstyle{pointer edge}=[->,very thick,gray]
\tikzstyle{pointer edge part}=[very thick,gray]
\tikzstyle{dashed edge}=[dashed]
\tikzstyle{thick dashed edge}=[very thick,dashed]
\tikzstyle{thick gray dashed edge}=[thick dashed edge,gray!40]
\tikzstyle{thick map edge}=[very thick,|->]

\makeatletter
\newcommand{\boxshape}[3]{%
\pgfdeclareshape{#1}{
\inheritsavedanchors[from=rectangle] 
\inheritanchorborder[from=rectangle]
\inheritanchor[from=rectangle]{center}
\inheritanchor[from=rectangle]{north}
\inheritanchor[from=rectangle]{south}
\inheritanchor[from=rectangle]{west}
\inheritanchor[from=rectangle]{east}
\backgroundpath{
\southwest \pgf@xa=\pgf@x \pgf@ya=\pgf@y
\northeast \pgf@xb=\pgf@x \pgf@yb=\pgf@y

\@tempdima=#2
\@tempdimb=#3

\pgfpathmoveto{\pgfpoint{\pgf@xa - 5pt + \@tempdima}{\pgf@ya}}
\pgfpathlineto{\pgfpoint{\pgf@xa - 5pt - \@tempdima}{\pgf@yb}}
\pgfpathlineto{\pgfpoint{\pgf@xb + 5pt + \@tempdimb}{\pgf@yb}}
\pgfpathlineto{\pgfpoint{\pgf@xb + 5pt - \@tempdimb}{\pgf@ya}}
\pgfpathlineto{\pgfpoint{\pgf@xa - 5pt + \@tempdima}{\pgf@ya}}
\pgfpathclose
}
}}

\boxshape{NEbox}{0pt}{3pt}
\boxshape{SEbox}{0pt}{-3pt}
\boxshape{NWbox}{5pt}{0pt}
\boxshape{SWbox}{-5pt}{0pt}
\boxshape{EBox}{-3pt}{3pt}
\boxshape{WBox}{3pt}{-3pt}
\makeatother

\tikzstyle{cloud}=[shape=cloud,draw,minimum width=1.5cm,minimum height=1.5cm]

\tikzstyle{map}=[draw,shape=NEbox,inner sep=1pt,minimum height=4mm,fill=white]
\tikzstyle{dashedmap}=[draw,dashed,shape=NEbox,inner sep=2pt,minimum height=6mm,fill=white]
\tikzstyle{mapdag}=[draw,shape=SEbox,inner sep=1pt,minimum height=4mm,fill=white]
\tikzstyle{mapadj}=[draw,shape=SEbox,inner sep=2pt,minimum height=6mm,fill=white]
\tikzstyle{maptrans}=[draw,shape=SWbox,inner sep=2pt,minimum height=6mm,fill=white]
\tikzstyle{mapconj}=[draw,shape=NWbox,inner sep=2pt,minimum height=6mm,fill=white]

\tikzstyle{medium map}=[draw,shape=NEbox,inner sep=2pt,minimum height=6mm,fill=white,minimum width=7mm]
\tikzstyle{medium map dag}=[draw,shape=SEbox,inner sep=2pt,minimum height=6mm,fill=white,minimum width=7mm]
\tikzstyle{medium map adj}=[draw,shape=SEbox,inner sep=2pt,minimum height=6mm,fill=white,minimum width=7mm]
\tikzstyle{medium map trans}=[draw,shape=SWbox,inner sep=2pt,minimum height=6mm,fill=white,minimum width=7mm]
\tikzstyle{medium map conj}=[draw,shape=NWbox,inner sep=2pt,minimum height=6mm,fill=white,minimum width=7mm]
\tikzstyle{semilarge map}=[draw,shape=NEbox,inner sep=2pt,minimum height=6mm,fill=white,minimum width=9.5mm]
\tikzstyle{semilarge map trans}=[draw,shape=SWbox,inner sep=2pt,minimum height=6mm,fill=white,minimum width=9.5mm]
\tikzstyle{semilarge map adj}=[draw,shape=SEbox,inner sep=2pt,minimum height=6mm,fill=white,minimum width=9.5mm]
\tikzstyle{semilarge map dag}=[draw,shape=SEbox,inner sep=2pt,minimum height=6mm,fill=white,minimum width=9.5mm]
\tikzstyle{semilarge map conj}=[draw,shape=NWbox,inner sep=2pt,minimum height=6mm,fill=white,minimum width=9.5mm]
\tikzstyle{large map}=[draw,shape=NEbox,inner sep=2pt,minimum height=6mm,fill=white,minimum width=12mm]
\tikzstyle{large map conj}=[draw,shape=NWbox,inner sep=2pt,minimum height=6mm,fill=white,minimum width=12mm]
\tikzstyle{very large map}=[draw,shape=NEbox,inner sep=2pt,minimum height=6mm,fill=white,minimum width=17mm]

\tikzstyle{medium dmap}=[draw,doubled,shape=NEbox,inner sep=2pt,minimum height=6mm,fill=white,minimum width=7mm]
\tikzstyle{medium dmap dag}=[draw,doubled,shape=SEbox,inner sep=2pt,minimum height=6mm,fill=white,minimum width=7mm]
\tikzstyle{medium dmap adj}=[draw,doubled,shape=SEbox,inner sep=2pt,minimum height=6mm,fill=white,minimum width=7mm]
\tikzstyle{medium dmap trans}=[draw,doubled,shape=SWbox,inner sep=2pt,minimum height=6mm,fill=white,minimum width=7mm]
\tikzstyle{medium dmap conj}=[draw,doubled,shape=NWbox,inner sep=2pt,minimum height=6mm,fill=white,minimum width=7mm]
\tikzstyle{semilarge dmap}=[draw,doubled,shape=NEbox,inner sep=2pt,minimum height=6mm,fill=white,minimum width=9.5mm]
\tikzstyle{semilarge dmap trans}=[draw,doubled,shape=SWbox,inner sep=2pt,minimum height=6mm,fill=white,minimum width=9.5mm]
\tikzstyle{semilarge dmap adj}=[draw,doubled,shape=SEbox,inner sep=2pt,minimum height=6mm,fill=white,minimum width=9.5mm]
\tikzstyle{semilarge dmap dag}=[draw,doubled,shape=SEbox,inner sep=2pt,minimum height=6mm,fill=white,minimum width=9.5mm]
\tikzstyle{semilarge dmap conj}=[draw,doubled,shape=NWbox,inner sep=2pt,minimum height=6mm,fill=white,minimum width=9.5mm]
\tikzstyle{large dmap}=[draw,doubled,shape=NEbox,inner sep=2pt,minimum height=6mm,fill=white,minimum width=12mm]
\tikzstyle{large dmap conj}=[draw,doubled,shape=NWbox,inner sep=2pt,minimum height=6mm,fill=white,minimum width=12mm]
\tikzstyle{large dmap trans}=[draw,doubled,shape=SWbox,inner sep=2pt,minimum height=6mm,fill=white,minimum width=12mm]
\tikzstyle{large dmap adj}=[draw,doubled,shape=SEbox,inner sep=2pt,minimum height=6mm,fill=white,minimum width=12mm]
\tikzstyle{large dmap dag}=[draw,doubled,shape=SEbox,inner sep=2pt,minimum height=6mm,fill=white,minimum width=12mm]
\tikzstyle{very large dmap}=[draw,doubled,shape=NEbox,inner sep=2pt,minimum height=6mm,fill=white,minimum width=19.5mm]

\tikzstyle{muxbox}=[draw,shape=rectangle,minimum height=3mm,minimum width=3mm,fill=white]
\tikzstyle{dmuxbox}=[muxbox,doubled]

\tikzstyle{box}=[draw,shape=rectangle,inner sep=2pt,minimum height=6mm,minimum width=6mm,fill=white]
\tikzstyle{dbox}=[draw,doubled,shape=rectangle,inner sep=2pt,minimum height=6mm,minimum width=6mm,fill=white]
\tikzstyle{dmap}=[draw,doubled,shape=NEbox,inner sep=2pt,minimum height=6mm,fill=white]
\tikzstyle{dmapdag}=[draw,doubled,shape=SEbox,inner sep=2pt,minimum height=6mm,fill=white]
\tikzstyle{dmapadj}=[draw,doubled,shape=SEbox,inner sep=2pt,minimum height=6mm,fill=white]
\tikzstyle{dmaptrans}=[draw,doubled,shape=SWbox,inner sep=2pt,minimum height=6mm,fill=white]
\tikzstyle{dmapconj}=[draw,doubled,shape=NWbox,inner sep=2pt,minimum height=6mm,fill=white]

\tikzstyle{ddmap}=[draw,doubled,dashed,shape=NEbox,inner sep=2pt,minimum height=6mm,fill=white]
\tikzstyle{ddmapdag}=[draw,doubled,dashed,shape=SEbox,inner sep=2pt,minimum height=6mm,fill=white]
\tikzstyle{ddmapadj}=[draw,doubled,dashed,shape=SEbox,inner sep=2pt,minimum height=6mm,fill=white]
\tikzstyle{ddmaptrans}=[draw,doubled,dashed,shape=SWbox,inner sep=2pt,minimum height=6mm,fill=white]
\tikzstyle{ddmapconj}=[draw,doubled,dashed,shape=NWbox,inner sep=2pt,minimum height=6mm,fill=white]

\boxshape{sNEbox}{0pt}{3pt}
\boxshape{sSEbox}{0pt}{-3pt}
\boxshape{sNWbox}{3pt}{0pt}
\boxshape{sSWbox}{-3pt}{0pt}
\tikzstyle{smap}=[draw,shape=sNEbox,fill=white]
\tikzstyle{smapdag}=[draw,shape=sSEbox,fill=white]
\tikzstyle{smapadj}=[draw,shape=sSEbox,fill=white]
\tikzstyle{smaptrans}=[draw,shape=sSWbox,fill=white]
\tikzstyle{smapconj}=[draw,shape=sNWbox,fill=white]

\tikzstyle{dsmap}=[draw,dashed,shape=sNEbox,fill=white]
\tikzstyle{dsmapdag}=[draw,dashed,shape=sSEbox,fill=white]
\tikzstyle{dsmaptrans}=[draw,dashed,shape=sSWbox,fill=white]
\tikzstyle{dsmapconj}=[draw,dashed,shape=sNWbox,fill=white]

\boxshape{mNEbox}{0pt}{10pt}
\boxshape{mSEbox}{0pt}{-10pt}
\boxshape{mNWbox}{10pt}{0pt}
\boxshape{mSWbox}{-10pt}{0pt}
\tikzstyle{mmap}=[draw,shape=mNEbox]
\tikzstyle{mmapdag}=[draw,shape=mSEbox]
\tikzstyle{mmaptrans}=[draw,shape=mSWbox]
\tikzstyle{mmapconj}=[draw,shape=mNWbox]

\tikzstyle{mmapgray}=[draw,fill=gray!40!white,shape=mNEbox]
\tikzstyle{smapgray}=[draw,fill=gray!40!white,shape=sNEbox]

\makeatletter

\pgfdeclareshape{cornerpoint}{
\inheritsavedanchors[from=rectangle] 
\inheritanchorborder[from=rectangle]
\inheritanchor[from=rectangle]{center}
\inheritanchor[from=rectangle]{north}
\inheritanchor[from=rectangle]{south}
\inheritanchor[from=rectangle]{west}
\inheritanchor[from=rectangle]{east}
\backgroundpath{
\southwest \pgf@xa=\pgf@x \pgf@ya=\pgf@y
\northeast \pgf@xb=\pgf@x \pgf@yb=\pgf@y

\pgfmathsetmacro{\pgf@shorten@left}{\pgfkeysvalueof{/tikz/shorten left}}
\pgfmathsetmacro{\pgf@shorten@right}{\pgfkeysvalueof{/tikz/shorten right}}

\pgfpathmoveto{\pgfpoint{0.5 * (\pgf@xa + \pgf@xb)}{\pgf@ya - 5pt}}
\pgfpathlineto{\pgfpoint{\pgf@xa - 8pt + \pgf@shorten@left}{\pgf@yb - 1.5 * \pgf@shorten@left}}
\pgfpathlineto{\pgfpoint{\pgf@xa - 8pt + \pgf@shorten@left}{\pgf@yb}}
\pgfpathlineto{\pgfpoint{\pgf@xb + 8pt - \pgf@shorten@right}{\pgf@yb}}
\pgfpathlineto{\pgfpoint{\pgf@xb + 8pt - \pgf@shorten@right}{\pgf@yb - 1.5 * \pgf@shorten@right}}
\pgfpathclose
}
}

\pgfdeclareshape{cornercopoint}{
\inheritsavedanchors[from=rectangle] 
\inheritanchorborder[from=rectangle]
\inheritanchor[from=rectangle]{center}
\inheritanchor[from=rectangle]{north}
\inheritanchor[from=rectangle]{south}
\inheritanchor[from=rectangle]{west}
\inheritanchor[from=rectangle]{east}
\backgroundpath{
\southwest \pgf@xa=\pgf@x \pgf@ya=\pgf@y
\northeast \pgf@xb=\pgf@x \pgf@yb=\pgf@y

\pgfmathsetmacro{\pgf@shorten@left}{\pgfkeysvalueof{/tikz/shorten left}}
\pgfmathsetmacro{\pgf@shorten@right}{\pgfkeysvalueof{/tikz/shorten right}}

\pgfpathmoveto{\pgfpoint{0.5 * (\pgf@xa + \pgf@xb)}{\pgf@yb + 5pt}}
\pgfpathlineto{\pgfpoint{\pgf@xa - 8pt + \pgf@shorten@left}{\pgf@ya + 1.5 * \pgf@shorten@left}}
\pgfpathlineto{\pgfpoint{\pgf@xa - 8pt + \pgf@shorten@left}{\pgf@ya}}
\pgfpathlineto{\pgfpoint{\pgf@xb + 8pt - \pgf@shorten@right}{\pgf@ya}}
\pgfpathlineto{\pgfpoint{\pgf@xb + 8pt - \pgf@shorten@right}{\pgf@ya + 1.5 * \pgf@shorten@right}}
\pgfpathclose
}
}

\makeatother

\pgfkeyssetvalue{/tikz/shorten left}{0pt}
\pgfkeyssetvalue{/tikz/shorten right}{0pt}

\tikzstyle{kpoint common}=[draw,fill=white,inner sep=1pt,minimum height=4mm]
\tikzstyle{kpoint sc}=[shape=cornerpoint,kpoint common]
\tikzstyle{kpoint adjoint sc}=[shape=cornercopoint,kpoint common]
\tikzstyle{kpoint}=[shape=cornerpoint,shorten left=5pt,kpoint common]
\tikzstyle{kpoint adjoint}=[shape=cornercopoint,shorten left=5pt,kpoint common]
\tikzstyle{kpoint conjugate}=[shape=cornerpoint,shorten right=5pt,kpoint common]
\tikzstyle{kpoint transpose}=[shape=cornercopoint,shorten right=5pt,kpoint common]
\tikzstyle{kpoint symm}=[shape=cornerpoint,shorten left=5pt,shorten right=5pt,kpoint common]

\tikzstyle{wide kpoint sc}=[shape=cornerpoint,kpoint common, minimum width=1 cm]
\tikzstyle{wide kpointdag sc}=[shape=cornercopoint,kpoint common, minimum width=1 cm]

\tikzstyle{black kpoint}=[shape=cornerpoint,shorten left=5pt,kpoint common,fill=black,font=\color{white}]

\tikzstyle{black kpoint sm}=[shape=cornerpoint,shorten left=5pt,kpoint common,fill=black,font=\color{white},scale=0.75]

\tikzstyle{black kpoint adjoint}=[shape=cornercopoint,shorten left=5pt,kpoint common,fill=black,font=\color{white}]
\tikzstyle{black kpointadj}=[shape=cornercopoint,shorten left=5pt,kpoint common,fill=black,font=\color{white}]

\tikzstyle{black kpointadj sm}=[shape=cornercopoint,shorten left=5pt,kpoint common,fill=black,font=\color{white},scale=0.75]

\tikzstyle{black dkpoint}=[shape=cornerpoint,shorten left=5pt,kpoint common,fill=black, doubled,font=\color{white}]
\tikzstyle{black dkpoint adjoint}=[shape=cornercopoint,shorten left=5pt,kpoint common,fill=black, doubled,font=\color{white}]
\tikzstyle{black dkpointadj}=[shape=cornercopoint,shorten left=5pt,kpoint common,fill=black, doubled,font=\color{white}]

\tikzstyle{black dkpoint sm}=[shape=cornerpoint,shorten left=5pt,kpoint common,fill=black, doubled,font=\color{white},scale=0.75]
\tikzstyle{black dkpointadj sm}=[shape=cornercopoint,shorten left=5pt,kpoint common,fill=black, doubled,font=\color{white},scale=0.75]

\tikzstyle{kpointdag}=[kpoint adjoint]
\tikzstyle{kpointadj}=[kpoint adjoint]
\tikzstyle{kpointconj}=[kpoint conjugate]
\tikzstyle{kpointtrans}=[kpoint transpose]

\tikzstyle{big kpoint}=[kpoint, minimum width=1.2 cm, minimum height=8mm, inner sep=4pt, text depth=3mm]

\tikzstyle{wide kpoint}=[kpoint, minimum width=1 cm, inner sep=2pt]
\tikzstyle{wide kpointdag}=[kpointdag, minimum width=1 cm, inner sep=2pt]
\tikzstyle{wide kpointconj}=[kpointconj, minimum width=1 cm, inner sep=2pt]
\tikzstyle{wide kpointtrans}=[kpointtrans, minimum width=1 cm, inner sep=2pt]

\tikzstyle{wider kpoint}=[kpoint, minimum width=1.25 cm, inner sep=2pt]
\tikzstyle{wider kpointdag}=[kpointdag, minimum width=1.25 cm, inner sep=2pt]
\tikzstyle{wider kpointconj}=[kpointconj, minimum width=1.25 cm, inner sep=2pt]
\tikzstyle{wider kpointtrans}=[kpointtrans, minimum width=1.25 cm, inner sep=2pt]

\tikzstyle{gray kpoint}=[kpoint,fill=gray!50!white]
\tikzstyle{gray kpointdag}=[kpointdag,fill=gray!50!white]
\tikzstyle{gray kpointadj}=[kpointadj,fill=gray!50!white]
\tikzstyle{gray kpointconj}=[kpointconj,fill=gray!50!white]
\tikzstyle{gray kpointtrans}=[kpointtrans,fill=gray!50!white]

\tikzstyle{gray dkpoint}=[kpoint,fill=gray!50!white,doubled]
\tikzstyle{gray dkpointdag}=[kpointdag,fill=gray!50!white,doubled]
\tikzstyle{gray dkpointadj}=[kpointadj,fill=gray!50!white,doubled]
\tikzstyle{gray dkpointconj}=[kpointconj,fill=gray!50!white,doubled]
\tikzstyle{gray dkpointtrans}=[kpointtrans,fill=gray!50!white,doubled]

\tikzstyle{white label}=[draw,fill=white,rectangle,inner sep=0.7 mm]
\tikzstyle{gray label}=[draw,fill=gray!50!white,rectangle,inner sep=0.7 mm]
\tikzstyle{black label}=[draw,fill=black,rectangle,inner sep=0.7 mm]

\tikzstyle{dkpoint}=[kpoint,doubled]
\tikzstyle{wide dkpoint}=[wide kpoint,doubled]
\tikzstyle{dkpointdag}=[kpoint adjoint,doubled]
\tikzstyle{wide dkpointdag}=[wide kpointdag,doubled]
\tikzstyle{dkcopoint}=[kpoint adjoint,doubled]
\tikzstyle{dkpointadj}=[kpoint adjoint,doubled]
\tikzstyle{dkpointconj}=[kpoint conjugate,doubled]
\tikzstyle{dkpointtrans}=[kpoint transpose,doubled]

\tikzstyle{kscalar}=[kpoint common, shape=EBox, inner xsep=-1pt, inner ysep=3pt,font=\small]
\tikzstyle{kscalarconj}=[kpoint common, shape=WBox, inner xsep=-1pt, inner ysep=3pt,font=\small]

\tikzstyle{spekpoint}=[kpoint sc,minimum height=5mm,inner sep=3pt]
\tikzstyle{spekcopoint}=[kpoint adjoint sc,minimum height=5mm,inner sep=3pt]

\tikzstyle{dspekpoint}=[spekpoint,doubled]
\tikzstyle{dspekcopoint}=[spekcopoint,doubled]

 \tikzstyle{upground}=[circuit ee IEC,thick,ground,rotate=90,scale=2.5]
 \tikzstyle{downground}=[circuit ee IEC,thick,ground,rotate=-90,scale=2.5]
 \tikzstyle{infupground}=[circuit ee IEC,thick,ground,rotate=0,scale=2.5]
 \tikzstyle{infdownground}=[circuit ee IEC,thick,ground,rotate=180,scale=2.5]
 \tikzstyle{bigground}=[regular polygon,regular polygon sides=3,draw=gray,scale=0.50,inner sep=-0.5pt,minimum width=10mm,fill=gray]

\tikzstyle{arrs}=[-latex,font=\small,auto]
\tikzstyle{arrow plain}=[arrs]
\tikzstyle{arrow dashed}=[dashed,arrs]
\tikzstyle{arrow bold}=[very thick,arrs]
\tikzstyle{arrow hide}=[draw=white!0,-]
\tikzstyle{arrow reverse}=[latex-]
\tikzstyle{cdnode}=[]

\tikzstyle{tilde}=[draw=blue]
\tikzstyle{tildelabel}=[text=blue]

\newkeycommand{\pointmap}[style=point][1]{\,\tikz{\node[style=\commandkey{style}] (x) at (0,-0.3) {$#1$};\node [style=none] (3) at (0, 1.0) {};\node [style=none] (3) at (0, -1.0) {};\draw (x)--(0,0.7);}\,}

\newkeycommand{\typedkpoint}[style=kpoint,edgestyle=][2]{%
\,\begin{tikzpicture}
  \begin{pgfonlayer}{nodelayer}
    \node [style=none] (0) at (0, 1) {};
    \node [style=\commandkey{style}] (1) at (0, -0.75) {$#2$};
    \node [style=right label] (2) at (0.25, 0.5) {$#1$};
  \end{pgfonlayer}
  \begin{pgfonlayer}{edgelayer}
    \draw [\commandkey{edgestyle}] (1) to (0.center);
  \end{pgfonlayer}
\end{tikzpicture}\,}

\newkeycommand{\typedkpointdag}[style=kpoint adjoint,edgestyle=][2]{%
\,\begin{tikzpicture}
  \begin{pgfonlayer}{nodelayer}
    \node [style=none] (0) at (0, -1) {};
    \node [style=\commandkey{style}] (1) at (0, 0.75) {$#2$};
    \node [style=right label] (2) at (0.25, -0.5) {$#1$};
  \end{pgfonlayer}
  \begin{pgfonlayer}{edgelayer}
    \draw [\commandkey{edgestyle}] (0.center) to (1);
  \end{pgfonlayer}
\end{tikzpicture}\,}

\newcommand{\bistateadj}[1]{\,\begin{tikzpicture}[yshift=-3mm]
  \begin{pgfonlayer}{nodelayer}
    \node [style=kpointadj, minimum width=1 cm, inner sep=2pt] (0) at (0, 0.5) {$\psi$};
    \node [style=none] (1) at (-0.75, 0.25) {};
    \node [style=none] (2) at (0.75, 0.25) {};
    \node [style=none] (3) at (-0.75, -0.5) {};
    \node [style=none] (4) at (0.75, -0.5) {};
  \end{pgfonlayer}
  \begin{pgfonlayer}{edgelayer}
    \draw (1.center) to (3.center);
    \draw (2.center) to (4.center);
  \end{pgfonlayer}
\end{tikzpicture}\,}

\newkeycommand{\pointketbra}[style1=copoint,style2=point][2]{\,%
\begin{tikzpicture}
  \begin{pgfonlayer}{nodelayer}
    \node [style=none] (0) at (0, -1.5) {};
    \node [style=\commandkey{style1}] (1) at (0, -0.75) {$#1$};
    \node [style=\commandkey{style2}] (2) at (0, 0.75) {$#2$};
    \node [style=none] (3) at (0, 1.5) {};
  \end{pgfonlayer}
  \begin{pgfonlayer}{edgelayer}
    \draw (0.center) to (1);
    \draw (2) to (3.center);
  \end{pgfonlayer}
\end{tikzpicture}\,}

\newkeycommand{\twocopointketbra}[style1=copoint,style2=copoint,style3=point][3]{\,%
\begin{tikzpicture}
  \begin{pgfonlayer}{nodelayer}
    \node [style=none] (0) at (-0.75, -1.5) {};
    \node [style=none] (1) at (0.75, -1.5) {};
    \node [style=\commandkey{style1}] (2) at (-0.75, -0.75) {$#1$};
    \node [style=\commandkey{style2}] (3) at (0.75, -0.75) {$#2$};
    \node [style=\commandkey{style3}] (4) at (0, 0.75) {$#3$};
    \node [style=none] (5) at (0, 1.5) {};
  \end{pgfonlayer}
  \begin{pgfonlayer}{edgelayer}
    \draw (0.center) to (2);
    \draw (1.center) to (3);
    \draw (4) to (5.center);
  \end{pgfonlayer}
\end{tikzpicture}\,}

\newkeycommand{\twopointketbra}[style1=copoint,style2=point,style3=point][3]{\,%
\begin{tikzpicture}
  \begin{pgfonlayer}{nodelayer}
    \node [style=none] (0) at (0, -1.5) {};
    \node [style=none] (1) at (-0.75, 1.5) {};
    \node [style=\commandkey{style1}] (2) at (0, -0.75) {$#1$};
    \node [style=\commandkey{style2}] (3) at (-0.75, 0.75) {$#2$};
    \node [style=\commandkey{style3}] (4) at (0.75, 0.75) {$#3$};
    \node [style=none] (5) at (0.75, 1.5) {};
  \end{pgfonlayer}
  \begin{pgfonlayer}{edgelayer}
    \draw (0.center) to (2);
    \draw (1.center) to (3);
    \draw (4) to (5.center);
  \end{pgfonlayer}
\end{tikzpicture}\,}

\newkeycommand{\pointbraket}[style1=point,style2=copoint][2]{\,%
\begin{tikzpicture}
  \begin{pgfonlayer}{nodelayer}
    \node [style=\commandkey{style1}] (0) at (0, -0.5) {$#1$};
    \node [style=\commandkey{style2}] (1) at (0, 0.5) {$#2$};
  \end{pgfonlayer}
  \begin{pgfonlayer}{edgelayer}
    \draw (0) to (1);
  \end{pgfonlayer}
\end{tikzpicture}\,}

\newkeycommand{\kpointbraket}[style1=kpoint,style2=kpoint adjoint][2]{\,%
\begin{tikzpicture}
  \begin{pgfonlayer}{nodelayer}
    \node [style=\commandkey{style1}] (0) at (0, -0.75) {$#1$};
    \node [style=\commandkey{style2}] (1) at (0, 0.75) {$#2$};
  \end{pgfonlayer}
  \begin{pgfonlayer}{edgelayer}
    \draw (0) to (1);
  \end{pgfonlayer}
\end{tikzpicture}\,}

\newkeycommand{\kpointmap}[style1=kpoint,style2=map][2]{\,%
\begin{tikzpicture}
  \begin{pgfonlayer}{nodelayer}
    \node [style=\commandkey{style1}] (0) at (0, -1.1) {$#1$};
    \node [style=\commandkey{style2}] (1) at (0, 0.2) {$#2$};
    \node [style=none] (2) at (0, 1.5) {};
  \end{pgfonlayer}
  \begin{pgfonlayer}{edgelayer}
    \draw (0) to (1);
    \draw (1) to (2.center);
  \end{pgfonlayer}
\end{tikzpicture}%
\,}

\newkeycommand{\sandwichtwo}[style1=point,style2=map,style3=map,style4=copoint][4]{\,%
\begin{tikzpicture}
  \begin{pgfonlayer}{nodelayer}
    \node [style=\commandkey{style2}] (0) at (0, -0.75) {$#2$};
    \node [style=\commandkey{style3}] (1) at (0, 0.75) {$#3$};
    \node [style=\commandkey{style1}] (2) at (0, -2) {$#1$};
    \node [style=\commandkey{style4}] (3) at (0, 2) {$#4$};
  \end{pgfonlayer}
  \begin{pgfonlayer}{edgelayer}
    \draw (2) to (0);
    \draw (0) to (1);
    \draw (1) to (3);
  \end{pgfonlayer}
\end{tikzpicture}\,}

\newkeycommand{\longbraket}[style1=point,style2=copoint][2]{\,%
\begin{tikzpicture}
  \begin{pgfonlayer}{nodelayer}
    \node [style=\commandkey{style1}] (0) at (0, -2) {$#1$};
    \node [style=\commandkey{style2}] (1) at (0, 2) {$#2$};
  \end{pgfonlayer}
  \begin{pgfonlayer}{edgelayer}
    \draw (0) to (1);
  \end{pgfonlayer}
\end{tikzpicture}\,}

\newkeycommand{\onb}[style=point]{\ensuremath{\left\{\tikz{\node[style=\commandkey{style}] (x) at (0,-0.3) {$j$};\draw (x)--(0,0.7);}\right\}}\xspace}

\newkeycommand{\copointmap}[style=copoint][1]{\,\tikz{\node[style=\commandkey{style}] (x) at (0,0.3) {$#1$};\draw (x)--(0,-0.7);}\,}

\theoremstyle{plain}
\newtheorem*{main theorem}{Main Theorem}
\newtheorem{theorem}{Theorem}[section]
\newtheorem{corollary}[theorem]{Corollary}
\newtheorem{lemma}[theorem]{Lemma}

\newtheorem{example*}[theorem]{Example*}
\newtheorem{examples*}[theorem]{Examples*}

\newtheorem{remark*}[theorem]{Remark*}

\newtheorem*{search problem}{Search Problem}

\theoremstyle{definition}
\newtheorem{definition}[theorem]{Definition}
\newtheorem{examplealt}[theorem]{Example}
\newenvironment{examplep}[1]{
  \renewcommand\theexamplealt{#1}
  \examplealt
}{\endexamplealt}
\newtheorem{example}[theorem]{Example}

\usepackage{thmtools} 
\usepackage{thm-restate}

\newcommand{\nocontentsline}[3]{}
\let\oldaddcontentsline\addcontentsline
\newcommand{\tocless}[2]{%
  \let\addcontentsline=\nocontentsline#1{#2}
  \let\addcontentsline\oldaddcontentsline}
  
\newcommand\FinQT{\mathbf{FinQT}}

\newcommand\R{\mathds{R}}

\newcommand{\realvectexamplep}{\hyperlink{real_vector_examples}{\color{darkgray}{{\mathbf{Vect}_\mathds{R}}}}}

\newcommand{\complexvectexamplep}{\hyperlink{complex_vector_examples}{\color{darkgray}{{\mathbf{Vect}_\mathds{C}}}}}

\newcommand{\fincomplexvectexamplep}{\hyperlink{finite_complex_vector_examples}{\color{darkgray}{{\mathbf{FinVect}_\mathds{C}}}}}

\newcommand{\finrealvectexamplep}{\hyperlink{finite_real_vector_examples}{\color{darkgray}{{\mathbf{FinVect}_\mathds{R}}}}}

\newcommand{\finrealquasisubstoch}{\hyperlink{finite_real_vector_quasisubstoch}{\color{darkgray}{{\mathbf{FinQuasiSubStoch}_\mathds{R}}}}}

\newcommand{\fincomplexquasisubstoch}{\hyperlink{finite_complex_vector_quasisubstoch}{\color{darkgray}{{\mathbf{FinQuasiSubStoch}_\mathds{C}}}}}

\newcommand{\finiterealsubstoch}{\hyperlink{finite_real_vector_substoch}{\color{darkgray}{{\mathbf{FinSubStoch}_\mathds{R}}}}}

\newcommand{\finitequantumtheory}{\hyperlink{finite_quantum_theory}{\color{darkgray}{{\mathbf{FinQT}}}}}

\makeatletter
\newcommand\mathcolorbox[2]{%
  \begingroup
    \setbox\z@\hbox{\(\m@th#2\)}
    \colorbox{#1}{\raisebox{0pt}[\ht\z@][\dp\z@]{\copy\z@}}
  \endgroup
}
\makeatother

\definecolor{complex}{RGB}{242, 204, 217}

\tikzstyle{point}=[regular polygon, regular polygon sides=3, draw, scale=0.75, inner sep=-0.5pt, minimum width=9mm, fill=white, regular polygon rotate=180]
\tikzstyle{copoint}=[regular polygon, regular polygon sides=3, draw, scale=0.75, inner sep=-0.5pt, minimum width=9mm, fill=white]
\tikzstyle{copoint_complex}=[regular polygon, regular polygon sides=3, draw, scale=0.75, inner sep=-0.5pt, minimum width=9mm, fill=complex]
\tikzstyle{processes}=[fill=white, draw=black, shape=rectangle]
\tikzstyle{NodeSystemsToComplexVec}=[fill={rgb,255: red,204; green,204; blue,255}, draw=black, shape=rectangle]
\tikzstyle{NodesSystemsToComplex}=[fill={rgb,255: red,242; green,204; blue,217}, draw=black, shape=rectangle]
\tikzstyle{NodesSystemsToVec}=[fill={rgb,255: red,204; green,255; blue,204}, draw=black, shape=rectangle]

\tikzstyle{FVecMap}=[-, fill={rgb,255: red,204; green,255; blue,204}, draw={rgb,255: red,153; green,255; blue,153}]
\tikzstyle{FVecMapDarkOutside}=[-, fill={rgb,255: red,204; green,255; blue,204}, draw=black]
\tikzstyle{ComplexMap}=[-, fill={rgb,255: red,242; green,204; blue,217}, draw={rgb,255: red,230; green,153; blue,178}]
\tikzstyle{ComplexMapDark}=[-, fill={rgb,255: red,242; green,204; blue,217}]
\tikzstyle{FVecComplexDark}=[-, fill={rgb,255: red,204; green,204; blue,255}]
\tikzstyle{FVecComplex}=[-, fill={rgb,255: red,204; green,204; blue,255}, draw={rgb,255: red,173; green,173; blue,255}]
\tikzstyle{HardWire}=[-, line width=1, color=black]
\tikzstyle{white}=[-, fill=white]
\tikzstyle{FVecComplexWire}=[-, draw={rgb,255: red,173; green,173; blue,255}]

\date{\today}

\pgfplotsset{compat=1.18}
\begin{document}

\title{A structure theorem for complex-valued quasiprobability \\ representations of physical theories} 

\author{Rafael Wagner}
\email{rafael.wagner@uni-ulm.de}
\affiliation{\ulmshort}
\author{Roberto D. Baldijão}
\affiliation{\gdanskshort}
\affiliation{\perimetershort}
\author{Matthias Salzger}
\affiliation{\gdanskshort}
\author{Yìlè Y{\=\i}ng}
\affiliation{\perimetershort}
\affiliation{\waterlooshort}
\author{David Schmid}
\affiliation{\perimetershort}
\author{John H. Selby}
\affiliation{\gdanskshort}

\begin{abstract}
    Quasiprobability representations are well-established tools in quantum information science, with applications ranging from the classical simulability of quantum computation to quantum process tomography, quantum error correction, and quantum sensing. While traditional quasiprobability representations typically employ real-valued distributions, recent developments highlight the usefulness of \emph{complex-valued} ones---most notably, via the family of Kirkwood--Dirac quasiprobability distributions. Building on the framework of~\href{https://quantum-journal.org/papers/q-2024-03-14-1283/}{Schmid et al.~[Quantum 8, 1283 (2024)]}, we extend the analysis   to encompass complex-valued quasiprobability representations that need not preserve the identity channel. Additionally, we also extend previous results  to consider mappings towards  infinite-dimensional spaces. We show that, for each system, every such representation can be expressed as the composition of two maps that are completely characterized by their action on states and on the identity (equivalently, on effects) for that system. Our results apply to all complex-valued quasiprobability representations of any finite-dimensional, tomographically-local generalized probabilistic theory, with finite-dimensional quantum theory serving as a paradigmatic example. In the quantum case, the maps’ action on states and effects corresponds to choices of frames and dual frames for the representation. This work offers a unified mathematical framework for analyzing complex-valued quasiprobability representations in generalized probabilistic theories. 
\end{abstract}

\maketitle
\tableofcontents

\section{Introduction}

Despite its striking empirical success, the ontology implied by quantum theory remains fundamentally unclear. Foundational investigations have nevertheless yielded remarkable approaches for understanding and applying the theory---even if falling short of their original interpretational aims. This work delivers structural results at the intersection of two such approaches: quasiprobability representations and the generalized probabilistic theories framework. 

Generalized probabilistic theories (GPTs)~\cite{plavala2023general,muller2021probabilistic} originate from the quest to identify what distinguishes quantum theory within a broader landscape of possible physical theories~\cite{hardy2011reformulatingreconstructingquantumtheory,hardy2001quantumtheoryreasonableaxioms,barrett2007information}. This framework has elucidated not only aspects of quantum theory, but also properties of any theory that might supersede it. Moreover, it has also clarified that characteristics sometimes considered uniquely quantum are present in various other classes of theories---a short list includes no-cloning~\cite{spekkens2007evidence}, no-broadcasting~\cite{barnum2007nobroadcasting,erba2024measurement}, entanglement~\cite{aubrun2022entanglement,aubrun2021entangleability,aubrun2024monogamy,dariano2020classical}, 
and teleportation~\cite{spekkens2007evidence,hardy1999disentanglingnonlocalityteleportation,barnum2008teleportationgeneralprobabilistictheories}. From this point of view, quantum theory is merely one in a vast landscape of theories.

A seemingly unrelated topic is that of \emph{quasiprobability distributions}, which originate from Eugine Wigner's seminal attempt to represent quantum theory using probability distributions over phase space~\cite{wigner1932onthequantumcorrection}. In this approach, quantum-mechanical processes are mapped to processes on a phase-space (or more generally, any classical state space), Hilbert spaces $\mathcal{H}$ are replaced by sets of relevant phase-space parameters $\Lambda$, quantum states $\rho$ by functions $\mu_\rho$ over $\Lambda$, quantum channels $\mathcal{E}$ by maps $\Gamma_\mathcal{E}$ acting on such functions, and so on. Because the predictions of any quantum experiment can be fully reproduced in such a manner, we refer to this mapping as a \emph{representation} of quantum theory. 

If such a representation requires violations of standard axioms of probability~\cite{kolmogorov2018foundations}---such as allowing negative or non-real values in the distributions $\mu_\rho$---we call it a \emph{quasiprobability representation}, otherwise we call it a \emph{probability representation}. Formally, such representations can be characterized as \emph{frame representations}~\cite{ferrie2008frame,ferrie2009framed,ferrie2011quasiprobabilitiesreview,christensen2016introduction}. Various representations of quantum theory are known, and most commonly one considers \emph{real-valued} representations---i.e., those in which the quasiprobabilities are real-valued functions that may be negative~\cite{wigner1932onthequantumcorrection,husimi1940some,glauber1963coherent,sudarshan1963equivalence,margenau1961correlation,terletsky1937limiting}.

\begin{figure}[t]
    \centering
\input{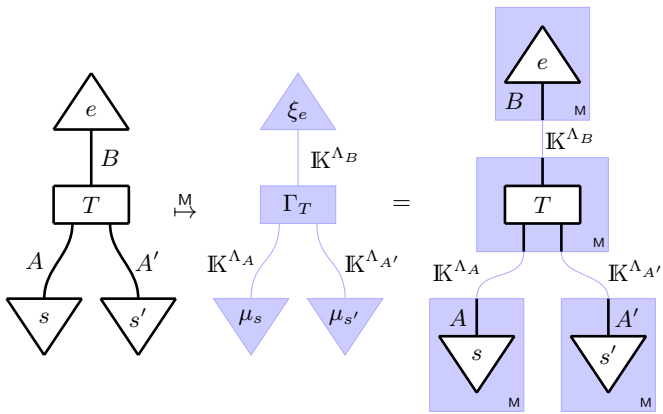}
    \caption{\textbf{Quasiprobability representation of a generalized probabilistic theory (GPT).} In a diagrammatic perspective, a physical theory defines a compositional structure for generic systems (denoted $A, B, \dots$) and processes between them (denoted $s,e,T, \dots$). A quasiprobability representation is a map $\mathsf{M}$ that assigns to each system $A$ a set of relevant variables $\Lambda_A$ taking values in $\mathds{K} \in \{\mathds{R},\mathds{C}\}$, and to each process, a quasistochastic element---such as a quasiprobability distribution $\mu_s$, a response function $\xi_e$, or a quasisubstochastic matrix $\Gamma_T$. The prefix ``quasi'' indicates that these objects may take values outside $[0,1]$, violating Kolmogovo's axioms of probability theory.}
    \label{fig:figure_introduction}
\end{figure}

Unsurprisingly, quantum theory is not the only theory for which one can meaningfully define a quasiprobability representation. In particular, this notion extends naturally to GPTs. Schmid~et~al.~\cite{schmid2024structure}, building on Refs.~\cite{vandeWetering2018quantum,gheorghiu2020ontological,fritz2020synthetic}, developed a framework for {\em real-valued finite-dimensional  identity preserving} quasiprobability representations that not only captures the structural features of GPTs but also emphasizes their compositional character (see Fig.~\ref{fig:figure_introduction}). Most GPTs admit of a diagrammatic representation with an associated compositional calculus, typically formalized as a \emph{symmetric monoidal category}~\cite{mac1998categories,awodey2010category,heunen2019categories}, in which physical systems and processes are represented by diagrams. A quasiprobability representation in this setting is then a map that takes diagrams from a GPT to corresponding diagrams in a target theory defined in terms of quasistochastic matrices.

Formally, this diagrammatic view captures the compositional structure of a GPT in the language of \emph{process theories}~\cite{selby2017process,coecke2017picturing,selby2025generalisedprocesstheories}---mathematical structures closely related to symmetric monoidal categories. The target theory of a quasiprobability representation is described using quasisubstochastic matrices: systems correspond to real vector spaces, processes to matrices, states to quasiprobability vectors, and so on. From this perspective, such representations correspond to \emph{semi-functors} (see Def.~\ref{def:semifunctors}). While GPTs have both sequential and parallel composition, in this work we focus only on \emph{sequential} composition, as this is the setting in which the semi-functorial maps we consider are well-defined. Hence we work with arbitrary semi-functors. Any notion of ``monoidal semi-functorial maps''\footnote{\label{ftnt:connectivity-preserving}More generally, a representation of a GPT is better described by a \emph{connectivity-preserving map}, a concept to be introduced in future work that extends the semi-functorial maps considered here to encompass the monoidal structure present in a process theory.} form a special case. 

In Schmid et al.~\cite{schmid2024structure}, some of us proved a structure theorem for  \emph{functors}---a subclass of semi-functorial maps that, additionally, preserve identity processes (see Def.~\ref{def:functors}). Their framework, together with the structure theorems proved there, can be viewed as a category-theoretic extension of the linear-algebraic program that classifies \emph{linear preservers}~\cite{li1992introlinearpreserver,li2001linearpreserverproblems,fosner2013linearpreserversquantuminfo}. In this program, which traces back at least to Frobenius's theorem characterizing all determinant preserving linear maps~\cite{frobenius1897uber}, one studies linear operators on a \emph{structured} matrix space  \(\mathcal{O} \subseteq \mathrm{Mat}_{\mathds{K}}(n,m)\)---examples include matrices of a specific rank~\cite{Beasley1988}, stochastic and doubly stochastic matrices~\cite{li2002linearmapspreservingpermutation}, bounded operators on complex infinite-dimensional Banach spaces~\cite{bresar_semrl1997linearpreservers}, and so on---and seeks all maps \(f:\mathcal{O}\to\mathcal{O}\) that preserve the set (i.e.\ \(f(\mathcal{O})=\mathcal{O}\), or in the weaker form \(f(\mathcal{O})\subseteq\mathcal{O}\)). 

Most results in this research program show that $f$ typically has the form \begin{equation}
X\mapsto f(X) = A_fXB_f \text{ or } X\mapsto f(X) = A_fX^{\rm T}B_f,
\end{equation}
for fixed $A_f,B_f \in \mathcal{O}$ and dependent on the map $f$. The structure theorem of Ref.~\cite{schmid2024structure} states that any real-valued  finite-dimensional functorial quasiprobability  representation $\mathsf{Q}$ which preserves a few structures relevant for the GPT is equivalently implemented by system-wise maps $A \mapsto \chi_A:A\to\mathsf{Q}(A)$ and acts by conjugation on transformations: \begin{equation}\mathsf{Q}(T)=\chi_B\circ T\circ\chi_A^{-1},
\end{equation}
for every transformation \(T:A\to B\) in the GPT. 

\begin{figure}[t]
    \centering
    \input{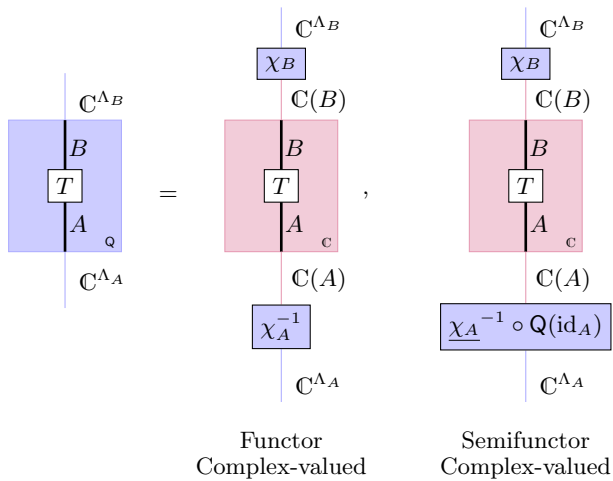}
    \caption{\textbf{Illustration of the main results.} Any linearity-preserving, empirically adequate, complex-valued quasiprobability representation $\mathsf{Q}$ of a generic transformation $T:A\to B$ in a tomographically-local, finite-dimensional GPT can be written as $\mathsf{Q}(T)=\chi_B\circ\mathds{C}(T)\circ\chi_A^{-1}$, provided $\mathsf{Q}$ is functorial and maps the GPT to finite-dimensional spaces (Theorem~\ref{thm: complex_valued_structure_theorem} and Corollary~\ref{corollary:infinite_dim_functors_must_be_fin}). If $\mathsf{Q}$ is only semi-functorial, then $\mathsf{Q}(\mathrm{id}_A)$ is an idempotent and $\mathsf{Q}(T)=\chi_B\circ\mathds{C}(T)\circ\underline{\chi_A}^{-1}\circ\mathsf{Q}(\mathrm{id}_A)$, where $\chi_A$ is injective and $\underline{\chi_A}$ denotes its invertible surjective corestriction (Theorem~\ref{thm: connectivity_preserving theorem}, Corollary~\ref{corollary: quasiprobability representations structure theorem}). In these formulas, $\mathds{C}$ denotes the complexification functor. 
    }
    \label{fig:figure_introduction_2}
\end{figure}

Recently, a particularly important class of \emph{complex-valued} quasiprobability representations has attracted significant attention: the Kirkwood--Dirac (KD) quasiprobability distributions. First introduced by Kirkwood in the 1930s~\cite{kirkwood1933quantum} and later rediscovered by Dirac~\cite{dirac1945ontheanalogy}, these distributions have undergone a revival of interest~\cite{arvidssonShukur2024properties,wagner2024quantumcircuits,halpern2018quasiprobability,gherardini2024quasiprobability}. Most relevantly for our discussion, they have been shown to lift to faithful complex-valued diagram-preserving representations of all of quantum theory~\cite{schmid2024kirkwood}.

Beyond their conceptual appeal, KD quasiprobability representations have found applications across a wide range of areas in quantum information science---including quantum metrology and sensing~\cite{arvidssonshukur2020quantumadvantage,lupu2022negative,das2023saturating,jenne2022unbounded}, quantum computation~\cite{burkat2025structurepositivityclassicalsimulability,thio2025kirkwooddiracnonpositivitynecessaryresource}, nonequilibrium thermodynamics~\cite{halpern2017jarzynski,levy2020quasiprobability,gherardini2024quasiprobability,donati2024energetics,hernandez2024projective,pezzutto2025nonpositiveenergyquasidistributionscoherent,santini2023work}, scrambling of quantum information~\cite{alonso2019out,halpern2018quasiprobability}, and indefinite causal order~\cite{gao2023measuring,ban2021onsequential,azado2025measuringunitaryinvariantsquantum}.

Kirkwood--Dirac representations---and more generally, complex-valued quasiprobability distributions~\cite{wagner2023simple,degosson2012weakvalues,hofmann2011ontherole,hofmann2012complex,hofmann2014derivation,terracunha2001}---lie outside the scope of the formalism and structure theorem of Ref.~\cite{schmid2024structure}. This motivates the need for a broader understanding of the structural features of complex-valued representations of quantum theory in particular, and more broadly of any GPT. The present work takes a step in that direction. Our main contribution is to establish a structure theorem showing that any  complex-valued quasiprobability representation of a finite-dimensional, tomographically-local GPT has a simple and constrained mathematical form.

Our results are illustrated in Fig.~\ref{fig:figure_introduction_2} and are obtained from the following sequence of ideas. Since we consider complex-valued representations, we begin by describing in detail the \emph{complexification} of real vector spaces, and linear maps between them. We then show how complexification can be viewed as a category-theoretic construction, and prove that it yields a strong monoidal functor. Equipped with the complexification functor, we then prove a structure theorem for semi-functors from tomographically-local GPTs to the process theory of  \emph{complex} vector spaces (Theorems~\ref{thm: complex_valued_structure_theorem} and~\ref{thm: connectivity_preserving theorem}). In Corollary~\ref{corollary: quasiprobability representations structure theorem}, we prove a structure theorem for complex-valued quasiprobability representations (which in our formalism are a subset of all possible semi-functors) of tomographically-local GPTs, extending the results of Ref.~\cite{schmid2024structure}. Note that unlike in Ref.~\cite{schmid2024structure} our results (specifically, Theorems~\ref{thm: complex_valued_structure_theorem} and~\ref{thm: connectivity_preserving theorem}) apply equally  well to both finite and infinite-dimensional representations. 

\

\textbf{Outline.} The balance of the paper is organized as follows. We start by providing a comprehensive background section. In Sec.~\ref{sec:background_diagrams} we recall the diagrammatic representation of process theories and GPTs and in Sec.~\ref{sec:examples} we present examples thereof.  Section~\ref{sec:quasiprob} reviews frame representation theory applied to quantum theory, including complex-valued quasiprobability representations like the Kirkwood-Dirac quasiprobability distributions. In Sec.~\ref{sec:connectivity_preserving_maps} we define functors and semi-functors, and formalize quasiprobability representations of GPTs. Section~\ref{sec:complexification_functor} constructs the complexification functor and establishes its categorical properties. 

In Sec.~\ref{sec:structure_theorems} we prove our main structure theorems. We start with a structure theorem for functors mapping to the process theory of (potentially infinite-dimensional) complex vector spaces in Sec.~\ref{sec:functors_to_complex_vector_spaces}.  Section~\ref{sec:semi_functors_to_complex_vectors_spaces} extends this to a structure theorem for semi-functors, which constitutes the key element for proving our structure theorem for quasiprobability representations of finite-dimensional GPTs, as we do in Sec.~\ref{sec:main_result_quasiprobs}. We connect our notion of quasiprobability representations of finite-dimensional GPTs to quantum theory using frame representation theory in Sec.~\ref{sec:examples_quasiprobability_quantum_theory}. Finally, we conclude with Sec.~\ref{sec:conclusion}, which discusses implications and future research directions.

\section{Background}

\subsection{Process theories and generalized probabilistic theories}\label{sec:background_diagrams}

To present our results in a general and conceptually transparent manner, we adopt the framework of generalized probabilistic theories (GPTs)~\cite{hardy2001quantumtheoryreasonableaxioms,barrett2007information,hardy2011reformulatingreconstructingquantumtheory,dariano2017probabilistic,dariano2017quantum,plavala2023general,lami2018nonclassicalcorrelationsquantummechanics,muller2021probabilistic}, a formalism that accommodates quantum theory as a particular case. Many of our constructions are most naturally and elegantly expressed in this broader setting, especially when coupled with the diagrammatic language of process theories, formalized within category theory~\cite{vandeWetering2018quantum,gogioso2018categorical,selby2021reconstructing,coecke2017picturing,coecke2018categorical,gheorghiu2020ontological,schmid2021unscramblingomelettecausationinference,schmid2024structure}. We emphasize that to follow our main results no prior expertise in category theory is assumed to be required---only familiarity with the basic diagrammatic conventions we introduce below. 

We will introduce  GPTs~\cite{dariano2017quantum} from the point of view of the \emph{diagrammatic approach} as discussed in  Ref.~\cite{schmid2024structure}. The class of GPTs under consideration here is the one where each is associated to a \emph{process theory}~\cite{selby2017process,selby2021reconstructing,selby2025generalisedprocesstheories}. Process theories, in their standard description, are in one-to-one correspondence {with} symmetric monoidal categories (SMCs)\footnote{In fact, much of the literature has conflated the two, while more recently foundational work has motivated the perspective that process theories could be viewed from a broader perspective assuming less than what is generically assumed by such categorical constructions~\cite{selby2025generalisedprocesstheories}. }. As has been thoroughly investigated by the aforementioned references, process theories have an associated faithful diagrammatic representation (and calculus), where equalities between diagrams are equivalent to category-theoretic functional relations, and manipulations of diagrams translate faithfully to their categorical counterparts. 

To provide an illustrative example before we go into the details of the specific process theories we will consider, take the following equation: 
\begin{equation}\label{eq:example_swap_and_pair_of_states}
    S_{A,B} \circ (s_1 \otimes s_2) = s_2 \otimes s_1,
\end{equation}
where $s_1,s_2$ are two preparations, $S_{A,B}$ is the swap and these are all processes within the same process theory. These are diagrammatically represented as
\begin{equation}
    \input{FIRST/GPT_exmp_16.tikz}\,\,,\quad \begin{tikzpicture}
	\begin{pgfonlayer}{nodelayer}
		\node [style=none] (3) at (0, 1) {};
		\node [style=none] (7) at (0.55, 0) {\footnotesize $A$};
		\node [style=point] (8) at (0, -1) {$s_1$};
	\end{pgfonlayer}
	\begin{pgfonlayer}{edgelayer}
		\draw [style=HardWire] (8) to (3.center);
	\end{pgfonlayer}
\end{tikzpicture}
\,\,,\text{ and}\,\, \begin{tikzpicture}
	\begin{pgfonlayer}{nodelayer}
		\node [style=none] (3) at (0, 1) {};
		\node [style=none] (7) at (0.55, 0) {\footnotesize $B$};
		\node [style=point] (8) at (0, -1) {$s_s$};
	\end{pgfonlayer}
	\begin{pgfonlayer}{edgelayer}
		\draw [style=HardWire] (8) to (3.center);
	\end{pgfonlayer}
\end{tikzpicture}
\,\,,\quad 
\end{equation}
respectively. A diagrammatic equation represents a valid equivalence between drawn diagrams and composition of processes within the theory. It allows us to substitute equations of the same form as Eq.~\eqref{eq:example_swap_and_pair_of_states} for those having a purely diagrammatic form 
\begin{equation}\label{eq:first_eq_example}    S_{A,B}\circ (s_1 \otimes s_2)\equiv  \input{FIRST/GPT_exmp_1.tikz}\equiv s_2 \otimes s_1.
\end{equation}
This diagrammatic equation concurrently describes: two preparation processes which prepare systems $A$ and $B$ in states $s_1$ and $s_2$ (we denote both the preparations and the states being prepared using the same symbol),  respectively,  the swap process $S_{A,B}$ acting on the pair of systems, and most relevantly the validity of Eq.~\eqref{eq:example_swap_and_pair_of_states}. The intuitively interesting aspect of the diagrammatic Eq.~\eqref{eq:first_eq_example} is that one can imagine sliding the two processes---as dictated by the swap---from the left-hand side to the right-hand side.

The process–theoretic framework is founded on category theory, so every element has a categorical description. For instance, the empty system---denoted $I$ and diagrammatically represented by an empty wire---is called the \emph{monoidal unit} of the process theory. Systems are called \emph{objects}. The swap is an example of a \emph{natural invertible braiding}~\cite{mac1998categories,awodey2010category}. General processes are  \emph{morphisms}, and the collection of all processes from a system $A$ to a system $B$ is called a \emph{hom-set}. 

Because our results span category theory, generalized probabilistic theories, and quantum theory, terminology sometimes shifts between these. To keep our presentation consistent, we will ordinarily adopt the process‑theoretic terms and only invoke alternative jargons when {especially useful} (specifically, we will use  category theory terms when  discussing the complexification functor and stating a few proofs in Appendix~\ref{app: complexification proof}). Table~\ref{tab:terminology_map} provides a dictionary for translating between these languages.

We now give a concise overview of process theories and of GPTs. We will  further focus on tomographically-local and finite-dimensional GPTs, and on process theories for which systems are vector spaces, and our presentation will reflect this choice. For a mathematically rigorous, algebraic treatment (including a historical context) of the GPTs framework, we refer to Ref.~\cite{lami2018nonclassicalcorrelationsquantummechanics}. Moreover, we also refer to  Refs.~\cite{plavala2023general,muller2021probabilistic} for introductory reviews on the topic and Ref.~\cite{dariano2017quantum} for good introduction to the diagrammatic perspective we consider here. 

Since the GPTs  considered here are the ones admitting a diagrammatic calculus\footnote{It is important to note that the GPTs considered in this manuscript, i.e., the ones  arising from (a \emph{quotienting}~\cite{dariano2017quantum,schmid2024structure,schmid2021unscramblingomelettecausationinference} of)  a symmetric monoidal category cannot describe all possible GPTs as formulated within the operator-algebraic tradition as from Refs.~\cite{lami2018nonclassicalcorrelationsquantummechanics,barnum2019stronglysymmetricspectralconvex,barnum2020composites,baldijao2022quantum}. For example, one may be interested in investigating GPTs for which the notion of parallel composition $\hat{\otimes}$ is \emph{not} associative, which implies that $A \otimes (B \otimes C) \neq (A \otimes B) \otimes C$. In category-theoretic terms, such GPTs fail to be monoidal as the associators $a_{A,B,C}$ for systems $A,B,C$ would not be natural isomorphisms $a_{A,B,C}: (A \hat \otimes B) \hat \otimes C \to A \hat \otimes (B \hat \otimes C)$. }, as we go along we will present the elements of GPT theory and describe their diagrammatic counterparts.

For the purpose of this paper, a given GPT, denoted {$\mathbf{G}$}, is, in particular, a process theory where \emph{systems} are real vector spaces, which we denote using uppercase letters $A$, $B$, etc., and processes are linear maps $T: A \to B$\footnote{It is also worth mentioning that not even every GPT having a diagrammatic representation can be described as a subprocess theory of $\realvectexamplep$, introduced shortly. One such example is the class of subtheories dubbed \emph{latent quantum theories} (LQT), recently introduced by Erba and Perinotti~\cite{erba2024compositionrulequantumsystems}; there, the authors conjecture these theories do not admit of an embedding as a strong monoidal functor from LQT to a subtheory of $(\realvectexamplep,\otimes_{\mathds R})$ where $\otimes_{\mathds R}$ is the usual real tensor product (see footnote~9 of~\cite{erba2024compositionrulequantumsystems}).}. In a process theory (or simply, a theory), systems are represented as directed (bottom to top) wires, and processes are represented as boxes having input and output wires. Since the direction is fixed in this manner, it is not diagrammatically represented{,} for simplicity. For example, we denote systems $A$, $B$, and processes $T:A\to B$, respectively, as
\begin{equation}
    \input{FIRST/GPT_exmp_2.tikz}.
\end{equation}
As mentioned before, we use an empty wire to represent the trivial system $I$ of a GPT: 
\begin{equation}
    \begin{tikzpicture}
	\begin{pgfonlayer}{nodelayer}
		\node [style=none] (0) at (0, 0) {$\underbrace{}_{\text{monoidal unit }I}$};
	\end{pgfonlayer}
\end{tikzpicture}
\,\,,
\end{equation}
which we take to be $\mathds{R}$, viewed as a real one-dimensional vector space. 

The set of all processes transforming a system $A$ into a system $B$ in a theory {$\mathbf{G}$} is denoted as {$\mathbf{G}(A,B)$}. For every fixed system $A$ the elements {$s \in \mathbf{G}(I,A)$} are called  \emph{states} of the system $A$, and are diagrammatically represented as a process for which the input is the empty wire and the output is a wire with label $A$. Similarly, the elements {$e \in \mathbf{G}(A,I)$} are called \emph{effects} of the system $A$, and are diagrammatically represented as a process for which the input is a wire with label $A$ and the output is the empty wire. The label of a wire is sometimes referred to as its \emph{type}.  States and effects {of a system $A$} are therefore represented diagrammatically as
\begin{equation}\label{eq: morphism_types}    \input{FIRST/GPT_exmp_4.tikz}.
\end{equation}

\begin{table*}[ht]
\centering
\caption{A glossary of constructions at the intersection of category theory, process theory (in one-to-one correspondence with symmetric monoidal categories), generalized probabilistic theories (special kinds of process theories) and quantum theory (a special kind of GPT).}
\begin{tabular}{l|l|l|l}
\toprule
\textbf{Process Theory} & \textbf{Category Theory} & \textbf{GPTs} & \textbf{Quantum Theory} \\
\hline
\midrule
Process & Morphism (a.k.a. arrow) & Transformation & Quantum channel \\
System & Object & Underlying vector space &$\mathcal{B}(\mathcal{H})_{sa}$\\
Trivial system & Monoidal unit $I$ & Trivial system & Trivial system   \\
Sequential composition of  processes & Sequential composition & Sequential transformations & Composition of channels \\
Parallel composition of processes & {Monoidal composition} & Parallel transformations & Tensor product \\
State (or preparations) & Morphism \( I \to A \) & State  & Density matrix \\
Effect (or measurement event) & Morphism \( A \to I \) & Effect & POVM element \\
Closed diagram & Morphism \( I \to I \) & Probabilities & Born rule probabilities \\
Discarding & Morphism to the terminal object & Deterministic effect & Trace \\
\bottomrule
\hline
\hline
\end{tabular}
\label{tab:terminology_map}
\end{table*}

For every {system} $A$ there exists a unique process $$u_A \equiv \,\,\input{FIRST/GPT_exmp_3.tikz} \,\,\,\in \mathbf{G}(A,I)$$ called {the} \emph{deterministic effect} (also known as the \emph{discarding effect} or the \emph{order unit}~\cite{baldijao2022quantum}).

It is worth pointing out that a generalized probabilistic theory {$\mathbf{G}$} also {has} additional geometric structure. For example, in standard formulations of GPTs the normalized state space {$\mathbf{G}(I,A)$} is  taken to be a pointed closed convex set---endowing the underlying real vector space $A$ with a partial order---and similarly for the effect space. In detail, what is commonly referred to as a \emph{GPT system}~\cite{baldijao2022quantum} is an ordered tuple whose components include: (1) a real vector space $A$, (2) a pointed, closed, convex cone $\mathbf{G}(I,A) \subseteq A$ of states, (3) its dual cone of effects $\mathbf{G}(A,I)\subseteq A^*$\footnote{Here, $A^*$ is the dual space of $\mathds{R}$-linear functionals $A^* = \mathcal{L}_\mathds{R}(A,\mathds{R})$. Whenever $A$ is a complex vector space, for simplicity we will also denote its dual space as $A^*$, even if in this case we have instead $A^* = \mathcal{L}_\mathds{C}(A,\mathds{C})$. }, (4) a discarding effect $u_A$ in $\mathbf{G}(A,I)$, (5) a set of reversible transformations $\mathbf{G}(A,A)$, and finally (6) a composition rule such that $\mathbf G(I,I)$ is a commutative monoid (sometimes referred to as a generalized Born rule).

In particular, for GPTs the states $\mathbf{G}(I,A) \subseteq A$ span the entire vector space $A$, i.e. $\mathrm{span}_{\mathds{R}}(\mathbf{G}(I,A)) = A$ and similarly for the effects $\mathrm{span}_{\mathds{R}}(\mathbf{G}(A,I)) = A^*$ (which is isomorphic to $A$ as these are finite-dimensional spaces). In other words, all processes in the theory are in the linear span of (appropriately typed) compositions of states and effects.

Although these convex-geometric features are \emph{essential} to the usual definition and analysis of GPTs, our results will focus primarily on their compositional (process-theoretic) structure. For this reason, we will not provide a full specification of each GPT, which would require describing every system, its composition rules, and its allowed transformations. Nevertheless, it is useful to have a simple criterion to distinguish process theories that \emph{do not} qualify as GPTs in our sense from those that \emph{do}, namely by possessing   additional geometric structures to the ones just mentioned. In what follows, it will be straightforward to tell whether a theory is a GPT simply by examining its set of states.~\footnote{We stress that this simple rule, valid \emph{only for the purpose of our work}, is a guide to distinguish process theories that are GPTs from those that are not; it is not generally valid and is never sufficient to identify a GPT. }   Whenever we define $\mathbf{G}(I,A)$ without assuming it to be a pointed, closed, convex cone in $A$, the theory will be treated as a process theory that is not a GPT. 

When performing diagrammatic calculations, virtually any manipulation that preserves the input  spaces, output  spaces, and overall topology of a diagram is allowed.  Consequently, it is often clearer to list the \emph{forbidden} manipulations. In the calculus of a process theory~\cite{coecke2017picturing,selby2017process}, one may \emph{not}:
\begin{enumerate}
\item[(i)] Create loops of wires, for example $$\input{FIRST/GPT_exmp_5.tikz}\,\,.$$ 
\item[(ii)] Introduce self‑loops, e.g. $$\input{FIRST/GPT_exmp_6.tikz}\,\,.$$ 
\item[(iii)] Cut a wire between two processes, since this generally produces a different process: $$\input{FIRST/GPT_exmp_7.tikz}.$$ 
\item[(iv)] Commute two {different} processes whose {inputs and outputs} are not the monoidal unit:  $$\input{FIRST/GPT_exmp_8.tikz}.$$ 
\item[(v)] Glue together wires of different types: $$\input{FIRST/GPT_exmp_9.tikz}.$$ 
\item[(vi)] Collapse multiple wires into a single {one}: $$\input{FIRST/GPT_exmp_10.tikz}.$$
\end{enumerate}

For every theory {$\mathbf{G}$} there is  {a} privileged set of processes: all  {those} whose diagrams are in the set {$\mathbf{G}(I,I)$;} {e.g.,}  
\begin{equation}    
\input{FIRST/GPT_exmp_11.tikz}.
\end{equation}
These elements are called \emph{closed diagrams}, and  are simply numbers in $[0,1]$. These are to be interpreted as the probabilistic empirical predictions of the theory. They are privileged in the sense that the set {$\mathbf{G}(I,I)$} forms a commutative monoid with respect to parallel composition since, while generic processes do not (parallel) commute, closed diagrams \emph{do} commute{. For example, } 
$$\input{FIRST/GPT_exmp_14.tikz}\,=\,\input{FIRST/GPT_exmp_13.tikz},$$ as is required if we want to identify these objects with numbers.

An important feature satisfied by many generalized probabilistic theories including quantum theory and all those considered here is  \emph{tomographic locality}, which implies that {the identity} of a composite system can be determined by suitably chosen local preparations and measurements. Every process in such GPTs admits a decomposition as a linear combination of {measure-and-prepare} processes, a fact we will exploit throughout{,} and therefore state now as a Lemma:

\begin{lemma}[Adapted from Ref.~\cite{schmid2024structure}]\label{lemma:decomposition_GPT_processes}
    A generalized probabilistic theory  {$\mathbf{G}$} is tomographically-local if and only if, for every pair of systems $A$ and $B$  there exists a set of states $\{s_i\}_i \in \mathbf{G}(I,B)$ and a set of effects $\{e_j\}_j\in \mathbf{G}(A,I)$ such that, for every process $T \in \mathbf{G}(A,B)$,  
    \begin{equation}
        \input{FIRST/GPT_local_tomography_def_1.tikz} \,\,=\,\, \sum_{ij}r_{ij}\input{FIRST/GPT_local_tomography_def_2.tikz}
    \end{equation}
    with real coefficients  $r_{ij}$ for all $i,j$. 
\end{lemma}

If tomographic locality fails then the composite vector space where systems live satisfies $AB \subsetneq A\otimes B,$ so there are global (holistic) degrees of freedom in $AB$ invisible to product effects; consequently product states $s_1\otimes s_2$ do not span the full composite space. 

\subsection{Examples of process theories and of GPTs}\label{sec:examples}

We now briefly introduce some concrete examples of process theories that will play a central role in this work. Each example specifies a class of systems and processes together with their compositional structure.

\begin{examplep}{E1}[\hypertarget{real_vector_examples}{Process theory $\realvectexamplep$}]  In this process theory, which is not a GPT,  systems are labeled by  vector spaces $V, W, \ldots$ over the field $\mathds{R}$. The trivial system is $\mathds{R}$, viewed as a vector space. Processes are $\mathds{R}$-linear maps $T: V \to W \in \mathbf{Vect}_{\mathds{R}}(V,W)$ between such vector spaces. Parallel composition of systems is given by the usual real tensor product $V \otimes_{\mathds{R}} W$ {between those spaces}, while sequential composition is given by the standard composition of linear maps. A process with no input corresponds to a vector
$$\mathbf{Vect}_{\mathds{R}}(\mathds{R},V) = \{s:\mathds{R}\to V\mid s\text{ is linear}\} \cong V.$$ 
A process with no output corresponds to a covector, i.e. an element of the dual vector space
$$\mathbf{Vect}_{\mathds{R}}(V,\mathds{R}) = \{e:V\to \mathds{R}\mid e\text{ is linear}\} = V^*.$$ A process with neither input nor output corresponds to a scalar in {$\mathds{R}$}.
\end{examplep}

\begin{examplep}{E2}[\hypertarget{complex_vector_examples}{Process theory  $\complexvectexamplep$}]
    This theory is entirely analogous to $\realvectexamplep$ by changing the field $\mathds{R}$ to the field $\mathds{C}$. 
\end{examplep}

In what follows, we will denote the tensor products $\otimes_\mathds{R}$ and $\otimes_\mathds{C}$ by $\otimes$ for simplicity. The type should be clear from the context. For the purposes of this work, we will refer to a subprocess theory  $\mathbf{G}'$ of a theory $\mathbf{G}$, and write $\mathbf{G}' \subseteq \mathbf{G}$, if for every system $A'$ in $\mathbf{G}'$ we have that $A'$ is also a valid system in $\mathbf{G}$ and, similarly, every process $T$ in $\mathbf{G}'$ is also a process in $\mathbf{G}$. Moreover, we also require that the empty wire, sequential and parallel compositions, and swapping operations from $\mathbf{G}'$ are equal to those from\footnote{In the case of GPTs, what constitutes a GPT subtheory is, in fact, rather subtle, as some of us discuss in a forthcoming work. Nevertheless, for the scope of the present work, the simple definition given for process theories is enough.} $\mathbf{G}$.  

\begin{example}
  We now present a few relevant subtheories of $\complexvectexamplep$ and $\realvectexamplep$.
  \begin{itemize}
    \item (\hypertarget{finite_real_vector_examples}{Process theory $\finrealvectexamplep$}) This is the subprocess theory of $\realvectexamplep$ where all wires correspond to \emph{finite-dimensional} real vector spaces and processes correspond to $\mathds{R}$-linear maps between these spaces.
    \item (\hypertarget{finite_complex_vector_examples}{Process theory $\fincomplexvectexamplep$}) Similarly to $\finrealvectexamplep$, this is the subtheory of $\complexvectexamplep$ where all wires correspond to finite-dimensional complex vector spaces and processes correspond to $\mathds{C}$-linear maps between these spaces. 
    \item (\hypertarget{finite_real_vector_substoch}{The GPT $\finiterealsubstoch$})  This is a subtheory of $\finrealvectexamplep$ in which all wires correspond to finite-dimensional real vector spaces of functions from a finite index set 
    $\Lambda$ to $\mathds{R}$, i.e.\ elements of $\mathds{R}^\Lambda$. 
    
    In this process theory, each wire $\mathds{R}^\Lambda$ has the standard basis 
    $\{\vert \lambda \rangle \}_{\lambda\in\Lambda}$, while the dual space carries the canonical dual basis $\{\langle \lambda \vert \}_{\lambda\in\Lambda}$. 
    Diagrammatically,
    \begin{equation}
        \vert \lambda \rangle  = \begin{tikzpicture}
	\begin{pgfonlayer}{nodelayer}
		\node [style=none] (1) at (0, 1) {};
		\node [style=point] (2) at (0, -0.5) {$\lambda$};
		\node [style=none] (3) at (0.5, 0.5) {\footnotesize $\mathds{R}^\Lambda$};
	\end{pgfonlayer}
	\begin{pgfonlayer}{edgelayer}
		\draw (1.center) to (2);
	\end{pgfonlayer}
\end{tikzpicture}
,\text{ and } \langle \lambda \vert = \begin{tikzpicture}
	\begin{pgfonlayer}{nodelayer}
		\node [style=none] (1) at (0, 0.5) {};
		\node [style=none] (3) at (0.5, -0.5) {\footnotesize $\mathds{R}^\Lambda$};
		\node [style=copoint] (4) at (0, 0.5) {$\lambda$};
		\node [style=none] (5) at (0, -1) {};
	\end{pgfonlayer}
	\begin{pgfonlayer}{edgelayer}
		\draw (4) to (5.center);
	\end{pgfonlayer}
\end{tikzpicture}
.
    \end{equation}
    For each $\Lambda$ we define the positive cone $\mathds{R}_+^\Lambda$ as the set of vectors \begin{equation}v = \sum_{\lambda \in \Lambda} v_\lambda \vert \lambda \rangle  \equiv  \sum_{\lambda \in \Lambda}\,\,\input{FIRST/Exemple_sec_5.tikz}\,\,\, \begin{tikzpicture}
	\begin{pgfonlayer}{nodelayer}
		\node [style=none] (1) at (0, 1) {};
		\node [style=point] (2) at (0, -0.5) {$\lambda$};
		\node [style=none] (3) at (0.5, 0.5) {\footnotesize $\mathds{R}^\Lambda$};
	\end{pgfonlayer}
	\begin{pgfonlayer}{edgelayer}
		\draw (1.center) to (2);
	\end{pgfonlayer}
\end{tikzpicture}
\end{equation} with $v_\lambda \geq 0$ for all $\lambda \in \Lambda$. Moreover, to every $\Lambda$ we define (the deterministic effect) $u_\Lambda: \mathds{R}^\Lambda \to \mathds{R}$ as the map \begin{equation}\label{eq:deterministic_effect_summation_functional}
    u_\Lambda(v) = \sum_{\lambda \in \Lambda}v_\lambda \equiv \input{FIRST/Exemple_sec_6.tikz}.
    \end{equation}
    For every $\Lambda$, we denote $${\mathds{R}^\Lambda_+}^*:=\{e \in {\mathds{R}^\Lambda}^* \mid e(v) \geq 0,\, \forall v \in \mathds{R}^\Lambda_+\}$$ the dual cone of $\mathds{R}^\Lambda_+$. Processes in this theory are $\mathds{R}$-linear maps $\Gamma:\mathds{R}^\Lambda \to \mathds{R}^{\Lambda'}$ such that $\Gamma(\mathds{R}^\Lambda_+ ) \subseteq \mathds{R}_+^{\Lambda'}$ and \begin{equation}\label{eq:substoch_normalization_preservation}
    u_{\Lambda'}(\Gamma(v)) = \sum_{\lambda' \in \Lambda'} (\Gamma(v))_{\lambda'} \leq \sum_\lambda v_\lambda = u_{\Lambda}(v)\end{equation} for every $v \in  \mathds{R}_+^\Lambda$. The set of all possible effects is   ${\mathds{R}_{+}^{\Lambda}}^*$ (i.e. this theory satisfies the \emph{no-restriction hypothesis}~\cite{plavala2023general}). Because every element in ${\mathds{R}^\Lambda_+}^*$ is a valid effect, $\Gamma$ preserves also the dual cone of effects, in the sense that for every effect $e \in {\mathds{R}^\Lambda_+}^*$ the \textit{pull-back} $e \circ \Gamma$ is again in the cone, i.e. $e \circ \Gamma \in {\mathds{R}^\Lambda_+}^*$.~\footnote{The pull-back defines what is also known as the dual map $\Gamma^*: {\mathds{R}^{\Lambda'}}^* \to {\mathds{R}^\Lambda}^*$ of $\Gamma: \mathds{R}^\Lambda \to \mathds{R}^{\Lambda'}$ via the definition of $\Gamma^*(e)(v) = e\circ \Gamma (v)$ for every $v \in \mathds{R}^\Lambda$, cf.~\cite{barnum2023selfdualityjordanstructurequantum}. Therefore, we conclude that provided $\Gamma (\mathds{R}_+^\Lambda) \subseteq \mathds{R}_+^{\Lambda'}$ and our theory satisfy the no-restriction hypothesis we have also that $\Gamma^*({\mathds{R}_+^{\Lambda'}}^*) \subseteq {\mathds{R}_+^\Lambda}^*$, i.e. processes also preserve the cone of effects.}   Sequential and parallel composition are given by ordinary composition of maps and the real tensor product, respectively.  
    
    States are processes with no input, $\mu: \mathds{R} \to \mathds{R}^\Lambda$, which we can identify with vectors via $\mathds{R} \cong \mathds{R}^{\{\star\}}$, where $\star$ labels the empty input system (a singleton set),   writing $\mu = \sum_\lambda \mu(\star)_\lambda \vert \lambda \rangle \equiv \sum_\lambda \mu(\lambda\vert \star) \vert \lambda \rangle$. These form the set $$\left\{\mu \mid \mu(\lambda\vert \star) \geq 0,\; \sum_\lambda \mu(\lambda\vert \star) \leq 1\right\},$$
    i.e.\ the subnormalized probability distributions on $\Lambda$. 
    
    Processes with no outputs $\xi:\mathds{R}^\Lambda \to \mathds{R}$ correspond to subnormalized response functions, which we can identify with covectors by writing $\xi = \sum_\lambda \xi(\star\vert \lambda)\langle \lambda \vert $. Closed diagrams satisfy a total probability rule such as 
    \begin{equation}\label{eq:total_prob_rule}
    \xi(\Gamma(\mu)) = \sum_{\lambda,\lambda'}\xi(\star\vert \lambda') \Gamma(\lambda'\vert \lambda)\mu(\lambda\vert \star)
    \end{equation}
     and correspond to scalars in $[0,1]$. 
  \end{itemize}
\end{example}

We remark that every object $V$ of the process theory $\finrealvectexamplep$ is (non-canonically) isomorphic to $\mathds{R}^d$, where $d=\dim_{\mathds{R}} V$. Equivalently, there exists a finite index set $\Lambda$ with $|\Lambda|=d$ and an isomorphism
$$
\varphi_V:\mathds{R}^\Lambda \xrightarrow{\ \simeq\ } V,
$$
but the choice of $\Lambda$ and of $\varphi_V$ is equivalent to the choice of an ordered basis of $V$.  Likewise, given a linear map $T:V\to W$ and chosen bases for $V$ and $W$ (isomorphisms $\varphi_V:\mathds R^\Lambda\!\to V$ and $\varphi_W:\mathds R^{\Lambda'}\!\to W$), the map $T$ is represented by the matrix
$$
L \;:=\; \varphi_W^{-1}\circ T\circ\varphi_V \;:\; \mathds R^\Lambda\longrightarrow\mathds R^{\Lambda'}.
$$
In these coordinates the action of the $\mathds{R}$-linear map $T$ is equivalent to the action of the matrix $L$, and composition of linear maps corresponds to matrix multiplication, etc.  In other words, by fixing a canonical basis $\{e_\lambda\}_{\lambda\in\Lambda}$ on each system we obtain a \emph{basis-dependent} or ``matrix'' version of $\finrealvectexamplep$ in which objects are the concrete vector spaces $\mathds R^\Lambda$ and morphisms are matrices relative to the chosen bases.  (The identification is basis-dependent and therefore non-canonical: different basis choices yield conjugate matrix representations of the same abstract map.) The same holds substituting $\mathds
R$ by $\mathds{C}$. It is precisely \emph{this} identification that will make it clearer that the process theory $\fincomplexquasisubstoch$  (which we will define soon)  is process-theoretically equivalent to  $\fincomplexvectexamplep$. 

The terminology adopted so far is consistent with standard category-theoretic usage~\cite{fritz2020synthetic,vandeWetering2018quantum}. We refer to Example 2.5 from Ref.~\cite{fritz2020synthetic} for a formal account of (a related category-theoretic construction to) the process theory $\finiterealsubstoch$.  

Quantum theory relative to finite-dimensional Hilbert spaces can also be viewed as a process theory---in fact a GPT---, which we will denote as $\FinQT$.

\begin{examplep}{E3}[\hypertarget{finite_quantum_theory}{The GPT $\finitequantumtheory$}]
Systems are associated to the real vector space of self-adjoint bounded operators acting on finite-dimensional (complex) Hilbert spaces $\mathcal{H}$, i.e. $$\mathcal{B}(\mathcal{H})_{\rm sa} = \{A\in \mathcal{B}(\mathcal{H})\mid A=A^\dagger\}.$$ Processes $\mathcal{E}:\mathcal{B}(\mathcal{H})_{\rm sa} \to \mathcal{B}(\mathcal{H}')_{\rm sa} $ are completely positive and trace non-increasing maps. Sequential composition is provided by the standard composition of maps. Parallel composition is provided by the standard tensor product of Hilbert spaces. The trivial system $I$ is again $\R$. Arbitrary  (possibly unnormalized)  states are positive operators
$$\mathbf{FinQT}(\mathds{R},\mathcal{B}(\mathcal{H})_{\rm sa}) \cong \mathcal{B}(\mathcal{H})^+= \{\rho\mid \mathrm{spec}(\rho)\subseteq \mathds{R}_{\geq 0}\},$$
and (possibly unnormalized)  effects are linear functionals on these $$\mathbf{FinQT}(\mathcal{B}(\mathcal{H})_{\rm sa},\mathds{R})= \{\mathrm{Tr}( E,\cdot)\mid E \in \mathcal{B}(\mathcal{H})^+\}.$$ Closed diagrams $\FinQT(\mathds{R},\mathds{R})$ are the diagrammatical equivalent to the Born rule
\begin{equation}
\input{FIRST/Exemple_sec_1.tikz} = \mathrm{Tr}(E \mathcal{E}(\rho)).
\end{equation}
The deterministic effect of a quantum system $ \mathcal{H}$ is provided by the identity operator  $\mathbb{1}_{\mathcal{H}}\simeq \mathrm{Tr}(\mathbb{1}_\mathcal{H} \cdot ) = u_{\mathcal{H}}$, so that it acts as a discard, hence defining \emph{normalized} states $\rho \in \mathcal{D}(\mathcal{H})$ to be those for which 
\begin{equation}
    \input{FIRST/Exemple_sec_2.tikz} = \,\,\text{Tr}(\mathbb{1} \rho) = \text{Tr}(\rho) =  1.
\end{equation}    
\end{examplep}

Since we restrict processes to completely positive maps and that these maps preserve adjoints, the operationally relevant space is the self-adjoint subspace $\mathcal{B}(\mathcal H)_{\mathrm{sa}}$---unnormalized states as defined constitute positive (hence self-adjoint) elements and any map as defined preserves this positivity. Since $\mathcal {B}(\mathcal{H})_{\mathrm{sa}}$ is a finite-dimensional real vector space, $\FinQT$ is clearly a subprocess theory of $\finrealvectexamplep$ (objects regarded as ordered real vector spaces equipped with the trace deterministic effect). 

We now introduce $\finrealquasisubstoch$ and its generalization to the complex field $\fincomplexquasisubstoch$.

\begin{examplep}{E4}[\hypertarget{finite_real_vector_quasisubstoch}{Process theory  $\finrealquasisubstoch$}] Systems are labeled by vector spaces $\mathds{R}^\Lambda$ relative to a fixed choice of a finite index set $\Lambda$. Processes are $\mathds{R}$-linear maps $T: \mathds{R}^\Lambda \to \mathds{R}^{\Lambda'}$ between such vector spaces. The remaining process-theoretic structure is the same as for $\finrealvectexamplep$. The discarding effects $u_\Lambda$ are the same as  the one defined in Eq.~\eqref{eq:deterministic_effect_summation_functional}. Note the difference to $\finiterealsubstoch$: there is no need for positive cone of states/effects, nor processes need to preserve those. In other words, the components of states in this process theories need not be positive.
\end{examplep}

We can further motivate why we term this subprocess theory that of quasisubstochastic processes even though process-theoretically this is equivalent to $\finrealvectexamplep$ (and similarly for complex values). Recall that a real matrix is substochastic iff there exists another nonnegative matrix such that their sum is stochastic.  Equivalently, the same reasoning holds for our definition: every quasisubstochastic process admits a (not necessarily nonnegative) complementary process such that their sum becomes quasistochastic.

We stress that states and effects in $\finrealquasisubstoch$ are \emph{not} necessarily subnormalized states and effects in the sense of $\sum_\lambda \mu(\lambda \vert \star) \leq 1$. Imposing such normalization (or related ordering constraints) incurs several problems. To start, any ordering-based requirement does not generalize to complex amplitudes since $\mathds{C}$ is not totally ordered. Moreover, if we want to grant the possibility that to every quasisubstochastic $\mu$ there should exist another quasisubstochastic $\mu'$ such that $\mu+\mu'$ is quasistochastic we need for $\mu+\mu' =: \mu''$ to be normalized. However, taking $\Lambda = \{0,1\}$ and $\mu''(0)<0$ would imply that $\mu''(1) > 1$ which is problematic if we were to assume that these must be subnormalized. 

There is, nevertheless, a distinguished subset of processes in $\finrealquasisubstoch$: the  \emph{normalization-preserving} processes. In more formal words,  $\mathds{R}$-linear maps $\Gamma:\mathds{R}^\Lambda \to \mathds R^{\Lambda'}$ such that for every $v \in \mathds{R}^\Lambda$,
    \begin{equation}\label{eq:quasistochastic_normalization_preserving}
        u_{\Lambda'} \circ \Gamma(v)=u_{\Lambda'} (\Gamma (v)) = \sum_{\lambda' \in \Lambda'}(\Gamma(v))_{\lambda'} = \sum_{\lambda}v_\lambda =u_{\Lambda}(v).
    \end{equation}
    In analogy with Eq.~\eqref{eq:substoch_normalization_preservation}, where the same equality constraint yields stochastic processes, maps satisfying Eq.~\eqref{eq:quasistochastic_normalization_preserving} can be regarded as the \emph{quasistochastic processes} in $\finrealquasisubstoch$. 
    When $\Lambda = \Lambda' = \{\star\}$, the only normalization preserving linear map $\Gamma:\mathds{R}\to \mathds{R}$ is $\mathrm{id}_{\mathds{R}}$ (i.e. the only dynamics allowed on the trivial system is the trivial dynamics). Processes with no input $\mu:\mathds{R} \to \mathds{R}^\Lambda$ are normalized, i.e., we have $\sum_\lambda \mu(\lambda \vert \star)= 1.$ Processes with no output are arbitrary $\mathds{R}$-linear functionals $\xi: \mathds{R}^\Lambda \to \mathds{R}$.

\begin{examplep}{E4}[\hypertarget{finite_complex_vector_quasisubstoch}{Process theory $\fincomplexquasisubstoch$}] This process theory has the exact same structure as $\finrealquasisubstoch$ exchanging the field $\mathds{R}$ by $\mathds{C}$. 
\end{examplep}

The process theory $\fincomplexquasisubstoch$ has an analogous distinguished subset of quasistochastic complex-valued processes, defined via normalization preservation, as for the quasistochsatic real-valued processes in $\finrealquasisubstoch$. 

To conclude, we note that $\mathbf{Vect}_{\mathds{K}}$, $\mathbf{FinVect}_{\mathds{K}}$, and $\mathbf{FinQuasiSubStoch}_{\mathds{K}}$, where $\mathds{K} \in \{\mathds{R},\mathds{C}\}$, are process theories but \emph{not} GPTs, since (considering our simple yet incomplete rule for distinguishing
between the two) for a given system $A$ the space of states does not form a closed, convex, and pointed cone within $A$. The theories $\finiterealsubstoch$ and $\finitequantumtheory$ \emph{are} GPTs. The state spaces  {$\mathbf{FinSubStoch}_{\mathds{R}}(I,A)$ for a given system $A$ form simplexes, and their dual space of effects dual simplexes. The theory is mathematically equivalent to what is known as a \emph{classical probabilistic theory} in the GPT literature. 

\subsection{Frame representations of quantum theory}\label{sec:quasiprob}

Quasiprobability distributions of quantum theory are 
 commonly defined via \emph{frame representation theory}~\cite{christensen2016introduction}, and in this section we set up the terminology and basic notation we will use. Let $\mathcal{H}$ be any finite-dimensional Hilbert space. We say that $f \equiv \{f_\lambda\}_{\lambda=1}^m$ is a (finite) frame for $\mathcal{H}$ if complex-valued linear combinations of elements span the entire space, i.e.  $\text{span}_{\mathds{C}}(f) = \mathcal{H}$. A frame $f$ is overcomplete if it has more elements than a basis; if it is a basis, it is said to be  nonovercomplete. A subset $f' \subseteq f$ of a frame $f$ is a subframe if $f'$ is also a frame for $\mathcal{H}$. Every frame $f$ for a finite-dimensional Hilbert space $\mathcal{H}$ contains a (possibly non-unique) nonovercomplete subframe $f'$.

The bounded operators acting on a Hilbert space can be viewed as a complex Hilbert space $(\mathcal{B}(\mathcal{H}), \Vert \cdot \Vert_{\mathrm{HS}})$ where $\Vert \cdot \Vert_{\mathrm{HS}}:\mathcal{B}(\mathcal{H}) \to \mathds{R}$ is the Hilbert--Schmidt (HS) norm\footnote{In finite-dimensional spaces, the Hilbert--Schmidt norm is equivalent to the \emph{Frobenius matrix norm} $\Vert A \Vert_{\mathrm{Frob}}^2 = \sum_{i,j} |\langle e_i,A e_j\rangle|^2 = \sum_i \Vert Ae_i \Vert^2$ where $\{e_i\}_i$ is a basis for $\mathcal{H}$.} $$\Vert A \Vert^2_{\mathrm{HS}} = \text{Tr}(A^\dagger A)= \langle A,  A\rangle_{\mathrm{HS}}.$$ 
We denote frames of the Hilbert space of bounded operators $\mathcal{B}(\mathcal{H})$ by $F$. If $(\cdot)^\dagger$ is the involution induced by the inner-product making $\mathcal{H}$ a Hilbert space, we say that $F = \{F_\lambda\}_{\lambda}$ is a \emph{self-adjoint} frame iff $F_\lambda = F_\lambda^\dagger$ for all $\lambda$. 

Quantum states $\rho \in \mathcal{D}(\mathcal{H})$ are represented by a finite frame $F = \{F_\lambda\}_{\lambda \in \Lambda}$ via the family of linear functionals $\mu(\lambda | \cdot ): \mathcal{B}(\mathcal{H}) \to \mathds{C}$ given by    
\begin{equation}\label{eq:def_mu}
    \mu(\lambda | \cdot ) := \langle F_\lambda, \cdot  \rangle_{\mathrm{HS}} = \mathrm{Tr}(F_\lambda^\dagger\,\cdot) 
\end{equation}
for every $\lambda \in \Lambda$. In particular, the family of functionals $\mu(\lambda \vert \cdot)$ provides a frame for $\mathcal{B}(\mathcal{H})^*$ whenever $F$ is a frame for $\mathcal{B}(\mathcal{H})$. We will refer to $\mu_F(\rho) \equiv \{\mu(\lambda|\rho)\}_{\lambda \in \Lambda} \in \mathds{C}^\Lambda$ as a  \emph{frame representation of the state} $\rho$. It is meaningful to call it a ``representation'' because there always exists some other frame $G = \{G_\lambda\}_{\lambda \in \Lambda}$ for which we can write $\rho$ as 
\begin{equation}\label{eq: representing rho}
    \rho = \sum_{\lambda \in \Lambda} \mu(\lambda|\rho) G_\lambda = \sum_{\lambda \in \Lambda} \text{Tr}(F_{\lambda}^\dagger \rho)G_\lambda.
\end{equation}
The set $G$ is called {a} \emph{dual frame} {for $F$}. In general, for each choice of frame $F$ there is no unique dual frame $G$ for which Eq.~\eqref{eq: representing rho} holds. The function $$\mu_F: \mathcal{B}(\mathcal{H}) \to \mathds{C}^\Lambda \simeq \underbrace{\mathds{C} \times \dots \times \mathds{C}}_{|\Lambda|\text{ times}},$$
$$X \mapsto (\langle F_\lambda,X\rangle_{\mathrm{HS}})_{\lambda \in \Lambda}$$
is $\mathds{C}$-linear\footnote{For every $X,Y \in \mathcal{B}(\mathcal{H})$ and every $\alpha,\beta \in \mathds{C}$ we have  $\mu_F(\alpha X + \beta Y) = \alpha \mu_F(X)+\beta \mu_F(Y)$.}
and, due to Eq.~\eqref{eq: representing rho}, invertible for \emph{any} specific choice of frame $F$ when restricted as a map onto the image $\mu_F(\mathcal{B}(\mathcal{H}))$\footnote{Note that $\mu_F^{-1}$ is unique on the image of $\mu_F$---as any inverse map---yet the \emph{reconstruction} $\rho$ via Eq.~\eqref{eq: representing rho} can have many representations depending on the choice of dual frame $G$.   This may be viewed as a freedom in representing the map $\mu_F^{-1}$ depending on the choice of $G$ yet, as maps, they all act equivalently.}.

In fact, the \emph{canonical dual frame} $G$ {for} $F$ is the frame constructed via $G = \mathcal{S}^{-1}(F)$ where $\mathcal{S}$ is the invertible self{-}adjoint bounded operator defined by 
\begin{equation}
    \mathcal{S}(A) := \sum_{\lambda \in \Lambda} \text{Tr}(A^\dagger F_\lambda)F_\lambda.
\end{equation}
$\mathcal{S}$ is called the frame operator.\footnote{ $\mathcal{S}(F)$ is also a frame because for an arbitrary frame $F$, and surjective operator $\mathcal{E} \in \mathcal{B}(\mathcal{B}(\mathcal{H}))$, we have that $\mathcal{E}(F) \equiv \{\mathcal{E}(F_\lambda)\}_{\lambda \in \Lambda}$ is also a frame~\cite[Corollary~5.3.2, pg.~129]{christensen2016introduction}. }
When $F$ is a basis {(hence a nonovercomplete frame)} for $\mathcal{B}(\mathcal{H})$ the canonical dual frame is the \emph{unique} frame for which 
\begin{equation}\label{eq: biorthogonality frames}
    \langle F_{\lambda}, G_{\lambda'} \rangle_{\mathrm{HS}} = \delta_{\lambda \lambda'}
\end{equation}
for all $\lambda, \lambda' \in \Lambda$. 

Given any dual frame $G${,} we say that the positive operators $E \in \mathcal{B}(\mathcal{H})^+$ are represented by the dual frame $G$ via the family of linear functionals
\begin{equation}
    \xi(\cdot|\lambda) := \langle \cdot, G_\lambda\rangle_{\mathrm{HS}}.
\end{equation}
We will refer to $\xi_G(E) \equiv \{\xi(E|\lambda)\}_{\lambda \in \Lambda}$ as \emph{a frame representation for the positive operator $E$.} As before, it is meaningful to call $\xi$ a representation for $E$ because 
\begin{equation}\label{eq: frame representation E}
    E = \sum_{\lambda \in \Lambda} \xi(E|\lambda)F_\lambda = \sum_{\lambda \in \Lambda}\text{Tr}(EG_\lambda)F_\lambda. \end{equation}
Note that above we have used that positive operators are self-adjoint, and therefore $\text{Tr}(E^\dagger G_\lambda)=\text{Tr}(E G_\lambda)$.  The functions $\xi_G: \mathcal{B}(\mathcal{H}) \to \mathds{C}^\Lambda$ are $\mathds{R}$-linear  but \emph{not} $\mathds{C}$-linear. They are instead  conjugate-linear since, denoting the complex-conjugate of a complex number $\alpha$ {by $\overline{\alpha}$},
\begin{equation}
    \xi_G(\alpha A+\beta B) = \overline{\alpha}\xi_G(A)+\overline{\beta}\xi_G(B),
\end{equation}
following from the sesquilinearity of the inner product $\langle \cdot , \cdot \rangle_{\mathrm{HS}}$.  Again, from Eq.~\eqref{eq: frame representation E}, these maps are invertible. 

Note that $F$ and $G$ can be overcomplete and still satisfy Eqs.~\eqref{eq: representing rho} and~\eqref{eq: frame representation E}. Moreover, provided that we have a frame and its dual frame (not necessarily the canonical dual frame), the representations $\xi_G(E)$ of $E$ and $\mu_F(\rho)$ of $\rho$ correctly reproduce the Born rule via the formula
\begin{align}
    \sum_{\lambda} \mu(F_\lambda|\rho)\xi(E|G_\lambda) &= \sum_\lambda \text{Tr}(F_\lambda^\dagger \rho ) \text{Tr}(EG_\lambda) \nonumber \\
    &=\text{Tr}\left(E \sum_\lambda \text{Tr}(F_\lambda^\dagger \rho)G_\lambda\right) \nonumber\\
    &= \text{Tr}(E\rho).\label{eq: empirical_adequacy_frames_nodiagrams}
\end{align}

If we consider a completely positive and trace-nonincreasing map $\mathcal{E}: \mathcal{B}(\mathcal{H_{\mathrm{in}}}) \to \mathcal{B}(\mathcal{H}_{\mathrm{out}})$, frames $F^{\mathrm{in}} \subseteq \mathcal{B}(\mathcal{H}_{\mathrm{in}}), F^{\mathrm{out}} \subseteq \mathcal{B}(\mathcal{H}_{\mathrm{out}}),$ and corresponding dual frames $G^{\mathrm{in}}, G^{\mathrm{out}}$ we can also consider \emph{frame representations of the channel $\mathcal{E}$} given by $$\Gamma_{F,G}(\mathcal{E}) \equiv \{\Gamma(\lambda'|\lambda,\mathcal{E})\}_{\lambda' \in \Lambda_{\mathrm{out}},\lambda\in \Lambda_{\mathrm{in}}}$$ defined as  

\begin{equation}    
\Gamma(\lambda_{\mathrm{out}}|\lambda_{\mathrm{in}},\mathcal{E} ) := \langle F^{\mathrm{out}}_{\lambda_{\mathrm{out}}}, \mathcal{E}(G^{\mathrm{in}}_{\lambda_{\mathrm{in}}})\rangle_{\mathrm{HS}}.
\end{equation}

Similarly to the case of states and effects, it is valid to call this a representation because we can always write the action of the map $\mathcal{E}$ as

\begin{equation}
    \mathcal{E}(\cdot) = \sum_{\lambda_{\rm in},\lambda_{\rm out}}  \Gamma(\lambda_{\mathrm{out}}|\lambda_{\mathrm{in}},\mathcal{E} ) \langle F_{\lambda_{\mathrm{in}}}^{\mathrm{in}}, \,\,\cdot \,\,\rangle_{\mathrm{HS}} G_{\lambda_{\mathrm{out}}}^{\mathrm{out}}.
\end{equation}

Table~\ref{tab: summary_notions} organizes the terminology of frame representation theory. 

Representations of the identity channel are not always the identity matrix. Let $F$ and $G$ be a frame and dual frame pair {for} $\mathcal{B}(\mathcal{H})$, and denote the identity quantum channel as $\mathrm{id}: \mathcal{H} \to \mathcal{H}$. In this case, $$\Gamma(\lambda'|\lambda, \mathrm{id}) = \langle F_{\lambda'},G_\lambda \rangle_{\mathrm{HS}}.$$ When $F$ and $G$ do not satisfy Eq.~\eqref{eq: biorthogonality frames}{,} the {resulting} representation $\Gamma_{F,G}(\mathrm{id})$ of $\mathrm{id}$ is not equivalent to the identity matrix. When $F$ is a basis then there always exists a canonical dual $G$ satisfying this property, while for overcomplete frames this is not  the case anymore. However, \emph{even if} $F$ is a basis, there are {(non-canonical)} choices of $G$ for which the identity operator $\mathrm{id}$ is not represented by the identity matrix.

It is, nevertheless, easy to see that any representation of the identity $\Gamma_{F,G}(\mathrm{id})$ must define an \emph{idempotent} matrix with entries $\Gamma(\lambda',\lambda\vert \mathrm{id})$, i.e, a matrix $P$ satisfying $P^2=P$. To see this, note that for every $\lambda,\lambda'$ we have
\begin{align}    &\sum_{\lambda''}\Gamma(\lambda',\lambda'' \vert \mathrm{id})\Gamma(\lambda'',\lambda \vert \mathrm{id})    =\sum_{\lambda''}\langle F_{\lambda'},G_{\lambda''}\rangle\langle F_{\lambda''},G_{\lambda}\rangle \nonumber\\
&=\left \langle F_{\lambda'}, \underbrace{\sum_{\lambda''} \langle F_{\lambda''},G_{\lambda}\rangle G_{\lambda''}}_{\stackrel{\text{\eqref{eq: representing rho}}}{=}G_\lambda}\right \rangle= \langle F_{\lambda'},G_\lambda \rangle = \Gamma(\lambda',\lambda \vert \mathrm{id}).\label{eq:idempotency_frame_representation_equation}
\end{align}

Frame representations provide a unified framework for classifying what is commonly known as a \emph{quasiprobability distribution} or a \emph{quasiprobability distribution function}~\cite{ferrie2008frame,ferrie2009framed,schmid2024structure}. These are defined in different ways in the literature. What is most common is to define real-valued quasiprobability functions as linear invertible functions $\mu: \mathcal{B}(\mathcal{H})_{\rm sa} \to \mathds{R}^\Lambda$~\cite{ferrie2009framed}, where $\mathcal{B}(\mathcal{H})_{\rm sa}$ denotes the self-adjoint operators. 

\begin{definition}[Quasiprobability distributions - Adapted from Ref.~\cite{ferrie2009framed}]~\label{def:quasiprobability_ferrie}
    Let $\mathcal{H}$ be a Hilbert space. A quasiprobability distribution  is a linear invertible map $\mu: \mathcal{B}(\mathcal{H})_{\rm sa} \to \mathds{R}^\Lambda$, where $\Lambda$ is a set. 
\end{definition}

We will extend this definition later to consider complex-valued distributions. As shown by Refs.~\cite{ferrie2008frame,ferrie2009framed}, it turns out that quasiprobability distributions according to Def.~\ref{def:quasiprobability_ferrie} are always {given by} frame representations for some self-adjoint frame $F$. 

\begin{theorem}[From Ferrie and Emerson, 2009~\cite{ferrie2009framed}]~\label{theorem: ferrie_and_emerson}
    A map $\mu: \mathcal{B}(\mathcal{H})_{\rm sa} \to \mathds{R}^\Lambda$ is a quasiprobability distribution  function (as by Def.~\ref{def:quasiprobability_ferrie}) iff there exists a self-adjoint frame $F \subseteq \mathcal{B}(\mathcal{H})$ for which $\mu = \mu_F$. 
\end{theorem}

\begin{table}[t]
    \centering
    \caption{Glossary for constructions in frame representation theory. We denote the bounded operators acting on a Hilbert space {as $\mathcal{B}(\mathcal{H})$}. The set $2^X$ denotes the power set of $X$. The symbol $\mathds{C}^\Lambda$ denotes the set of functions $\Lambda \to \mathds{C}$.}
    \begin{tabular}{ccc}
    \hline\hline
       Symbol  &  Description  &  Belongs to \\
    \hline
        $F = \{F_\lambda\}_\lambda$  & Frame for $\mathcal{B}(\mathcal{H})$ & $2^{\mathcal{B}(\mathcal{H})}$ \\
    \hline 
    $G = \{G_\lambda\}_\lambda$  & Dual frame for $\mathcal{B}(\mathcal{H})$ & $2^{\mathcal{B}(\mathcal{H})}$  \\
    \hline 
    $\mu_F(\rho)$  & Frame representation of $\rho$ & $\mathds{C}^\Lambda$  \\
    \hline 
    $\xi_G(E)$  & Frame representation of $E$ & $\mathds{C}^\Lambda$  \\
    \hline 
    $\Gamma_{F,G}(\mathcal{E})$  & Frame representation of $\mathcal{E}$ & $\mathds{C}^{\Lambda \times \Lambda'}$  \\
    \hline 
    $\mathcal{S}$  & Frame operator & $\mathcal{B}(\mathcal{B}(\mathcal{H}))$  \\
    \hline 
    $\mathcal{S}^{-1}(F)$  & Canonical dual frame of $F$ & $2^{\mathcal{B}(\mathcal{H})}$  \\
    \hline\hline
    \end{tabular}
    \label{tab: summary_notions}
\end{table}

Therefore, we shall refer to quasiprobability distributions as \emph{quasiprobability representations} of a state since they can always be viewed as arising from \emph{some} frame representation. Moreover, due to terminology followed by  Refs.~\cite{spekkens2008negativity,ferrie2011quasiprobabilitiesreview}, one has also a notion of a quasiprobability representation of a prepare-and-measure fragment of quantum theory as opposed to just a representation of the quantum states (or measurement effects), which is provided by functions $\mu$ and $\xi$ capable of representing the Born rule via Eq.~\eqref{eq: empirical_adequacy_frames_nodiagrams}. 

A crucial aspect of moving from quasiprobability representations of states to that of a theory is that one requires \emph{more} than just writing down a pair of invertible functions $\mu$ and $\xi$, i.e., one requires Eq.~\eqref{eq: empirical_adequacy_frames_nodiagrams}.   Refs.~\cite{vandeWetering2018quantum,schmid2024structure} have generalized these ideas to consider quasiprobability representations beyond such prepare-and-measure fragments to representations of any generalized probabilistic theory (as we will describe in Sec.~\ref{sec:connectivity_preserving_maps}). 

There are various functions that do not fit Def.~\ref{def:quasiprobability_ferrie} that have been considered valid quasiprobability distributions. In particular, as pointed out by Ref.~\cite{arvidssonShukur2024properties}, KD distributions fail to satisfy Def.~\ref{def:quasiprobability_ferrie} both because they may be non-invertible, {and because they may be} complex-valued. 

\begin{definition}[Standard KDQ]
    Let $\mathcal{H}$ be a $d$-dimensional Hilbert space and fix two orthonormal bases $\mathbb{a} \equiv  \{\vert a\rangle \}_{a=1}^d$, and $ \mathbb{b} \equiv \{\vert b \rangle \}_{b=1}^d$. The \emph{standard KDQ} distribution with respect to these two bases is the function\footnote{Since the function depends on the pairs of orthonormal bases $\mathbb{a}$ and $\mathbb{b}$ a more appropriate (yet heavy) notation would be $\mu_{\mathrm{KD}}^{\mathbb{a},\mathbb{b}}$.} $\mu_{\mathrm{KD}}:  \mathcal{B}(\mathcal{H}) \to \mathds{C}^{d \times d}$ given by    
    \begin{equation}\label{eq:standard_KQD}
    \mu_{\mathrm{KD}}(a,b|\rho) := \langle a|\rho|b\rangle \langle b|a\rangle.
    \end{equation}
\end{definition}

Thus, we extend Def.~\ref{def:quasiprobability_ferrie} as follows:

\begin{definition}[Faithful complex-valued quasiprobability distributions]~\label{def:quasiprobability_complex_valued}
    A \emph{(complex-valued) quasi-probability distribution} is a linear map $\mu: \mathcal{B}(\mathcal{H}) \to \mathds{C}^\Lambda$, where $\Lambda$ is an index set. We say that $\mu$ is \emph{faithful} if it is injective.
\end{definition}
Note that although quasiprobability distributions are ultimately used to represent quantum states (elements of $\mathcal{B}(\mathcal{H})_{\mathrm{sa}}$), their domain is taken to be the full operator space $\mathcal{B}(\mathcal{H})$. 
This choice, which may seem unmotivated at this point, will be useful for later constructions of frame representations that lead to complex-valued quasiprobability representations.  

Definitions~\ref{def:quasiprobability_complex_valued} and~\ref{def:quasiprobability_ferrie} are broad enough to allow normalized states $\rho$ (i.e. $\mathrm{Tr}(\rho)=1$) to be represented by quasiprobability
distributions $\mu$ whose components need not sum to one, i.e.
$\sum_{\lambda\in\Lambda}\mu(\lambda\!\mid\!\rho)\neq 1$.  This usage is consistent with the common practice of calling any map that violates Kolmogorov's axioms a ``quasiprobability''. For example, Husimi's quasiprobability distribution~\cite{husimi1940density} is positive everywhere and normalized. Nevertheless, it is still considered a quasiprobability distribution since it does not yield marginals as a classical probability distribution. 

Faithful quasiprobability distributions are invertible when restricted to their image  $\mu:\mathcal{B}(\mathcal{H}) \to \mu(\mathcal{B}(\mathcal{H}))\subseteq \mathds C^\Lambda$. When $\mu$ is faithful, it defines a frame representation of quantum states; consequently Definition~\ref{def:quasiprobability_complex_valued} admits an analogue of Theorem~\ref{theorem: ferrie_and_emerson}. 

\begin{theorem}
\label{theorem:extension_ferrie_and_emerson}
    Let $\mathcal{H}$ be a finite-dimensional Hilbert space. A  $\mathds{C}$-linear map $\mu: \mathcal{B}(\mathcal{H}) \to \mathbb{C}^\Lambda$ is a faithful quasiprobability distribution (as in Def.~\ref{def:quasiprobability_complex_valued}) iff there exists a frame $F \subseteq \mathcal{B}(\mathcal{H})$ for which $\mu = \mu_F$. 
\end{theorem}

\begin{proof}
    This follows directly from the proof strategy given in Refs.~\cite{ferrie2008frame,ferrie2009framed}. The ``if'' part is trivial. For the ``only if'' part, given some quasiprobability distribution function $\mu: \mathcal{B}(\mathcal{H}) \to \mathbb{C}^\Lambda$ we note that for each fixed $\lambda \in \Lambda$, $\mu(\lambda|\cdot): \mathcal{B}(\mathcal{H}) \to \mathbb{C}$ is a bounded linear functional $\mu(\lambda \vert \cdot) \in \mathcal{B}(\mathcal{B}(\mathcal{H}),\mathds{C})$. From the Riesz--Fréchet representation theorem~\cite[Theorem 7.214]{bru2023cstaralgebras} there exists a \emph{unique} $F_\lambda \in \mathcal{B}(\mathcal{H})$  such that $$\mu(\lambda|X) = \langle F_\lambda, X \rangle_{\mathrm{HS}},$$
    for every $X \in \mathcal{B}(\mathcal{H})$.
    To conclude, we note that $F := \{F_\lambda\}_{\lambda \in \Lambda} \subseteq \mathcal{B}(\mathcal{H})$ is necessarily a frame since $\mu$ is faithful (and thus $\text{span}_{\mathds{C}}(F) = \mathcal{B}(\mathcal{H})\,$).
\end{proof}

An immediate corollary of this theorem is that Kirkwood--Dirac quasiprobability distributions are not always frame representations of states for some frame, since they are not always faithful (e.g. KDQs for which there are $(a,b) \in \mathbb{a}\times \mathbb{b}$ s.~t.~$\langle a\vert b\rangle = 0$). Ref.~\cite{schmid2021characterization} shows that whenever $\mathbb{a},\mathbb{b}$ are such that $\langle a\vert b\rangle \neq 0$ for all $(a,b) \in \mathbb{a} \times \mathbb{b}$ the function $\mu_{\rm KD}$ has an associated frame $F = \{F_{a,b}\}_{(a,b) \in \mathbb{a}\times \mathbb{b}}$ given by 
\begin{equation}\label{eq:KD_frame}
    F_{a,b} = \vert a\rangle \langle b \vert \langle a\vert b \rangle, 
\end{equation}
so that 
\begin{align*}
    \mu_F(a,b\vert \rho) := \mathrm{Tr}(F_{a,b}^\dagger \rho) = \langle b\vert a \rangle \mathrm{Tr}(\vert b\rangle \langle a \vert \rho) \\= \langle a \vert \rho \vert b \rangle \langle b \vert a \rangle \stackrel{\eqref{eq:standard_KQD}}{=} \mu_{\rm KD}(a,b \vert \rho).
\end{align*}
This frame $F$ is also a basis and has a canonical dual frame
\begin{equation}\label{eq:KD_dual_frame}
    G_{a,b} = \frac{\vert a \rangle \langle b\vert }{\langle b\vert a \rangle}
\end{equation}
which is the unique dual frame satisfying 
\begin{align}
\langle F_{a,b},G_{a',b'}\rangle_{\mathrm{HS}} &= \mathrm{Tr}(F_{a,b}^\dagger G_{a',b'}) \\&= \mathrm{Tr}\left(\frac{\vert b\rangle \langle a \vert \langle b \vert a \rangle \vert a'\rangle \langle b'\vert}{\langle b'\vert a'\rangle }\right) =\delta_{a,a'}\delta_{b,b'}.  
\end{align}

\subsection{Quasiprobability representations of generalized probabilistic theories}\label{sec:connectivity_preserving_maps}

Our definition of a quasiprobability representation of a generalized probabilistic theory will be related to the notion of a map between different theories. Maps that preserve the sequential composition \emph{and} the identity map are called functors.

\begin{definition}[Functors]\label{def:functors}
    A map $\mathsf{M}: \mathbf{A} \to \mathbf{B}$ between process theories, diagrammatically represented as
    \begin{equation}
        \mathsf{M}\,::\,\input{FIRST/Connectivity_preserving_sec_1.tikz} \,\mapsto\, \input{FIRST/Connectivity_preserving_sec_2.tikz}
    \end{equation}
    is a \emph{functor} if:
    \begin{enumerate}
    \item For any pair of processes  $T_1, T_2$ that can be sequentially composed{,} {the map} preserves sequential composition{:}          \begin{equation}\label{eq:composition_preservation}
        \input{FIRST/Connectivity_preserving_sec_32.tikz} \quad = \quad \input{FIRST/Connectivity_preserving_sec_33.tikz}\,\,.
        \end{equation}
    \item For every system $A${,} {the map} preserves the identity map $\mathrm{id}_A: A \to A$:
\begin{equation}\label{eq:identity_preservation}
    \input{FIRST/Connectivity_preserving_sec_34.tikz} = \input{FIRST/Connectivity_preserving_sec_35.tikz}.
\end{equation}
    \end{enumerate}
\end{definition}

Note that from Def.~\ref{def:functors} a functor is {defined by its} action on systems \textit{and} processes. We denote the processes in the target with different color and thickness of wires for clarity. Moreover, whenever necessary, we also label the box representing the map  in the bottom right corner  since in some cases we will consider different maps in a single diagrammatic equation.

If we wish to account for additional structure arising from all forms of composition allowed by the theory---such as parallel composition or swapping systems---then we must refine our notion to account for this additional structure. As such, we arrived at the notion of \emph{diagram-preserving maps}, which preserve the \emph{full} compositional structure of the theory.

\begin{definition}[Diagram-preserving map]\label{def:diagram_preserving_map}
    A diagram-preserving map $\mathsf{M}: \mathbf{A} \to \mathbf{B}$ between  process theories is a \emph{functor} that is a structure-preserving map between the two process theories, i.e. a map that also preserves parallel compositions, the swap operations, and the empty wires{:}
    \begin{equation*}
    \input{FIRST/Connectivity_preserving_sec_3.tikz} \stackrel{\mathsf{M}}{\mapsto} \input{FIRST/Connectivity_preserving_sec_4.tikz} = \input{FIRST/Connectivity_preserving_sec_5.tikz}.
\end{equation*}
    From a category-theoretic perspective, if we view process theories as symmetric monoidal categories, diagram-preserving maps are \emph{strict symmetric monoidal functors} between the two categories.
\end{definition}

We separated the definition of a functor from that of a diagram-preserving map because our main structural results need only the axioms of the former and not the stronger diagram-preserving requirement. The stronger condition of diagram-preservation is motivated by the notion of an \emph{ontological model} of a generalized probabilistic theory, which is why it was particularly relevant in Ref.~\cite{schmid2024structure}. In this work, however, neither the parallel composition of systems and processes nor the notion of classicality provided by ontological models is of central concern, so we do not require diagram-preservation. 

At the same time we emphasize that a representation of a \emph{theory} ought to respect that theory's \emph{compositional} structure: if a map is to be interpreted as a representation, it should reproduce how processes and systems compose, though not necessarily the  operational description of those processes. Concretely, we introduce a \emph{semi}-functor: a map satisfying only the first axiom in Def.~\ref{def:functors}. Here we therefore focus on sequential composition and adopt the weaker requirement that enforces only the sequential composition axiom; in particular a semi-functor need not preserve identity processes and may map identity processes to non-identity ones. If a semi-functorial representation also preserves parallel composition, it approaches the diagram-preservation notion of Def.~\ref{def:diagram_preserving_map} up to the preservation of the identity, but that stronger condition is not needed for the results developed in the present work; we leave a systematic treatment of parallel composition to future work.

\begin{definition}[Semi-functors]\label{def:semifunctors}
    A map $\mathsf{N}: \mathbf{A} \to \mathbf{B}$ between process theories, diagrammatically represented as
    \begin{equation}
    \mathsf{N}\,::\,\input{FIRST/Connectivity_preserving_sec_15.tikz} \,\mapsto\, \input{FIRST/Connectivity_preserving_sec_16.tikz}
    \end{equation}
    is a \emph{semi-functor} if {f}or any pair of processes $T_1, T_2$ that can be sequentially composed, it satisfies: $\mathsf{N}(T_1 \circ T_2) = \mathsf{N}(T_1) \circ \mathsf{N}(T_2)$. This is diagrammatically represented as
        \begin{equation}
        \input{FIRST/Connectivity_preserving_sec_17.tikz} \quad = \quad \input{FIRST/Connectivity_preserving_sec_18.tikz}\,\,.
        \end{equation}
\end{definition}

The action of a  semi-functor $\mathsf{N}$ on an identity must be of the form
\begin{equation}
     \input{FIRST/Connectivity_preserving_sec_8.tikz} = \input{FIRST/Connectivity_preserving_sec_9.tikz},
\end{equation}
where $\mathsf{N}(\mathrm{id}_A:A \to A) = D_A:\mathsf{N}(A) \to \mathsf{N}(A)$, and $D_A$ is an idempotent map, i.e., for which $D_A \circ D_A = D_A$. This follows from the fact that semi-functors preserve composition for all processes $\mathsf{N}(T\circ G) = \mathsf{N}(T) \circ \mathsf{N}(G)$ for all $T,\,G$ with compatible domain and codomain. Thus, the image of the identity must be an idempotent:

\begin{equation}
    \input{FIRST/Connectivity_preserving_sec_10.tikz} = \input{FIRST/Connectivity_preserving_sec_11.tikz} = \input{FIRST/Connectivity_preserving_sec_12.tikz} = \input{FIRST/Connectivity_preserving_sec_13.tikz} = \,  \input{FIRST/Connectivity_preserving_sec_14.tikz}\,\,.  
\end{equation}

Another crucial requirement for any map that qualifies as a representation of a theory is what we refer to as \emph{empirical adequacy}. We will first define it for the real-valued case. 

\begin{definition}[Empirically adequate maps]
    {Let $\mathbf{G}$ be a tomographically-local GPT.} We say that a map {$\mathsf{M}: \mathbf{G} \to \realvectexamplep$} is \emph{empirically adequate} if it acts as the identity {on} every element in $\mathbf{G}(I,I)$ where $I$ denotes the empty wire. In other words, for every {$p \in \mathbf{G}(I,I)$} we have $\mathsf{M}(p) = p$. 
\end{definition}

Note that this definition, which was taken from Ref.~\cite{schmid2024structure}, is structured relative to a map $\mathsf{M}$ whose domain and codomain theories are subtheories of $\realvectexamplep$. We will later extend this to include a situation where the codomain is $\complexvectexamplep$, given that we are interested in complex-valued representations. Diagrammatically, empirical adequacy for a map $\mathsf{M}: \mathbf{G} \to \realvectexamplep$ implies  that to every state-effect pair  defined on the same system, the map $\mathsf{M}$ satisfies
\begin{equation}
\input{FIRST/Connectivity_preserving_sec_26.tikz}\, = \,\input{FIRST/Connectivity_preserving_sec_27.tikz},
\end{equation}
{and} similarly for any other {closed diagram---that is, element} of {$\mathbf{G}(I,I)$}. In short, if we recall that scalars in $\realvectexamplep$ can be arbitrary numbers in $\mathds{R}$, empirical adequacy of a map $\mathsf{M}$ guarantees that the empirical predictions of the theory---provided by the generalized Born rule which defines closed diagrams---are preserved.

Another structural notion (that we will consider to be a requirement for a quasiprobability representation of a theory) is that of \emph{linearity preservation}. This requirement encodes the idea that the representation should respect probabilistic mixtures: convex combinations of states and processes in the theory should map to the corresponding convex combinations in the representation. 

\begin{definition}[Linearity-preserving real maps]
    We say that a map $\mathsf{M}: \mathbf{A} \to \mathbf{B}$ where $\mathbf{A},\mathbf{B}$ are both subprocess theories of {$\realvectexamplep$} is linearity preserving if for every $\alpha,\beta \in \mathds{R}$ and every {$T_1,T_2\in \mathbf{A}(A,A')$} we have that
        \begin{equation}
            \mathsf{M}\left(\alpha \,T_1 + \beta \,T_2 \right ) = \alpha \,\mathsf{M}(T_1) + \beta \,\mathsf{M}(T_2),
        \end{equation}
        where $\mathsf{M}(T): \mathsf{M}(A)\to\mathsf{M}(A')$.
\end{definition}

Since we want to introduce the notion of a complex-valued quasiprobability representation we need to update the notion of empirical adequacy and of linearity-preservation to consider complex vector spaces. We can do so using the \emph{standard embedding} given by the map $e_{\mathds{C}}: W \hookrightarrow W'$ of any system $W\in \realvectexamplep$ into a system $W'\in \complexvectexamplep$ by the mapping $w \mapsto w+i\cdot 0$.

Given that we have this description, we can update our notion of empirically adequate maps.

\begin{definition}[Empirically adequate complex-valued  maps]\label{def:extension_of_empirical_adequacy}
    We say that a map $\mathsf{M}: \mathbf{G} \to \complexvectexamplep$ is \emph{empirically adequate} if it acts as the composition of the standard embedding and the identity in every element in {$\mathbf{G}(I,I)$, i.e., for every $p \in \mathbf{G}(I,I)$ we have $\mathsf{M}(p) =e_{\mathds{C}}(p)$}. 
\end{definition}

Diagrammatically, taking again the example of a state-effect pair in  {$\mathbf{G}$}, a map $\mathsf{M}$ satisfying empirical adequacy implies the following:

\begin{equation}
\input{FIRST/Complexification_5.tikz} \,= e_{\mathds{C}}\left(\input{FIRST/Complexification_6.tikz}\right)\,= \input{FIRST/Complexification_7.tikz}+\,\,i\cdot 0.  
\end{equation}

Similarly, we can update our notion of linearity preservation.

\begin{definition}[Linearity-preserving real-to-complex maps]\label{def:complex_valued_linearity_preserving}
    We say that a map $\mathsf{M}: \mathbf{A} \to \mathbf{B}$ where $\mathbf{A}\subseteq \realvectexamplep$ and $\mathbf{B}\subseteq \complexvectexamplep$ is linearity preserving if, for every $\alpha,\beta \in \mathds R$ and every $T_1,T_2:A \to A' \in \mathbf{A}(A,A')$, we have that 
    \begin{equation}
        \mathsf{M}(\alpha T_1 + \beta T_2) = e_{\mathds{C}}(\alpha) \mathsf{M}(T_1)+e_{\mathds{C}}(\beta) \mathsf{M}(T_2),
    \end{equation}
    where $\mathsf{M}(T):\mathsf{M}(A)\to \mathsf{M}(A')$.
\end{definition}

The above provides a simple notion of linearity-preservation for any map $\mathsf{M}:\mathbf{A} \to \mathbf{B}$ whenever $\mathbf{A} \subseteq \realvectexamplep$ and $\mathbf{B} \subseteq \complexvectexamplep$. These will satisfy 
\begin{equation}\label{eq:linearity-preservation}   \input{FIRST/Complexification_9.tikz} = \input{FIRST/Complexification_10.tikz}
\end{equation}
for arbitrary processes $T$ and arbitrary $\alpha_i \in \mathds{R}$. Note that, of course, the left-hand-side $\alpha_i \in \mathds{R}$ while for the right-hand-side we have $\alpha_i+i\cdot 0 \in \mathds C$. 

With all the constructions described previously, we are ready to state our definition of {a complex-valued} quasiprobability representation of a generalized probabilistic theory. (Recall that here we only consider finite-dimensional quasiprobability representations, as mentioned at the end of Sec.~\ref{sec:quasiprob}.)

\begin{definition}[Complex-valued quasiprobabilistic representation of a GPT]\label{def:quasiprobability_representation}
    A \emph{complex-valued} quasiprobability representation of a generalized probabilistic theory {$\mathbf{G}$} is an empirically-adequate,  linearity-preserving semi-functor {$\mathsf{Q}:\mathbf{G} \to \fincomplexquasisubstoch$}. Diagrammatically, it is represented as
    \begin{equation}
        \mathsf{Q}\,::\,\input{FIRST/Connectivity_preserving_sec_28.tikz} \,\mapsto\, \input{FIRST/Connectivity_preserving_sec_29.tikz}.
    \end{equation}
    If, additionally, for every system $A$ we have that \begin{equation}\label{eq:determinist_effect_preservation}
        \mathsf{Q}(u_A) = \input{FIRST/Connectivity_preserving_sec_30.tikz} \,=\, \input{FIRST/Connectivity_preserving_sec_31.tikz} = u_{\Lambda_A} = u_{\mathsf{Q}(A)},
    \end{equation}
    i.e. if the quasiprobability representation  $\mathsf{Q}$ preserves the deterministic effect of a fixed system, we call it a \emph{discard-preserving} representation. 
\end{definition}

When $\mathsf{Q}$ is a quasiprobability representation that preserves both deterministic effects and identity morphisms (i.e., is a functor), we call it a \emph{quasistochastic representation}.

It is simple to consider quasiprobability representations which are \emph{not} discard-preserving representations. For instance, the map $\rho \mapsto \langle i\vert \rho \vert j\rangle$---which simply vectorizes a state with respect to a certain orthonormal basis---can be thought of as a quasiprobability representation of a state. However, this map fails to represent the deterministic effect of a given system appropriately: for a qubit, choosing the ordering $$\rho \mapsto \vec{\rho} := (\langle 0 \vert \rho \vert 0 \rangle, \langle 0 \vert \rho \vert 1 \rangle, \langle 1 \vert \rho \vert 0 \rangle, \langle 1 \vert \rho \vert 1\rangle )$$ it maps the identity operator to the covector
\begin{align}
    u_{\mathcal{H}} \mapsto \vec{u}_{\mathcal{H}} = &(u_{\mathcal{H}}(\vert 0 \rangle \langle 0 \vert ),u_{\mathcal{H}}(\vert 0 \rangle \langle 1 \vert,u_{\mathcal{H}}(\vert 1 \rangle \langle 0 \vert,u_{\mathcal{H}}(\vert 1 \rangle \langle 1   \vert)^T \nonumber \\
    = & (1,0,0,1)^T,
\end{align}
rather than the deterministic effect $$u_{\{1,2,3,4\}} = (1,1,1,1)^T \equiv \mathbf{1}_4,$$ relative to $\mathds{C}^{\{1,2,3,4\}}$ which sums all vector components associated to the canonical basis (see Eq.~\eqref{eq:deterministic_effect_summation_functional}).  

If we take the domain GPT to be $\finitequantumtheory$ we can now see that any mapping $\mathcal H \mapsto (F,G)$
where for each system $\mathcal{B}(\mathcal{H})_{sa}$ we associate a pair of a frame $F$ and a dual frame $G$, yields a complex-valued quasiprobability representation of $\finitequantumtheory$ as by Def.~\eqref{def:quasiprobability_complex_valued}. This holds because, if we carefully see the descriptions in our definition:

\begin{enumerate}
    \item The frame and its dual may contain non-Hermitian elements of $\mathcal{B}(\mathcal{H})$ (hence elements lying outside $\finitequantumtheory$). Consequently, the resulting representation can be complex-valued.
    \item As reviewed above in Eq.~\eqref{eq: empirical_adequacy_frames_nodiagrams}, the product at the representation level reproduces the Born rule, so the representation is empirically adequate.
    \item As is well-known, the frame representation respects sequential composition: the representation of a composition of two maps equals the composition of their representations. Moreover, as shown in Eq.~\eqref{eq:idempotency_frame_representation_equation}, the generic representation of the identity channel is an idempotent, which motivates considering semi-functorial representations.
\end{enumerate}

To conclude, we note that the converse will also hold: every quasiprobability representation of quantum theory (viewed as a GPT) can be realized as some frame representation. We discuss this correspondence and its conditions later on in Sec.~\ref{sec:examples_quasiprobability_quantum_theory}.
}

\subsection{Complexification functor}\label{sec:complexification_functor}

Since our main structural result concerns complex-valued quasiprobability distributions, an important aspect of our construction is that it involves a certain complexification procedure, which we now describe in detail.

The \emph{complexification functor} $\mathds{C}: \realvectexamplep \to \complexvectexamplep$ will be defined in two steps: first we define the complexification of \textit{a vector space}, and second, we define the complexification of \textit{linear maps} between real vector spaces. Note that we will denote the functor by $\mathds{C}$ using the same symbol as the set of complex numbers $\mathds{C}$. We will follow closely Ref.~\cite{conrad2014complexification}. 

We start by defining the complexification of a real vector space $W$, denoted $W_{\mathds{C}}$.

\begin{definition}[Complexification of a real vector space and the standard embedding]\label{def: complexification_vector_spaces}
    Let $W$ be a real vector space. We say that its \emph{complexification} is the complex vector space $W_{\mathds{C}} = W \oplus W \cong W+iW$, $W \cap iW = \{0\}$, with scalar  multiplication law $$(a+bi)(w_1,w_2) := (aw_1-bw_2,bw_1+aw_2),$$
    for every scalar $a+ib \in \mathds{C}$ and every $w_1,w_2 \in W$. 
\end{definition}

The choice of multiplication law just described allows us to make the association $W_{\mathds{C}} \simeq W+iW$ with $W \cap iW = \{0\}$. This mapping just described is  intuitive since it follows from the desiderata
\begin{equation*}
    (a+ib)(w_1+iw_2) = (aw_1-bw_2)+i(bw_1+aw_2).
\end{equation*}
According to this multiplication rule we have that $i(w,0) = (0,w)$ we can write \begin{align}
(w_1,w_2) &= (w_1,0)+(0,w_2) = (w_1,0)+i(w_2,0) \nonumber\\
&= e_{\mathds{C}}(w_1)+ie_{\mathds{C}}(w_2)
\end{align}
where $e_{\mathds{C}}: w \mapsto (w,0) \simeq w+i\cdot 0$ is the standard embedding.

With this notion of complexification, we view $W_{\mathds{C}}$ as generated by two copies of $W$. From this association it follows the usual associations such as 
\begin{align}
    \mathds{C} &= \mathds{R} \oplus \mathds{R} \simeq \mathds{R}+i\mathds{R} \equiv \mathds{R}_{\mathds{C}},
\end{align}
from which we know that any complex number is written as a real and imaginary ``real part'', and ``generated'' by $\mathds{C}$-linear combinations of elements of $\mathds{R}$. Similarly, any element of a $C^*$-algebra~\cite{bru2023cstaralgebras} can be generated by complex combinations of self-adjoint ones~\cite{landsman2012mathematical,landsman2017foundations}, i.e., 
\begin{equation}
    \mathfrak{U} = \mathfrak{U}_{\rm sa}\oplus \mathfrak{U}_{\rm sa} \simeq \mathfrak{U}_{\rm sa}+i\mathfrak{U}_{\rm sa} \equiv \mathfrak{U}_{\mathds{C}}.
\end{equation}
For us it will be relevant that for the finite-dimensional algebras $\mathfrak{U}=\mathcal{B}(\mathcal{H})$ we end up having
\begin{equation}\label{eq:complexification_self_adjoint}
    \mathcal{B}(\mathcal{H}) \simeq \mathcal{B}(\mathcal{H})_{\rm sa}+i\mathcal{B}(\mathcal{H})_{\rm sa}\equiv (\mathcal{B}(\mathcal{H})_{\rm sa})_{\mathds{C}},
\end{equation}
from which we conclude that we can view $\mathcal{B}(\mathcal{H})$ as the complexification of $\mathcal{B}(\mathcal{H})_{\rm sa}$.
\begin{definition}[Complexification of a linear map]\label{def: complexification_linear_map}
    Let $W,\,V$ be real vector spaces, and $f: W \to V$ be an $\mathds{R}$-linear map. We define the complexification $f_{\mathds{C}}$ of $f$ as the $\mathds{C}$-linear map from $W_{\mathds{C}}$ to $V_\mathds{C}$ via  
    \begin{align}
        f_{\mathds{C}} (w_1,w_2) &:= (f(w_1),f(w_2)) \label{eq: linear_map_complexification}\\
        &= e_{\mathds{C}}(f(w_1)) +i \, e_{\mathds{C}}(f(w_2)).\nonumber
    \end{align}
\end{definition}

From above we can see that $f$ defines $f_{\mathds{C}}$ via (using $(w_1,w_2) = e_{\mathds{C}}(w_1)+ie_{\mathds{C}}(w_2) \equiv w_1 + i w_2$) the simple equation $$f_{\mathds{C}}(w_1 +i w_2) := f(w_1)+if(w_2).$$ It is easy to see that $f_{\mathds{C}}$ is $\mathds{C}$-linear whenever $f$ is $\mathds{R}$-linear. For a fixed $V \stackrel{f}{\to}W$ the complexification $V_{\mathds{C}} \stackrel{f_{\mathds{C}}}\to W_{\mathds{C}}$ is also the \emph{unique} map that makes the diagram
\begin{equation}\label{eq: commutation complexification linear maps}
    \begin{tikzcd}
W \arrow[rr, "f"] \arrow[dd, "e_{\mathds{C}}"] &  & V \arrow[dd, "e_{\mathds{C}}"] \\
                                               &  &                                \\
W_{\mathds{C}} \arrow[rr, "f_{\mathds{C}}"]    &  & V_{\mathds{C}}                
\end{tikzcd}
\end{equation}
commute. 
\begin{theorem}[Adapted from Ref.~\cite{conrad2014complexification}]\label{theorem: unique_complex_extension}
    Let $W$ be a real vector space and $V$ be a complex vector space. For each $\mathds{R}$-linear map $W \stackrel{\hat f}{\to} V$ there is a unique $\mathds{C}$-linear map $W_\mathds{C} \stackrel{f}\to V$ making the diagram
    
    \begin{equation}
    \begin{tikzcd}
                                             &  & W_{\mathds{C}} \arrow[dd, " f", dotted] \\
W \arrow[rru, "e_{\mathds{C}}"] \arrow[rrd, "\hat f"] &  &                                             \\
                                             &  & V                                          
\end{tikzcd}
    \end{equation}
    commute, where $W \stackrel{e_{\mathds{C}}}\to W_{\mathds{C}}$ is the standard embedding. 
\end{theorem}

Above the dotted arrow indicates that the extension $f$ exists and is induced uniquely by $\hat f$. In category theory, results such as Theorem~\ref{theorem: unique_complex_extension} are common to many  constructions and are known as \emph{universal properties}. In particular, the proof of the theorem above is constructive and $f$ is given by  
\begin{equation}
    f: W_{\mathds{C}} \ni w = (w_1,w_2) \mapsto \hat f(w_1)+i \hat f(w_2) \in V. 
\end{equation}
This result is particularly relevant to us as it shows that $\hat f$ has a unique $\mathds{C}$-linear extension and it satisfies 
\begin{equation}\label{eq:universal_property_equation}
\hat f = f \circ e_{\mathds{C}},
\end{equation}
$$w \mapsto f(w+i\cdot 0) = \hat f(w)+i\cdot \hat f(0) = \hat f(w)+i\cdot 0,$$ 
where the first equation follows from the definition of $f$ and the second follows from $\hat f$ being $\mathds{R}$-linear.

As a concrete example of a mapping from a real to a complex vector space we can take the standard KD quasiprobability distribution $\mu_{\rm KD}$ where we restrict its domain to be $\mathcal{B}(\mathcal{H})_{\rm sa}$. This is a mapping from $\mathcal{B}(\mathcal{H})_{\rm sa}$---viewed as a real-vector space---to the space $\mathds{C}^{d \times d}$ (where $d$ is the dimension of the underlying Hilbert space). In general, this will be a complex-valued distribution. For instance, take the canonical example where $\mathcal{B}(\mathbb{C}^2)$ and the two KD basis are taken to be the mutually unbiased bases $\mathbb{a} = \{\vert 0\rangle\,\vert 1\rangle \}$ and $\mathbb{b} = \{\vert +\rangle, \vert -\rangle\}$. In this case, 
\begin{equation}
    \mu_{\rm KD}:: \ni\rho \mapsto \left(\begin{matrix}
        \langle +\vert \rho \vert 0\rangle\langle 0\vert +\rangle  & \langle +\vert \rho \vert 1\rangle \langle 1\vert +\rangle   \\
        \langle -\vert \rho \vert 0\rangle \langle 0\vert -\rangle & \langle -\vert\rho\vert 1\rangle\langle 1\vert -\rangle  
    \end{matrix}\right).
\end{equation}
This KDQ distribution of a normalized state is complex-valued iff $\rho$ lies outside the X-Z plane~\cite{li2025multistateimaginaritycoherencequbit}. The possible values $\langle a\vert \rho \vert b \rangle \langle b \vert a \rangle \in \mathds C$ can take have been completely characterized in Refs.~\cite{fernandes2024unitary,pratapsi2025elementarycharacterizationbargmanninvariants,xu2025numericalrangesbargmanninvariants,zhang2025geometry,li2025bargmann}. 

We are now ready to define the complexification functor $\mathds{C}: \realvectexamplep \to \complexvectexamplep$. We define it by its action on objects $W \in \realvectexamplep$ as $$\mathds{C}(W) = W_{\mathds{C}},$$ where $W_\mathds{C}$ is the complexification of the vector space $W$ as by Def.~\ref{def: complexification_vector_spaces}, and its action on morphisms $f: W \to V$ as $$\mathds{C}(f: W \to V) = f_\mathds{C}:W_{\mathds{C}} \to V_{\mathds{C}},$$ where $f_{\mathds{C}}$ is the complexification of the linear map $f$ as by Def.~\ref{def: complexification_linear_map}. These choices clearly map objects and morphisms from $\realvectexamplep$ into $\complexvectexamplep$. Also, note that for any object $W$ we have that for all $(w_1,w_2) \in \mathds{C}(W)$, \begin{align*}
\mathds{C}(\mathrm{id}_{W})(w_1,w_2)&=(\mathrm{id}_{W}(w_1),\mathrm{id}_{W}(w_2)) \\
&= \mathrm{id}_{\mathds{C}(W)}(w_1,w_2).
\end{align*}
Diagrammatically, this can be represented as follows:
\begin{equation}\label{eq:complexification_preserves_the_identity}
\input{FIRST/Complexification_1.tikz} \quad = \quad  \begin{tikzpicture}
	\begin{pgfonlayer}{nodelayer}
		\node [style=none] (0) at (0, 0.75) {};
		\node [style=none] (1) at (0, -0.75) {};
	\end{pgfonlayer}
	\begin{pgfonlayer}{edgelayer}
		\draw [style=ComplexMap] (0.center) to (1.center);
	\end{pgfonlayer}
\end{tikzpicture}
\,\,.
\end{equation}
In our work, we also use a color scheme to distinguish between the complexification functor $\mathds{C}$ (in pink) and a generic map $\mathsf{M}$ (in blue). Nevertheless, we always indicate the type of map in the right corner of the colored boxes when needed.

In addition to identity preservation, for every pair of $\mathds{R}$-linear maps $f:X\to Y$ and $g: Y \to Z$ we have that 
\begin{align}
    \mathds{C}(g \circ f)(w_1,w_2) &:= (g \circ f(w_1),g \circ f(w_2))\\
    &=(g(f(w_1)),g(f(w_2))) \\
    &=\mathds{C}(g)(f(w_1),f(w_2))\\
    &=\mathds{C}(g) \circ \mathds{C}(f)(w_1,w_2),
\end{align}
for every $(w_1,w_2) \in \mathds{C}(W)$. Diagrammatically, this can be represented as: 
\begin{equation}
    \input{FIRST/Complexification_3.tikz} = \input{FIRST/Complexification_4.tikz}\,\,.
\end{equation}
This shows that $\mathds{C}$ is indeed a functor. Note, moreover, that to show the above we have not used anything specific to the functor $\mathds{C}$. The only requirement used so far is that functions are lifted from their action on one space to their action on \emph{tuples}. 

As a matter of fact, it is simple to see that $\mathds{C}$ is also linearity-preserving in the sense of Def.~\ref{def:complex_valued_linearity_preserving}. Because of that it preserves the span, as we show in Appendix~\ref{app:span_proofs}.

The functor $\mathds{C}$ is not only a functor but also a faithful strong monoidal functor, a technical result we prove in Appendix~\ref{app: complexification proof}. While we were unable to locate an explicit proof in the literature that the complexification functor is a faithful strong monoidal functor, related observations have appeared in previous works. For instance, Ref.~\cite{baez2011division} discusses the complexification process as a functor---apparently assuming its strong monoidal nature without proof, as indeed we see in Appendix~\ref{app: complexification proof} follows from elementary calculations.

\begin{lemma}\label{lemma:faithful_strong_monoidal_functor}
    The complexification functor $\mathds{C}:\realvectexamplep\to \complexvectexamplep$ is a faithful strong monoidal functor. 
\end{lemma}

As a category-theory remark, we recall a point made in Ref.~\cite[Remark~2.1, pg.~6]{schmid2024structure}: any strong monoidal  functor between symmetric monoidal categories (i.e. process theories as we define here) can be extended to a diagram-preserving map. Because of that, $\mathds{C}$ is effectively---up to this extension---a diagram-preserving map.

\section{Structure theorems for complex-valued  representations}\label{sec:structure_theorems}

\subsection{Functors mapping GPT systems to complex vector spaces}\label{sec:functors_to_complex_vector_spaces}

We are now ready to state and prove our main results. Recall that from our presentation  {$\mathbf{G}$} is a process theory of $\mathds{R}$-linear spaces and maps while $\complexvectexamplep$ is a process theory of $\mathds{C}$-linear spaces and maps. In what follows, we will denote the monoidal unit $I$ of the domain GPTs $\mathbf{G}$ as $\mathds{R}$. We start with our main structure theorem for functors from {$\mathbf{G}$} to $\complexvectexamplep$.

\begin{theorem}[Structure theorem for functors]\label{thm: complex_valued_structure_theorem}
    Let $\mathbf{G}$ be a tomographically-local finite-dimensional generalized probabilistic theory. Any linearity-preserving and  empirically-adequate functor {$\mathsf{M}: \mathbf{G} \to \complexvectexamplep$} can be decomposed as \begin{equation}\label{eq:CrepTrans}
    \input{FIRST/Structure_theorem_1.tikz}\quad = \quad \input{FIRST/Structure_theorem_2.tikz},
    \end{equation}
    where, for every fixed $A$, $\chi_A$ is an invertible linear map within $\complexvectexamplep$ uniquely determined by the action of $\mathsf{M}$ on states $s \in \mathbf{G}(\mathds R,A)$. 
\end{theorem}

\begin{proof}
    From tomographic locality of  {$\mathbf{G}$} (and therefore Lemma~\ref{lemma:decomposition_GPT_processes})  and linearity-preservation of $\mathsf{M}$ it follows that for every $A$ there exists $\{e_j\}_j$ spanning $A^*$ and $\{s_i\}_i$ spanning $A$ such that 
\begin{equation}
	\input{FIRST/Structure_theorem_3.tikz}\ =\ \input{FIRST/Structure_theorem_4.tikz}\ = \sum_{ij} \mathcolorbox{complex}{r_{ij}}\ \input{FIRST/Structure_theorem_5.tikz}.
\label{eq:step1}
\end{equation}

Above and henceforth, we use the notation  $\mathcolorbox{complex}{a} = \mathds{C}(a)$ for $a \in \mathds{R}$. Using Lemma~\ref{lemma:decomposition_GPT_processes} applied to the identity process $\mathrm{id}_A: A \to A$, and the fact that the GPT is finite-dimensional, we find that there is always a resolution of the identity process,
\begin{equation}\label{eq:partition_of_unity_GPT}
    \begin{tikzpicture}
	\begin{pgfonlayer}{nodelayer}
		\node [style=none] (0) at (0, -1) {};
		\node [style=none] (1) at (0, 1) {};
		\node [style=none] (2) at (0.5, 0) {\footnotesize $A$};
	\end{pgfonlayer}
	\begin{pgfonlayer}{edgelayer}
		\draw [style=HardWire] (0.center) to (1.center);
	\end{pgfonlayer}
\end{tikzpicture}
\,=\sum_{ij}t_{ij}\,\input{FIRST/Structure_theorem_80.tikz},
\end{equation}
where $t_{ij}\in \mathds{R}$ for every $i,j$,  provided by the same set of states and effects as in Eq.~\eqref{eq:step1}. Hence, for every $s \in \mathbf{G}(\mathds{R},A)$, and every $e \in \mathbf{G}(A,\mathds R)$ we have 
\begin{equation}\label{eq:decomposition_states}
    \begin{tikzpicture}
	\begin{pgfonlayer}{nodelayer}
		\node [style=none] (0) at (0, -1) {};
		\node [style=none] (1) at (0, 1) {};
		\node [style=none] (2) at (0.5, 0) {\footnotesize $A$};
		\node [style=point] (3) at (0, -1) {$s$};
	\end{pgfonlayer}
	\begin{pgfonlayer}{edgelayer}
		\draw [style=HardWire] (0.center) to (1.center);
	\end{pgfonlayer}
\end{tikzpicture}
 \,=\, \sum_{ij}t_{ij} \,\,\input{FIRST/Structure_theorem_82.tikz},
\end{equation}
and
\begin{equation}\label{eq:decomposition_effects}
    \begin{tikzpicture}
	\begin{pgfonlayer}{nodelayer}
		\node [style=none] (0) at (0, -1) {};
		\node [style=none] (1) at (0, 1) {};
		\node [style=none] (2) at (0.5, 0) {\footnotesize $A$};
		\node [style=copoint] (3) at (0, 1) {$e$};
	\end{pgfonlayer}
	\begin{pgfonlayer}{edgelayer}
		\draw [style=HardWire] (0.center) to (1.center);
	\end{pgfonlayer}
\end{tikzpicture}
 \,=\, \sum_{ij}t_{ij}\,\,\input{FIRST/Structure_theorem_84.tikz}\,.
\end{equation}
From empirical adequacy (Def.~\ref{def:extension_of_empirical_adequacy}) we obtain 
\begin{equation}
\input{FIRST/Structure_theorem_20.tikz} = \input{FIRST/Structure_theorem_21.tikz} 
\equiv \input{FIRST/Structure_theorem_22.tikz}+\,\,0\cdot i,
\end{equation}
for every state $s$ and effect $e$.  Putting those facts together, we conclude that for every state $s: \mathds{R}\to B$ in $\mathbf{G}$ we can use the fact that $\mathsf{M}$ is a functor, together with linearity-preservation (as by Def.~\ref{def:complex_valued_linearity_preserving}) and empirical adequacy to rewrite $\mathsf{M}(s):\mathds{C}\to V_B$ as 
\begin{align*}
    &\input{FIRST/Structure_theorem_90.tikz}\,\stackrel{\eqref{eq:decomposition_states}}{=}\,\input{FIRST/Structure_theorem_85.tikz} = \input{FIRST/Structure_theorem_91.tikz} \\&=\input{FIRST/Structure_theorem_86.tikz}\,=\,\input{FIRST/Structure_theorem_87.tikz}\\
    &=\input{FIRST/Structure_theorem_88.tikz}\,=:\,\input{FIRST/Structure_theorem_89.tikz},
\end{align*}
where we have defined 
\begin{equation}\label{eq:definition_chi_map}
    \input{FIRST/Structure_theorem_74.tikz}\ :=\ \sum_{ij}\mathcolorbox{complex}{t_{ij}}\ \input{FIRST/Structure_theorem_75.tikz}\,\,.
\end{equation}
By repeating the same argument---\emph{mutatis mutandis}, with states replaced by effects and using Eq.~\eqref{eq:decomposition_effects}---we obtain that for every effect $e:A \to \mathds{R}$, the corresponding map $\mathsf{M}(e):V_A\to \mathds{C}$ can be rewritten accordingly, i.e.
\begin{equation}
    \input{FIRST/Structure_theorem_92.tikz} \quad = \quad  \input{FIRST/Structure_theorem_93.tikz}
\end{equation}
where $\phi_A$ is given by: 
\begin{equation}\label{eq:restricted_map_phi}
    \input{FIRST/Structure_theorem_76.tikz}\ :=\ \sum_{ij}\mathcolorbox{complex}{t_{ij}}\ \input{FIRST/Structure_theorem_77.tikz}\,.
\end{equation}

Combining these two properties, together with the fact that the complexification functor is a linearity-preserving map, we end up with 

\begin{align*}
    \input{FIRST/Structure_theorem_15.tikz}\ &= \input{FIRST/Structure_theorem_16.tikz} = \input{FIRST/Structure_theorem_17.tikz}\\
    &=\input{FIRST/Structure_theorem_18.tikz}=\input{FIRST/Structure_theorem_19.tikz}.
\end{align*}

Next, using Eqs.~\eqref{eq:definition_chi_map} and~\eqref{eq:restricted_map_phi} we show that for every $A$ in $\mathbf{G}$ we  have that $\phi_A$ is the left inverse of $\chi_A$ since 
\begin{widetext}
\begin{align}\label{eq:phi_left_inverse_CA}
    \input{FIRST/Structure_theorem_25.tikz} = \sum_{i,j,i',j'}\mathcolorbox{complex}{t_{ij}t_{i'j'}}\ \input{FIRST/Structure_theorem_94.tikz} = \sum_{i,j,i',j'}\mathcolorbox{complex}{t_{ij}t_{i'j'}}\ \input{FIRST/Structure_theorem_95.tikz}=  \sum_{i,j,i',j'}\mathcolorbox{complex}{t_{ij}t_{i'j'}}\ \input{FIRST/Structure_theorem_96.tikz} =\,  \input{FIRST/Structure_theorem_97.tikz} =  \begin{tikzpicture}
	\begin{pgfonlayer}{nodelayer}
		\node [style=none] (0) at (0, 1) {};
		\node [style=none] (1) at (0, -1) {};
		\node [style=right label] (2) at (0.25, 0) {$\mathds{C}(A)$};
	\end{pgfonlayer}
	\begin{pgfonlayer}{edgelayer}
		\draw [style=ComplexMap] (1.center) to (0.center);
	\end{pgfonlayer}
\end{tikzpicture}
.
\end{align}
\end{widetext}
In the last equation we jointly used Eq.~\eqref{eq:partition_of_unity_GPT} and Eq.~\eqref{eq:complexification_preserves_the_identity}. 

Equation~\eqref{eq:phi_left_inverse_CA} also implies that $\chi_A$ is injective. To see this, we start recalling that the maps $\chi_A$ and $\phi_A$ are $\mathds{C}$-linear maps, implying that $\chi_A(0) = 0$ and $\phi_A(0) =0$; then, for every $a \in \mathds{C}(A)$ we have that $$\chi_A(a)=0 \Leftrightarrow a=\mathrm{id}_{\mathds C(A)}(a)=\phi_A\circ \chi_A(a)=\phi_A(0)=0,$$ meaning that for any $x,y$, $\chi_A(x)=\chi_A(y)$ implies $x=y$ (since $\chi_A(x-y)=0$ due to $\mathds{C}$-linearity). 

Finally, we consider the representation of the identity, 
\begin{equation}
\input{FIRST/Structure_theorem_32.tikz}\ =\ \begin{tikzpicture}
	\begin{pgfonlayer}{nodelayer}
		\node [style=none] (0) at (0, 1) {};
		\node [style=none] (1) at (0, -1) {};
	\end{pgfonlayer}
	\begin{pgfonlayer}{edgelayer}
		\draw [style=FVecComplex] (0.center) to (1.center);
	\end{pgfonlayer}
\end{tikzpicture}\,\,,
\end{equation}
which is a consequence of $\mathsf{M}$ being a functor. Since $\mathds{C}$ is also a functor we find that
\begin{widetext}
\begin{equation}\label{eq: phi_right_inverse_chi}
\input{FIRST/Structure_theorem_37.tikz} \,= \input{FIRST/Structure_theorem_98.tikz} \quad =\quad  \input{FIRST/Structure_theorem_99.tikz}\quad=\quad\input{FIRST/Structure_theorem_100.tikz}  \,\,=\,\, \input{FIRST/Structure_theorem_32.tikz} \,=\,  \begin{tikzpicture}
	\begin{pgfonlayer}{nodelayer}
		\node [style=none] (0) at (0, 1) {};
		\node [style=none] (1) at (0, -1) {};
		\node [style=right label] (2) at (0, -0.75) {$V_A$};
	\end{pgfonlayer}
	\begin{pgfonlayer}{edgelayer}
		\draw [style=FVecComplex] (0.center) to (1.center);
	\end{pgfonlayer}
\end{tikzpicture}

\end{equation}
\end{widetext}
which means that $\phi_A$ is also the right inverse of $\chi_A$, and hence that $\chi_A$ is surjective and we can write $\phi_A = \chi_A^{-1}$.  
\end{proof}

While we use a diagrammatic representation for the maps $\chi_A$ it is important to note that the only diagrammatic elements that have physical relevance are those that are defined within {$\mathbf{G}$}.  That is, $\chi_A$ is \textit{not} necessarily a physical process in general. Moreover, we have used a choice of color scheme that emphasizes the fact that the maps $A \mapsto \chi_A$ are dependent on  $\mathsf{M}$. 

We note importantly that Theorem~\ref{thm: complex_valued_structure_theorem} has an immediate (yet relevant) corollary. 

\begin{corollary}\label{corollary:infinite_dim_functors_must_be_fin}
    Let $\mathbf{G}$ be a tomographically-local finite-dimensional generalized  probabilistic theory. Any linearity-preserving and empirically-adequate functor $\mathsf{M}: \mathbf{G} \to \complexvectexamplep$ necessarily  is a functor onto the subcategory of finite-dimensional complex vector spaces, i.e., $\mathsf{M}: \mathbf{G} \to \fincomplexvectexamplep$. Moreover, for every system $A$ in $\mathbf{G}$ we have  $\dim_{\mathds{C}}(\mathsf{M}(A)) = \mathrm{dim}_{\mathds{C}}(\mathds{C}(A))$. 
\end{corollary}

\begin{proof}
    We merely note that from the proof of Theorem~\ref{thm: complex_valued_structure_theorem} it follows that for every system $A$ the image of $\chi_A$ is the finite-dimensional space give by 
    \begin{equation}
        \mathrm{im}\chi_A = \mathrm{span}_{\mathds{C}}\{\mathsf{M}(s)\mid s \in \mathbf{G}(\mathds{R},A)\}
    \end{equation}
    which is finite-dimensional since $A$ is finite-dimensional and $\mathsf{M}$ is a linearity-preserving map. However, since $\mathsf{M}$ is a functor we have that $\mathsf{M}(\mathrm{id}_A) = \mathrm{id}_{V_A}$ and this implies that $\chi_A$ is also surjective, forcing $\mathsf{M}(A) = V_A$ to be finite-dimensional and equal to $\mathrm{im}\chi_A$, for every $A$. Since the $\mathds{C}$-linear map $\chi_A:\mathds{C}(A)\to \mathsf{M}(A)$ is invertible, the (complex) dimension of its domain $\mathds{C}(A)$ and codomain $\mathsf{M}(A)$ must be the same.
\end{proof}

This corollary is a simple (yet interesting) extension of Corollary 4.2 in Schmid et al.~\cite{schmid2024structure}. In other words, the corollary above shows that a map $\mathsf{M}$ cannot satisfy the theorem’s conditions (empirical adequacy, linearity preservation, and functoriality) while still embedding finite-dimensional spaces $A$ into infinite-dimensional ones $\mathsf{M}(A)$. In Ref.~\cite{schmid2024structure}, $\mathsf{M}$ was assumed to map $\mathbf{G}$ to  $\finrealvectexamplep$ from the start, and it was shown that $\mathsf{M}(A)$ must have a dimension depending on the maps $\chi_A$. This leads to the conclusion: one cannot embed $A$ into a much higher (yet still finite) dimensional space $\mathsf{M}(A)$; the dimension of the two spaces must be the \emph{same}. Our result reaches the same conclusion, but under the assumption that $\mathsf{M}$ map $\mathbf{G}$ to $\complexvectexamplep$ from the start.

We can give  a  physical motivation for Theorem~\ref{thm: complex_valued_structure_theorem}.
Recall that $\fincomplexvectexamplep$ is essentially equivalent to the space $\fincomplexquasisubstoch$ used to define quasiprobability representations of a theory (see Sec.~\ref{sec:examples}).
Thus, the theorem provides a structure theorem for functorial quasiprobabilistic representations. Moreover, following the discussion in Ref.~\cite{schmid2024structure}, when $\mathbf{G}=\finitequantumtheory$ the assignment
$A \mapsto (\chi_A,\chi_A^{-1})$ is closely related to the frame-representation formalism for quantum systems. In particular,  the KD representations of $\finitequantumtheory$ considered in Ref.~\cite{schmid2024kirkwood} can be realized as the action of a functor $\mathsf{KD}_{\mathrm{std}}:\finitequantumtheory\to\fincomplexquasisubstoch$. In Secs.~\ref{sec:main_result_quasiprobs} and~\ref{sec:examples_quasiprobability_quantum_theory} we expand on these physical implications.

\subsection{Semi-functors mapping GPT systems to complex vector spaces}\label{sec:semi_functors_to_complex_vectors_spaces}

We now state a version of Theorem~\ref{thm: complex_valued_structure_theorem} where we relax the assumption of $\mathsf{M}$ being a functor. Instead, we shall consider maps $\mathsf{N}$ that are merely semi-functors.  Recall that the main difference for our scope is that semi-functors need not preserve the identity. 

To state and prove the result, we briefly recall the definition of  the  \emph{surjective corestriction} of a map. For any fixed map $f:X\to Y$ its surjective corestriction is the map $\underline{f}:X \to \mathrm{im}(f)$ such that $ \underline{f}(x):=f(x)$ for all $x \in X$. It is therefore the ``surjective version'' of the map. 

\begin{theorem}[Structure theorem for semi-functors]\label{thm: connectivity_preserving theorem}
    Let  {$\mathbf{G}$} be a tomographically-local finite-dimensional generalized  probabilistic theory. Any linearity-preserving and empirically-adequate semi-functor $\mathsf{N}: \mathbf{G} \to \complexvectexamplep$ can be represented as \begin{equation}\label{eq:connectivity_preserving}
    \input{FIRST/Structure_theorem_38.tikz}\quad = \quad \input{FIRST/Structure_theorem_39.tikz},
    \end{equation}
    where, for each system $A$, $\chi_A$ is an injective  $\mathds{C}$-linear map in $\complexvectexamplep$   
    uniquely determined  by the action of $\mathsf{N}$ on states $s \in \mathbf{G}(\mathds R,A)$ and $\phi_A = \underline{\chi_A}^{-1}\circ \mathsf{N}(\mathrm{id}_A)$, where $\underline{\chi_A}$ is the surjective corestriction of $\chi_A$.
\end{theorem}

\begin{proof}
The proof of this theorem is exactly the same as the one from Theorem~\ref{thm: complex_valued_structure_theorem} up to the point where we show that $\chi_A \circ \phi_A = \mathrm{id}_A$.  Instead, the last part is substituted with 
\begin{equation}\label{eq:idempotent_chi_phi}
    \input{FIRST/Structure_theorem_40.tikz}\quad = \quad \input{FIRST/Structure_theorem_41.tikz}= \quad \input{FIRST/Structure_theorem_42.tikz}= \quad \input{FIRST/Structure_theorem_43.tikz}.
\end{equation}
We now show that for every $A$ the map $\phi_A$ is fixed by the pair $(\chi_A,D_A)$. To see this, recall that $\chi_A:\mathds{C}(A) \to V_A$ is injective, which implies that its surjective corestriction  $\underline{\chi_A}:\mathds{C}(A) \to \mathrm{im}\chi_A$ is invertible and therefore (since both $D_A$ and $\chi_A$ have the same image)
\begin{equation}
    D_A = \chi_A\circ \phi_A \implies \underline{\chi_A}^{-1}\circ D_A =  \phi_A,
\end{equation}
as we wanted. This concludes the proof, since it shows that $D_A$ and $\chi_A$ uniquely fix $\phi_A$.
\end{proof}

We can comment on the formal notion of \emph{uniqueness} that applies to the maps $\chi_A$ and $\phi_A$. We recall that $\chi_A$ is uniquely (up to isomorphisms) defined by the action of the (semi)functor. The map $\phi_A$ is also uniquely specified by $\chi_A$ and $D_A$ in a similar manner.  Recall that if a given idempotent map $D_A:V_A\to V_A$ in $\complexvectexamplep$ splits $$D_A = \chi_A\circ \phi_A,$$ in which case we call $(\chi_A,\phi_A)$ a splitting for $D_A$, then any such splitting is \emph{unique up to a unique isomorphism}. In other words, given any two splittings $$D_A = \chi_A\circ \phi_A  = \chi_A'\circ \phi_A'$$ there is only one isomorphism mapping these two choices, i.e. there is a single map $\xi_A$ such that $\chi_A'\circ \xi_A = \chi_A$ and $\xi_A \circ \phi_A=\phi_A'$. Explicitly, the map is given by 
\begin{equation}
\xi_A = \phi_A'\circ \chi_A
\end{equation}
and its inverse by $\xi_A^{-1} = \phi_A\circ \chi_A'$. Hence, provided that $\chi_A$ and $D_A$ are fixed for every $A$, the map $\phi_A$ is fixed as well in the above sense.

We remark a distinction between Theorems~\ref{thm: complex_valued_structure_theorem} and~\ref{thm: connectivity_preserving theorem} concerning the mappings $A \mapsto \chi_A$. In the case of Theorem~\ref{thm: complex_valued_structure_theorem}, the maps $\chi_A$ are invertible $\mathds{C}$-linear maps $\chi_A: \mathds{C}(A) \to \mathsf{M}(A)$, which implies that $\chi_A$ must be surjective and that $\mathsf{M}(A)$ is finite-dimensional (as discussed in Corollary~\ref{corollary:infinite_dim_functors_must_be_fin}). Hence, $\chi_A$ is a map in $\fincomplexvectexamplep$ for every system $A$. In contrast, in Theorem~\ref{thm: connectivity_preserving theorem} we only have that $\chi_A$ is injective, which implies that, in principle, for some maps $\mathsf N$ one may have $\chi_A:\mathds C(A) \to \mathsf{N}(A)$ surjective with $\mathsf{N}(A)$ infinite-dimensional, and thus a map in $\complexvectexamplep$ \emph{but not} in $\fincomplexvectexamplep$.

Theorem~\ref{thm: connectivity_preserving theorem} allows us to draw the following conclusion: for \emph{functors} $\mathsf M$ we obtain that the maps $\chi_A$ are uniquely characterized by their action on states up to the choices of isomorphisms identifying the vector spaces $A$ and $\mathcal{L}(\mathds R,A)$ (and similarly for the complex vector spaces). If we fix this choice to be given by the canonical isomorphisms, there is no ambiguity in the choices for the maps $\chi_A$. For \emph{semi-functors} $\mathsf N$ the same continues to hold for $\chi_A$, but now we have additionally that $\phi_A\neq \chi_A^{-1}$ but is fixed by the action of $\mathsf N$ on the identity and by the idempotent splitting, which implies that $\phi_A$ is unique up to a unique isomorphism. It is also simple to see that, when $\mathsf{N}$ is a functor and not only a semi-functor we recover Theorem~\ref{thm: complex_valued_structure_theorem} since in that case every map $\chi_A$ becomes invertible (hence equal to its surjective corestriction) and $\phi_A = \chi_A^{-1}\circ \mathsf{N}(\mathrm{id}_A) = \chi_A^{-1}\circ \mathrm{id}_{V_A}  =\chi_A^{-1}$.

\subsection{Quasiprobability representations of GPTs}\label{sec:main_result_quasiprobs}

With Theorems~\ref{thm: complex_valued_structure_theorem} and~\ref{thm: connectivity_preserving theorem} we can now state our results for the subclass of maps from  {$\mathbf{G}$} to $\fincomplexquasisubstoch$ which constitute our definition for a complex-valued finite-dimensional quasiprobability representation. 

\begin{corollary}[Structure theorem for complex-valued representations]\label{corollary: quasiprobability representations structure theorem}
    A complex-valued quasiprobability representation  {$\mathsf{Q}: \mathbf{G} \to \fincomplexquasisubstoch$} (as in Def.~\ref{def:quasiprobability_representation}) of a tomographically-local finite-dimensional generalized probabilistic theory can always be written as
    \begin{equation}
        \input{FIRST/Structure_theorem_44.tikz} = \input{FIRST/Structure_theorem_45.tikz},
    \end{equation}
    where, for each system $A$, $\chi_A$ is an injective $\mathds C$-linear map in $\fincomplexvectexamplep$ uniquely determined by the action of $\mathsf Q$ on states $s \in \mathbf{G}(\mathds R, A)$ and $\phi_A = \underline{\chi_A}^{-1}\circ \mathsf{Q}(\mathrm{id}_A)$, where $\underline{\chi_A}$ is the surjective corestriction of $\chi_A$. Furthermore, if $\mathsf{Q}$ is  discard preserving  we have: 
    \begin{equation}\label{eq: quasiprob_chi_deterministic}
        \input{FIRST/Structure_theorem_46.tikz} \quad =\quad  \input{FIRST/Structure_theorem_47.tikz} \quad = \quad \input{FIRST/Structure_theorem_48.tikz}.
    \end{equation} 
\end{corollary}

\begin{proof}
    The only non-trivial part that demands proof and that does not immediately follow from Theorems~\ref{thm: complex_valued_structure_theorem} and~\ref{thm: connectivity_preserving theorem} is the statement present in Eq.~\eqref{eq: quasiprob_chi_deterministic}. This follows from the fact that, for every $A$,
    \begin{align*}
        \input{FIRST/Connectivity_preserving_sec_31.tikz} \quad \stackrel{\eqref{eq:determinist_effect_preservation}}{=} \quad  \input{FIRST/Structure_theorem_49.tikz} \quad=\quad \input{FIRST/Structure_theorem_50.tikz}
    \end{align*}
    which then implies that 
    \begin{align*}
        \input{FIRST/Structure_theorem_52.tikz}\quad =\quad \input{FIRST/Structure_theorem_53.tikz}\quad =\quad \input{FIRST/Structure_theorem_54.tikz}
    \end{align*}
where the first equation follows from Eq.~\eqref{eq:phi_left_inverse_CA}. 
\end{proof}

We can now take the opportunity to comment on our structure theorem and relate it to previous work. From our perspective, a key requirement for any quasiprobability representation of a GPT is that it respects the compositional structure of the theory. In particular, representations should not only assign quasiprobabilities to individual states and effects but must also  properly reflect how these objects compose in the theory. This requirement is crucial for ensuring that the representation captures the full operational content of the theory, rather than just reproducing individual probabilities in isolation. In Ref.~\cite{schmid2024kirkwood}, we introduced a class of functorial quasiprobability representations---specifically, those based on the Kirkwood–Dirac (KD) construction---which explicitly preserved the compositional and monoidal structure of quantum theory. In the next section we briefly revisit this construction. In the present work, we go beyond this specific construction: we prove that \emph{any} complex-valued quasiprobability representation (not only one constructed out of KD representations) of a tomographically-local finite-dimensional GPT (not only quantum theory) obeys the structure theorem established in Theorems~\ref{thm: complex_valued_structure_theorem} and~\ref{thm: connectivity_preserving theorem}.  

This result also generalizes the earlier finding from Ref.~\cite{schmid2024structure} that functorial representations are uniquely determined by their action on states. When the functoriality condition is relaxed to semi-functoriality, the action on states alone is \emph{no longer} sufficient to fix the representation: one must also specify how the identity is represented. 

\subsection{Quasiprobability representations of quantum theory}\label{sec:examples_quasiprobability_quantum_theory}

We can now discuss how the mappings $A \mapsto \chi_A, \phi_A$, which assign to each system a pair of $\mathds{C}$-linear maps, relate to the mappings $\mathcal{B}(\mathcal{H}) \mapsto F,G$, which assign to each quantum system a frame and its dual, providing any quantum system with a complex-valued frame representation. Consider a quasiprobability representation of quantum theory, viewed as a generalized probabilistic theory. This corresponds to a linearity-preserving semi-functor $$\mathsf{Q}:\FinQT \to \fincomplexquasisubstoch.$$ Recall from Eq.~\eqref{eq:complexification_self_adjoint} that $\mathds{C}(\mathcal{B}(\mathcal{
H})_{\rm sa}) \cong \mathcal{B}(\mathcal{H})$. By the structure theorem~\ref{corollary: quasiprobability representations structure theorem}, to each system $\mathcal{B}(\mathcal{H})_{\rm sa}$ there is a unique choice (dependent on the action of $\mathsf{Q}$ on quantum (unnormalized) states $\rho \in \mathcal{L}(\mathds{R},\mathcal{B}(\mathcal{H})^+)\cong \mathcal{B}(\mathcal{H})^+$) of a map $\chi_{\mathcal{H}}:\mathcal{B}(\mathcal{H}) \to \mathds{C}^\Lambda$ such that the quasiprobability representation of any state is given by
\begin{equation*}
    \mathsf{Q}(\rho) \quad =\quad  \input{FIRST/Structure_theorem_57.tikz} \quad = \quad \input{FIRST/Structure_theorem_58.tikz}.
\end{equation*}
Using the identity $\mathrm{id}_{\mathds{C}^\Lambda} = \sum_{\lambda} \vert \lambda \rangle \rangle  \langle \langle \lambda \vert$ (we explain this notation below) we end up with 
\begin{equation*}
    \mathsf Q (\rho) = \input{FIRST/Structure_theorem_58.tikz} \,= \, \input{FIRST/Structure_theorem_59.tikz}.
\end{equation*}

As a remark, we are following here a notion used by other references as well, e.g. Ref.~\cite{burkat2025structurepositivityclassicalsimulability}, where $\vert \lambda \rangle \rangle$  emphasizes that this is a vectorized representation of a matrix basis. For faithful non-overcomplete representations one has $|\Lambda| \blk = d^2 > d$, hence we are distinguishing vectors in $\mathcal{H} \simeq \mathds{C}^d$ (denoted $\vert \lambda \rangle$) from those in $\mathds{C}^{d^2}$ (denoted $\vert \lambda \rangle \rangle$). If we recall the description of a frame relative to the Hilbert--Schmidt inner product
\begin{align}
\mu(\lambda \vert \rho) = \langle F_\lambda,\rho \rangle_{\mathrm{HS}} = \mathrm{Tr}(F_\lambda^\dagger\rho)
\end{align}
we can then conclude that 
\begin{equation*}
    \mathsf Q(\rho) =\,\input{FIRST/Structure_theorem_60.tikz}\,=\,\sum_{\lambda \in \Lambda} \mu(\lambda \vert \rho)\vert \lambda \rangle \rangle.
\end{equation*}
Here, $\mathsf{Q}(\rho)$ is described as a vector in $\mathds{C}^\Lambda$ for a finite index set $\Lambda$.  

Therefore, using $\chi_\mathcal{H}: \mathds{C}(\mathcal{B}(\mathcal{H})_{\rm sa}) \to \mathds{C}^\Lambda$ we have defined $F_\lambda : \mathcal{B}(\mathcal{H}) \to \mathds{C}$ via 
\begin{align}
    F_\lambda \simeq \langle F_\lambda, \cdot  \rangle_{\rm HS} := \langle \langle \lambda \vert \chi_{\mathcal{H}}
\end{align}
where we have used the identification $\mathcal{B}(\mathcal{H}) \cong \mathds{C}(\mathcal{B}(\mathcal{H})_{\rm sa})$, so that 
\begin{equation}
    \langle \langle \lambda \vert \chi_\mathcal{H} \mathds{C}(\rho) = \langle F_{\lambda},\mathds{C}(\rho)\rangle_{\rm HS}.
\end{equation}
Recall, moreover, that $\mathds{C}(\rho) = \rho +i\cdot 0$. 

We can proceed similarly for the effects, and conclude that 
\begin{align}
\mathsf{Q}(E) &= \input{FIRST/Structure_theorem_61.tikz}\\
&=\sum_{\lambda}\xi(E\vert \lambda) \langle \langle \lambda \vert. 
\end{align}

Therefore, using $\phi_{\mathcal{H}}:\mathds{C}^\Lambda \to \mathds{C}(\mathcal{B}(\mathcal{H})_{\rm sa})$ we have defined the operators
\begin{equation}
    G_\lambda = \phi_\mathcal{H}\vert \lambda \rangle \rangle . 
\end{equation}

We now proceed to show that $F_\lambda$ and $G_\lambda$ are both a pair of frame and dual frame for $\mathcal{B}(\mathcal{H})$. To see this, we start noticing that $\{F_\lambda\}_\lambda$ as defined spans $\mathcal{B}(\mathcal{H})^*$ and $\{G_\lambda\}_\lambda$ spans $\mathcal{B}(\mathcal{H})$. This easily follows from the fact that, for every  $\mathcal
B (\mathcal{H})$ we have that $\phi_\mathcal{H}$ is the left inverse of $\chi_\mathcal{H}$. Let $A \in \mathcal{B}(\mathcal{H}) = \mathds{C}(\mathcal{B}(\mathcal{H})_{\rm sa})$ be any, we have that 
\begin{equation}
    \begin{tikzpicture}
	\begin{pgfonlayer}{nodelayer}
		\node [style=copoint_complex] (0) at (0, 0.5) {$A$};
		\node [style=none] (1) at (0, -0.5) {};
		\node [style=none] (2) at (0.75, -0.2) {\footnotesize$\mathcal{B}(\mathcal{H})$};
	\end{pgfonlayer}
	\begin{pgfonlayer}{edgelayer}
		\draw [style=ComplexMap] (0) to (1.center);
	\end{pgfonlayer}
\end{tikzpicture}
 \quad=\quad \input{FIRST/Structure_theorem_63.tikz} \quad = \quad \sum_{\lambda \in \Lambda} \,\,\input{FIRST/Structure_theorem_64.tikz}
\end{equation}
from which we conclude that $F_\lambda$ spans $\mathcal{B}(\mathcal{H})^*$. Similarly, we can conclude that $G_\lambda$ spans $\mathcal{B}(\mathcal{H})$. 

In fact, the calculations above conclusively show that $F$ and $G$ are a pair of frame and dual frame. When $\mathsf{Q}$ is a functor (as opposed to being a semi-functor) we have that $F$ and $G$ are nonovercomplete frames, as pointed out by Ref.~\cite{schmid2024structure}. Otherwise, we can also show that there is a relationship between the idempotents $D_{\mathcal{H}} := \mathsf{Q}(\mathrm{id}_{\mathcal{B}(\mathcal{H})})$ and the pair of frame and dual frame just defined. In fact, we have that these are given by

\begin{equation}
    \input{FIRST/Structure_theorem_65.tikz}\quad = \quad \input{FIRST/Structure_theorem_66.tikz}\quad =\quad \sum_{\lambda,\lambda'}\input{FIRST/Structure_theorem_67.tikz},
\end{equation}
and therefore we conclude that 
\begin{equation}
    D_{\mathcal{H}} = \mathsf{Q}(\mathrm{id}_{\mathcal{B}(\mathcal{H})}) = \sum_{\lambda,\lambda'}\langle F_{\lambda'},G_\lambda\rangle_{\rm HS} \vert \lambda'\rangle \rangle \langle  \langle \lambda\vert.
\end{equation}
From this relation it is clear that whenever $\mathsf{Q}$ is a functor (and not a semi-functor) we have $\phi_{\mathcal{H}} = \chi_{\mathcal{H}}^{-1}$, the frame and dual frame pair becomes a basis and cobasis pair, and the idempotent becomes $$D_{\mathcal{H}} = \sum_{\lambda,\lambda'}\delta_{\lambda,\lambda'}\vert \lambda' \rangle \rangle \langle \langle \lambda \vert= \sum_{\lambda}\vert \lambda \rangle \rangle \langle \langle \lambda \vert = \mathrm{id}_{\mathds{C}^\Lambda}.$$

This is precisely the case for the family of standard KD representations introduced in Ref.~\cite{schmid2024kirkwood}. As shown there, it is possible to extend the KDQs $\mu_{\rm KD}$ as a distribution of states to a functorial, empirically adequate, diagram-preserving representations of quantum theory. To do so, we define functors $$\mathsf{KD}_{\rm std}:\finitequantumtheory \to \fincomplexquasisubstoch$$ via the following specification: to every Hilbert space $\mathcal H$ take the standard basis $\mathbb{a}_{\mathcal H}:=\{\vert e_i \rangle \}_{i=1}^d$ and a rotated basis $\mathbb{b}_{\mathcal{H}} = V \mathbb{a}_{\mathcal H} := \{V \vert e_i \rangle \}_{i=1}^d$ such that $\langle a \vert b \rangle  \neq 0$ for every $(a,b) \in \mathbb{a}_{\mathcal H}\times \mathbb{b}_{\mathcal H}$. Each choice of $V$ defines a different $\mathsf{KD}_{\rm std}$ (exactly as it defines a different KDQ distribution).  Relative to that Hilbert space, define pairs of frame $F^{(\mathcal{H})}$ and dual frames $G^{(\mathcal{H})}$ as by Eqs.~\eqref{eq:KD_frame} and~\eqref{eq:KD_dual_frame}, respectively. Then, the action of $\mathsf{KD}_{\rm std}$ on systems is given by $$\mathcal{B}(\mathcal{H})_{\rm sa} \mapsto \mathds{C}^{\mathbb{a}_{\mathcal{H}}\times \mathbb{b}_{\mathcal{H}}}\cong \prod_{(a,b)\in \mathbb{a}_{\mathcal{H}}\times \mathbb{b}_{\mathcal{H}}}\mathds{C},$$ and the action of it on generic maps is defined by 
\begin{equation}
    \input{FIRST/Structure_theorem_68.tikz}\, := \sum_{a,b,a',b'}\input{FIRST/Structure_theorem_69.tikz}.
\end{equation}

As showed by Ref.~\cite{schmid2024kirkwood}, any such map $\mathsf{KD}_{\rm std}$ is indeed a functor, that is moreover empirically adequate and linearity-preserving. 

\section{Conclusion and outlook}\label{sec:conclusion}

In this work, we have developed a structural and diagrammatic framework to represent generalized probabilistic theories, particularly focusing on quasiprobability representations. Our approach emphasizes the compositional and categorical aspects of physical theories, leveraging the language of process theories and their diagrammatic calculus, and shows that any such quasiprobability representation needs to satisfy a structure theorem which extends the one described by Schmid et al.~\cite{schmid2024structure} in three different directions: first, we obtain a structure theorem that does not require diagram-preservation, but only the weaker form of representation provided by semi-functorial maps; second, we allow for infinite-dimensional representations; third, we consider \emph{complex-valued} representations. For example, our structure theorem applies to every faithful Kirkwood-Dirac quasiprobability representation as the one constructed in Ref.~\cite{schmid2024kirkwood}. 

More specifically, the structure theorem we present shows that any complex-valued quasiprobability representation of a finite-dimensional, tomographically-local GPT has a simple and constrained mathematical form---for each system, every such representation can be expressed by a pair of complex-linear maps that are completely characterized by their action on states and on the identity (equivalently, on effects) for that system. Of course, the case of primary interest is finite-dimensional quantum theory which serves as a paradigmatic example of a finite-dimensional, tomographically-local GPT. In this case, the pair of maps for a system simply correspond to a choice of frame and dual frame for the system. 

Beyond the results presented here, several compelling directions remain open for future investigation. For example, it would be interesting to explore extensions of our framework to encompass more exotic algebraic settings, such as those arising from the \emph{quaternionification} of vector spaces and therefore considering representations provided by the quaternions (as opposed to the complex or real fields).  The study of quantum theory over such fields---real, complex, and quaternion numbers---from a categorical perspective can be found in Ref.~\cite{baez2011division}. The most relevant challenge is that quaternionic vector spaces do not form a symmetric monoidal category, that is, currently  there is no good notion of the tensor product of two quaternionic vector spaces. While this would not be an obstruction to the structure theorem that we prove here which did not rely on the monoidal structure, it would be an obstruction towards using such a quaternionic representation when considering composite systems. 

Second, our work builds on the well-established identification of process theories with symmetric monoidal categories. However, recent advances~\cite{selby2025generalisedprocesstheories} have pointed to the possibility of formulating process theories in a broader and more operationally grounded manner, dispensing with some of the more rigid categorical assumptions. Investigating whether the structural features we analyze here persist in these more general settings---possibly beyond the reach of traditional symmetric monoidal categories, using the formalism of \emph{operad algebras}.

Third, a natural setting where our extension of semi-functorial maps $\mathsf{N}:\mathbf{G}\to\complexvectexamplep$ arises is in bosonic quantum error-correcting codes~\cite{albert2025bosoniccoding}: the encoding map embeds a discrete, finite-dimensional system into a continuous, infinite-dimensional bosonic Hilbert space. From the quasiprobability-representation viewpoint (cf.~Ref.~\cite{davis2024identifyingquantumresourcesencoded}) this encoding is precisely a map from the subprocess theory of $\finitequantumtheory$ determined by the code spaces into $\realvectexamplep$, and studying it therefore leads naturally to infinite-dimensional representations of originally finite systems. This is illustrated in Ref.~\cite{davis2024identifyingquantumresourcesencoded}, where the authors construct code-adapted quasiprobability representations (extending methods of Ref.~\cite{brief1999phasespace}) that preserve key code symmetries and enable a concrete study of nonclassicality (i.e.,  the need for negativity in the real-valued representation) in error-corrected computation. Their Zak---Gross Wigner representation (when restricted to the code space subtheory) is indeed a linearity-preserving and empirically-adequate semi-functorial map with idempotents provided by twirling maps~\cite{davis2024identifyingquantumresourcesencoded}. It would be interesting to develop this connection further and investigate whether our approach---combined with the insights of Ref.\cite{davis2024identifyingquantumresourcesencoded}---can provide a broadly applicable phase-space framework for quantum error correction, and clarify links between nonclassicality in fault-tolerant computation and negativity in quasiprobability distributions. 

Fourth, while our formulation has focused on finite-dimensional systems, an important generalization is to extend the framework to infinite-dimensional Hilbert spaces, i.e., considering semifunctors $\mathsf{M}: \mathbf{QT}\to \mathbf{Vect}_{\mathds{C}}$. Such an extension would be essential to connect our categorical and diagrammatic results with the infinite-dimensional systems used in quantum optics. Relaxing the structure theorem to semi-functorial maps is a first step towards understanding if similar results would apply to quasiprobability representations that notably do not satisfy the property of being described by a functor, such as those from Glauber--Sudarshan~\cite{glauber1963coherent,sudarshan1963equivalence}, and Husimi~\cite{husimi1940some} representations.

\

\begin{acknowledgements}
    We would like to thank Koenraad Audenaert, Rui Soares Barbosa, Jack Davis, Marco Erba, and Matt Wilson for useful discussions. 
    RW acknowledges support from the European Research Council (ERC) under the European Union's Horizon 2020 research and innovation programme (grant agreement No. 856432, HyperQ). JHS was funded by the European Commission by the QuantERA project ResourceQ under the grant agreement UMO-2023/05/Y/ST2/00143. RDB acknowledges support by the Digital Horizon Europe project FoQaCiA, Foundations of quantum computational advantage, GA No. 101070558, funded by the European Union, NSERC (Canada), and UKRI (UK). RDB, YY and DS are supported by Perimeter Institute for Theoretical Physics. Research at Perimeter Institute is supported in part by the Government of Canada through the Department of Innovation, Science and Economic Development and by the Province of Ontario through the Ministry of Colleges and Universities. YY is also supported by the Natural Sciences and Engineering Research Council of Canada (Grant No. RGPIN-2024-04419). MS is supported by National Science Centre, Poland (Preludium, Classicality and compositionality of general probabilistic theories in infinite dimensions, project no.2024/53/N/ST2/01192).
\end{acknowledgements}

\bibliography{bib}

\begin{appendix}

\section{Proof that the complexification functor preserves generating sets}\label{app:span_proofs}

We denote by $\mathrm{span}_{\mathds{K}}(\{s_i\}_i)$ the \emph{algebraic} linear span of the set $\{s_i\}_i\subseteq A$, i.e., the set of all finite $\mathds{K}$-linear combinations of the elements $s_i$. 

\begin{lemma}\label{lemma:span_preservation_states}
    Let $A$ be a real vector space and $\{s_i\}_i \subseteq A$ such that $\mathrm{span}_{\mathds{R}}(\{s_i\}_i) = A$. Then, $\mathrm{span}_{\mathds{C}}(\{\mathds{C}(s_i)\}_i) = \mathds{C}(A)$.
\end{lemma} 


\begin{proof}
    The inclusion $\mathrm{span}_{\mathds{C}}(\{\mathds{C}(s_i)\}_i) \subseteq \mathds{C}(A)$ is trivial. Suppose now that $a \in \mathds{C}(A)$ which implies that $a = (a_1,a_2)$ such that $a_1,a_2 \in A$. Since $\mathrm{span}_{\mathds{R}}(\{s_i\}_i) = A$ there exists $\{\alpha_1^i\}_i,\{\alpha_2^i\}_i \subseteq \mathds{R}$ such that $a_k = \sum_i \alpha_k^i s_i$ for $k =1,2$. Therefore, 
    \begin{align*}
        a &= (a_1,a_2) = e_{\mathds{C}}(a_1)+ie_{\mathds{C}}(a_2)\\
        &= \sum_j \alpha_1^j e_{\mathds{C}}(s_j)+i\sum_j \alpha_2^j e_{\mathds{C}}(s_j)\\
        &=\sum_j (\alpha_1^j+i\alpha_2^j)e_{\mathds{C}}(s_j)\\
        &=\sum_j \alpha_j \mathds{C}(s_j),
    \end{align*}
    where $\alpha_j\in\mathbb{C}$. 
    From the above we conclude that $a \in \mathrm{span}_{\mathds{C}}(\{\mathds{C}(s_i)\}_i)$.  
\end{proof}

Note that above we have not restricted $A$ to be finite-dimensional. The same holds true for covectors (immediately for the case of finite-dimensional spaces).

\begin{lemma}\label{lemma:span_preservation_effects}
    Let $A$ be a real vector space and $A^*$ its dual space. Let $\{e_i\}_i$ be a set of linear functionals $e_i:A \to \mathds{R}$ such that $\mathrm{span}_{\mathds{R}}(\{e_i\}_i) = A^*$. Then, $\mathrm{span}_{\mathds{C}}(\{\mathds{C}(e_i)\}_i) = \mathds{C}(A)^*$.
\end{lemma}

\begin{proof}
    The inclusion $\mathrm{span}_{\mathds{C}}(\{\mathds{C}(e_i)\}_i) \subseteq \mathds{C}(A)^*$ is trivial. Denote $\mathcal{L}_{\mathds{R}}(A,\mathds C)$ the set of $\mathds{R}$-linear maps from $A$ to $\mathds{C}$, i.e. all maps $\phi: A\to \mathds{C}$ satisfying \begin{align}\phi(\alpha a_1+\beta a_2) &= e_{\mathds{C}}(\alpha)\phi(a_1)+e_{\mathds{C}}(\beta)\phi(a_2)\label{eq:linearity_real_to_complex} \\&\equiv \alpha \phi(a_1)+\beta \phi(a_2).\nonumber\end{align} Consider the restriction map $r:\mathds{C}(A)^* \to \mathcal{L}_{\mathds{R}}(A,\mathds C)$ defined by $r(\Phi) = \Phi\vert_A = \phi$. From Theorem~\ref{theorem: unique_complex_extension} we know that $r(\Phi)$ uniquely defines $\Phi$. For any $\Phi \in \mathds{C}(A)^*$, we  decompose $r(\Phi) = \phi$ into its real and imaginary parts $\phi = \phi_1+i\phi_2$. In this way, $\phi_1,\phi_2 \in A^*$, from which we have that $\phi_k = \sum_j \alpha_k^j e_j$ for $k=1,2$. Therefore, 
    \begin{equation*}
        \phi = \sum_j \alpha_1^je_j+i\sum_j\alpha_2^je_j = \sum_j\alpha_je_j
    \end{equation*}
    where $\alpha_j = \alpha_1^j+i\alpha_2^j \in \mathds C$ for all $j$. Clearly, by construction, we have that $\mathds{C}(r(\Phi)) = \mathds{C}(\phi) = \Phi$. Defining
    \begin{equation*}
        \Psi := \sum_j \alpha_j \mathds{C}(e_j),
    \end{equation*}
    it is simple to see that $\Phi = \Psi \in \mathrm{span}_{\mathds{C}}(\{\mathds{C}(e_i)\}_i)$ from which we conclude that $\mathds{C}(A)^* \subseteq \mathrm{span}_{\mathds{C}}(\{\mathds{C}(e_i)\}_i)$. 
\end{proof}

Note that if $A$ is a finite-dimensional real vector space, then $\mathds{C}(A^*)\simeq \mathds{C}(A)^*$. 

\section{Proof that the complexification functor is a faithful strong monoidal functor}\label{app: complexification proof}

We now proceed to give a detailed proof that the complexification functor $\mathds{C}: \realvectexamplep \to \mathbf{Vect}_{\mathds{C}}$ is indeed a faithful strong monoidal functor. 

We start noting that by the very nature of the complexification functor, we have that $\mathds{C} \cong \mathds{R}+i\mathds{R} \equiv \mathds{R}_{\mathds{C}}$, from which we have that the identity object (a.k.a. the monoidal unit) $\mathds{R}$\footnote{Note that here we take the monoidal structure  (a.k.a. parallel composition) of $\mathbf{Vect}_{\mathds{K}}$ to be given by the tensor product of vector spaces, from which we conclude that the identity object in this case is given by the space $\mathds{K}\in\{\mathds{R},\mathds{C}\}$.} in $\realvectexamplep$ is sent to the identity object $\mathds{C}$ in $\mathbf{Vect}_{\mathds{C}}$. We now proceed to show in detail that $\mathds{C}$ is a faithful strong monoidal functor:
    
\subsection{The complexification functor is faithful}
    
    Let any $f,g \in \mathbf{Vect}_{\mathds{R}}(W,V)$ then 
    \begin{align*}
        &\mathds{C}(f) = \mathds{C}(g) \iff \forall w \in \mathds{C}(W),\mathds{C}(f)(w) = \mathds{C}(g)(w) \\
        &\iff 
        \forall (w_1,w_2),(f(w_1),f(w_2)) = (g(w_1),g(w_2))\\
        & \iff \forall w, f(w)=g(w),\\
        & \iff f=g.
    \end{align*}
    Note that this does not mean that for every $f$ we must have that $\mathds{C}(f)$ is injective. Diagrammatically, this shows that 
    \begin{equation}
        \begin{tikzpicture}
	\begin{pgfonlayer}{nodelayer}
		\node [style=processes] (0) at (0, 0) {$f$};
		\node [style=none] (1) at (0, 1) {};
		\node [style=none] (2) at (0, -1) {};
	\end{pgfonlayer}
	\begin{pgfonlayer}{edgelayer}
		\draw (1.center) to (2.center);
	\end{pgfonlayer}
\end{tikzpicture}
=\begin{tikzpicture}
	\begin{pgfonlayer}{nodelayer}
		\node [style=processes] (0) at (0, 0) {$g$};
		\node [style=none] (1) at (0, 1) {};
		\node [style=none] (2) at (0, -1) {};
	\end{pgfonlayer}
	\begin{pgfonlayer}{edgelayer}
		\draw (1.center) to (2.center);
	\end{pgfonlayer}
\end{tikzpicture}
 \iff \input{FIRST/App_3.tikz}=\input{FIRST/App_4.tikz}.
    \end{equation}
    
    \subsection{Defining coherent maps $\varepsilon$ and $\mu_{W,V}$}
    
    We need to construct natural isomorphisms $\varepsilon$ and $\mu_{V,W}$ for every $V,W$ satisfying certain coherence conditions. Recall that the monoidal units are $I_{\mathds{K}} = \mathds{K}$. We define the morphism $\varepsilon: I_{\mathds{C}} \to \mathds{C}(I_{\mathds{R}})$ via 
    \begin{equation}
        \varepsilon(x+iy) = (x,y),
    \end{equation}
    which is clearly $\mathds{C}$-linear. 
    
    Let $W$ and $V$ be any pair of objects of $\realvectexamplep$, we can define $\hat\mu_{W,V}: \mathds{C}(W) \otimes \mathds{C}(V) \to \mathds{C}(W \otimes V)$ as
    \begin{align*}
        &\hat\mu_{W,V}((w_1,w_2) \otimes (v_1,v_2)) =\\
        &= (w_1 \otimes v_1-w_2 \otimes v_2, w_1 \otimes v_2+w_2 \otimes v_1).
    \end{align*}
    Note that, stricly speaking, the two tensor products in the definition of $\hat{\mu}$ are different, and that to avoid an unnecessarily heavy notation we simply write $\otimes$ for both, instead of writing $\otimes_{\mathds{R}}$ and $\otimes_{\mathds{c}}$. The intuition is, obviously, that we want the action of $\hat\mu$ on $(w_1, w_2)\otimes (v_1 , v_2)$ to mimic 
    the distributive property 
    \begin{align*}
    (w_1+i w_2)\otimes (v_1 +iv_2) &= w_1\otimes v_1-w_2\otimes v_2\\&+i\left(w_2\otimes v_1 + w_1 \otimes v_2 \right)
    \end{align*}
    between the tensor product and the summation. We then extend this map to a map  $\mu_{W,V}$ that is $\mathds{C}$-linear, i.e., we let $\mu_{W,V}$ to be the unique $\mathds{C}$-linear extension of $\hat \mu_{W,V}$, for every pair of vector spaces $W$ and $V$. 

    \subsection{Showing that $\varepsilon$ and $\mu_{W,V}$ are isomorphisms}
    
    Clearly, $\varepsilon$ is invertible, with inverse given by  $\varepsilon^{-1}(x,y) = x+iy.$ Now, let $V,W \in \realvectexamplep$, we want to show that $\mu_{W,V}$ is invertible. Let $\{w_i\}_i \subseteq W, \{v_i\}_i \subseteq V$ be arbitrary basis of these vector spaces, then $\{(w_i \otimes v_j,0),(0,w_i \otimes v_j)\}_{i,j}$ is a basis for $\mathds{C}(W \otimes V)$. We will define the action of the inverse map on these, and then extend it $\mathds{C}$-linearly over the entire space $\mathds{C}(W \otimes V)$. 
    
    Define the map $\hat\mu_{W,V}^{-1}:\mathds{C}(W \otimes V) \to \mathds{C}(W)\otimes \mathds{C}(V)$ via its action on the basis elements 
    \begin{equation}
        \hat\mu_{W,V}^{-1}((w_i \otimes v_j,0)) := (w_i,0)\otimes (v_j,0),
    \end{equation}
    since in this case we have that 
    \begin{align*}
        &\hat\mu_{W,V}^{-1}\circ \mu_{W,V}((w_i,0)\otimes (v_j,0))= \\&= \hat\mu_{W,V}^{-1}((w_i\otimes v_j-0\otimes 0,w_i\otimes 0+0\otimes v_j))\\
        &=\hat\mu_{W,V}^{-1}((w_i\otimes v_j,0))\\
        &=(w_i,0) \otimes (v_j,0),
    \end{align*}
    and also 
    \begin{equation}
        \hat\mu_{W,V}^{-1}(0,w_i \otimes v_j) := (w_i,0) \otimes (0,v_j) + (0,w_i)\otimes (v_j,0),
    \end{equation}
    since we have moreover that,
    \begin{align*}
        &\hat\mu_{W,V}^{-1}\circ \mu_{W,V}((0,w_i)\otimes (0,v_j))=\\
        &=\hat\mu_{W,V}^{-1}(0 \otimes 0-w_i \otimes v_j,0 \otimes v_j+w_i \otimes 0)\\
        &=\hat\mu_{W,V}^{-1}(-w_i \otimes v_j,0)\\
        &=i(w_i,0)\otimes i(v_j,0)\\
        &=(0,w_i)\otimes (0,v_j),
    \end{align*}
    where we have used that $i^2=-1$ and $i(a,0)=(0,a)$. 

    Now, we take $\mu_{W,V}^{-1}$ to be the unique $\mathds{C}$-linear extension of $\hat \mu_{W,V}^{-1}$, in which case we have that the extension is the inverse of $\mu_{W,V}$ from which we conclude that $\mu_{W,V}$ is an isomorphism. 
    
    \subsection{$\mu_{W,V}$ is a natural transformation}
    
    We need to show that for every pair of $\mathds{R}$-linear maps $f: W \to W', g: V \to V'$ we have that  
    \begin{equation}
        \mathds{C}(f \otimes g) \circ \mu_{W,V} = \mu_{W',V'}\circ (\mathds{C}(f) \otimes \mathds{C}(g)).
    \end{equation}
    On one side, we have that, for any element $\sum_i \alpha_i (w_i^1,w_i^2)\otimes (v_i^1,v_i^2) \in \mathds{C}(W)\otimes \mathds{C}(V)$ it holds that  
    \begin{align*}
        &\mu_{W',V'}( (\mathds{C}(f) \otimes \mathds{C}(g))(\sum_i \alpha_i((w_i^1,w_i^2) \otimes (v_i^1,v_i^2)))) = \\
        &=\sum_i\alpha_i\mu_{W',V'}\left((f(w_i^1),f(w_i^2))\otimes (g(v_i^1),g(v_i^2))\right)\\
        &=\sum_i\alpha_i\hat\mu_{W',V'}\left((f(w_i^1),f(w_i^2))\otimes (g(v_i^1),g(v_i^2))\right)\\
        &=\sum_i\alpha_i\Bigr(f(w_i^1)\otimes g(v_i^1)-f(w_i^2)\otimes g(v_i^2),\\&f(w_i^1)\otimes g(v_i^2)+f(w_i^2)\otimes g(v_i^1)\Bigr). 
    \end{align*}
    On the other, we have that for every element 
    \begin{align*}
        &\mathds{C}(f \otimes g)\left(\mu_{W,V}\left(\sum_i \alpha_i (w_i^1,w_i^2) \otimes (v_i^1,v_i^2)\right)\right) =\\
        &=\sum_i \alpha_i \mathds{C}(f \otimes g)\left(w_i^1 \otimes v_i^1-w_i^2 \otimes v_i^2,w_i^1 \otimes v_i^2+w_i^2 \otimes v_i^1\right)\\
        &=\sum_i \alpha_i\Bigr(f\otimes g\left(w_i^1 \otimes v_i^1-w_i^2 \otimes v_i^2\right),\\& f \otimes g \left(w_i^1 \otimes v_i^2+w_i^2 \otimes v_i^1\right)\Bigr)\\
        &=\sum_i \alpha_i\Bigr(f(w_i^1) \otimes g(v_i^1)-f(w_i^2) \otimes g(v_i^2),\\& f(w_i^1) \otimes g(v_i^2)+f(w_i^2) \otimes g(v_i^1)\Bigr).
    \end{align*}
    
    \subsection{Associativity} 
    
    We now need to show that the coherence condition of associativity which is described by the condition that for every object $V,W,Z \in \realvectexamplep$ the following diagram
    \begin{widetext}
    \begin{equation*}
       \begin{tikzcd}
    (\mathds{C}(V) \otimes \mathds{C}(W)) \otimes \mathds{C}(Z) \arrow[rr, "\cong"'] \arrow[dd, "{\mu_{V,W} \otimes \mathrm{id}}"] &  & \mathds{C}(V) \otimes (\mathds{C}(W) \otimes \mathds{C}(Z)) \arrow[dd, "{\mathrm{id} \otimes \mu_{W,Z}}"] \\
                                                                                                                                    &  &                                                                                                           \\
    (\mathds{C}(V \otimes W) \otimes \mathds{C}(Z) \arrow[dd, "{\mu_{V \otimes W,Z}}"]                                              &  & \mathds{C}(V) \otimes \mathds{C}(W \otimes Z) \arrow[dd, "{\mu_{V,W \otimes Z}}"]                         \\
                                                                                                                                    &  &                                                                                                           \\
    (\mathds{C}((V \otimes W) \otimes Z) \arrow[rr]                                                                                 &  &  \mathds{C}(V \otimes (W \otimes Z))                                                                     
    \end{tikzcd}
    \end{equation*}
    \end{widetext}
    commute. The first horizontal arrow denotes the associator isomorphism in $\mathbf{Vect}_{\mathds{C}}$ (guaranteed to exist since $\mathbf{Vect}_{\mathds{C}}$ is a monoidal category). The second horizontal arrow denotes the action of the complexification functor on the associator isomorphism in $\realvectexamplep$. Therefore, showing that the above diagram commutes is equivalent to show that 
    \begin{align}
        \mu_{V,W \otimes Z} \circ \mathrm{id} \otimes \mu_{W,Z} \circ a^{\mathds{C}}_{\mathds{C}(V),\mathds{C}(W),\mathds{C}(Z)} = \label{eq: associativity}\\
        \mathds{C}(a^{\mathds{R}}_{V,W,Z})\circ \mu_{V \otimes W,Z}\circ \mu_{V,W}\otimes \mathrm{id}_{\mathds{C}(Z)} \nonumber 
    \end{align}
    Since everything extends linearly, we can just verify these on the simple tensors $(v_1,v_2)\otimes (w_1,w_2) \otimes (z_1,z_2)$. If we consider a generic such element the left-hand side of Eq.~\eqref{eq: associativity} becomes
        \begin{widetext}
        \begin{align*}
        &\mu_{V,W \otimes Z} \circ \mathrm{id} \otimes \mu_{W,Z} \circ a^{\mathds{C}}_{\mathds{C}(V),\mathds{C}(W),\mathds{C}(Z)} \left(((v_1,v_2)\otimes (w_1,w_2)) \otimes (z_1,z_2) \right) =  \mu_{V,W \otimes Z} \left(  \left((v_1,v_2)\otimes  \mu_{W,Z}((w_1,w_2) \otimes (z_1,z_2))\right) \right)\\
        &=\mu_{V,W\otimes Z}\left((v_1,v_2)\otimes (w_1 \otimes z_1-w_2\otimes z_2,w_1\otimes z_2+w_2\otimes z_1)\right)\\
        &=\left(v_1\otimes (w_1 \otimes z_1-w_2\otimes z_2)-v_2\otimes (w_1\otimes z_2+w_2\otimes z_1), v_1\otimes (w_1\otimes z_2+w_2\otimes z_1)+v_2\otimes (w_1 \otimes z_1-w_2\otimes z_2)\right)\\
        &=\Bigr(v_1 \otimes (w_1 \otimes z_1)-v_1 \otimes (w_2 \otimes z_2)-v_2 \otimes (w_1 \otimes z_2)-v_2 \otimes (w_2 \otimes z_1),\\&
        v_1\otimes (w_1 \otimes z_2)+v_1\otimes (w_2 \otimes z_1)+v_2 \otimes (w_1 \otimes z_1)-v_2\otimes  (w_2 \otimes z_2)\Bigr)
        \end{align*}
        \end{widetext}
        and the right-hand side becomes
        \begin{widetext}
        \begin{align*}
            &\mathds{C}(a_{V,W,Z}^{\mathds{R}})\circ \mu_{V\otimes W,Z}\circ \mu_{V,W}\otimes \mathrm{id}_{\mathds{C}(Z)}(((v_1,v_2)\otimes (w_1,w_2))\otimes (z_1,z_2)) \\
            &= \mathds{C}(a_{V,W,Z}^{\mathds{R}})\left(\mu_{V \otimes W,Z}\left((v_1 \otimes w_1-v_2\otimes w_2,v_1\otimes w_2+v_2\otimes w_1)\otimes (z_1,z_2)\right)\right)\\
            &=\mathds{C}(a_{V,W,Z}^{\mathds{R}})\left( (v_1 \otimes w_1-v_2\otimes w_2)\otimes z_1-(v_1\otimes w_2+v_2\otimes w_1)\otimes z_2,(v_1 \otimes w_1-v_2\otimes w_2)\otimes z_2+(v_1\otimes w_2+v_2\otimes w_1)\otimes z_1\right)\\
            &=\Bigr(v_1\otimes (w_1 \otimes z_1)-v_2\otimes (w_2\otimes z_1)-v_1\otimes (w_2\otimes z_2)-v_2\otimes (w_1\otimes z_2),\\&
            v_1 \otimes (w_1\otimes z_2)-v_2 \otimes (w_2 \otimes z_2)+v_1 \otimes (w_2 \otimes z_1)+v_2 \otimes (w_1 \otimes z_1)\Bigr)
        \end{align*}
        \end{widetext}
    The proof is then complete by extending this result $\mathds{C}$-linearly to every tensor product. Note that some of the brackets introduced were ``unnatural'', and used for clarity of the proof, since we know that $\mathbf{Vect}_{\mathds{K}}$ is a symmetric monoidal category. 

    Diagrammatically, this shows that we can represent parallel processes in the natural manner, i.e., 
    \begin{equation}
        \input{FIRST/App_5.tikz}
    \end{equation}
    where we do not need to keep track of artificial brakets. 
    
    \textit{Unitality.}---To complete the unitality (triangle) coherence condition we need to show that for every real vector space $V$ the two natural ways of inserting the unit agree. With respect to the isomorphism $\varepsilon: I_{\mathds{C}} \to \mathds{C}(I_{\mathds{R}})$ defined above, we must show that two diagrams commute. The first one is (where we denote the monoidal units as $I_{\mathds{R}}$ and $I_{\mathds{C}}$):
    \begin{equation}
        \begin{tikzcd}[column sep=huge, row sep=large]
        \mathds{C} \otimes \mathds{C}(V)
          \arrow[r,"\;\varepsilon\otimes\mathrm{id}\;"]
          \arrow[d,"\ell^{\mathds{C} }_{\mathds{C} (V)}"']
        & 
        \mathds{C} (\R)\otimes \mathds{C} (V)
          \arrow[d,"\mu_{\R,V}"]
        \\
        \mathds{C} (V)
          & 
        \mathds{C} (\R\otimes V)
          \arrow[l,"{\mathds{C} (\,\ell^{\R}_{V}\,)}"']
        \end{tikzcd}
    \end{equation}
    where we have denoted $\ell^{\mathds{C}}_{\mathds{C}(V)}$ the left-unitor for $\mathbf{Vect}_{\mathds{C}}$, i.e, the natural isomorphism such that for every element $\mathds{C}(V)$ of $\mathbf{Vect}_{\mathds{C}}$ one has that $\ell^{\mathds{C}}_{\mathds{C}(V)}: I_{\mathds{C}} \otimes \mathds{C}(V) \to \mathds{C}(V)$. The diagram above commutes since, for any element $\alpha \otimes (v_1,v_2) \in I_{\mathds{C}}\otimes \mathds{C}(V)$, we have that   
    \begin{align*}
        &\mathds{C}(\ell_V^\mathds{R})\circ\mu_{\mathds{R},V}\circ \varepsilon \otimes \mathrm{id}_{\mathds{C}(V)}(\alpha \otimes (v_1,v_2)) \\
        &= \mathds{C}(\ell_V^\mathds{R})\circ \mu_{\mathds{R},V}((\alpha_1,\alpha_2) \otimes (v_1,v_2))\\
        &=\mathds{C}(\ell_V^\mathds{R})(\alpha_1 \otimes v_1-\alpha_2 \otimes v_2, \alpha_1 \otimes v_2 + \alpha_2 \otimes v_1)\\
        &=\mathds{C}(\ell_V^\mathds{R})(1 \otimes \alpha_1 v_1-1 \otimes \alpha_2 v_2, 1 \otimes \alpha_1 v_2 + 1 \otimes \alpha_2v_1)\\
        &=(\alpha_1 v_1-\alpha_2 v_2, \alpha_1 v_2 + \alpha_2 v_1)\\
        &=\ell_{\mathds{C}(V)}^{\mathds{C}} \left (1 \otimes (\alpha_1 v_1-\alpha_2 v_2, \alpha_1 v_2 + \alpha_2 v_1) \right)\\
        &=\ell_{\mathds{C}(V)}^{\mathds{C}} \left (1 \otimes \alpha (v_1,v_2)\right)\\
        &=\ell_{\mathds{C}(V)}^{\mathds{C}} \left (\alpha \otimes (v_1,v_2) \right),
    \end{align*}
    where in the last steps we have used that 
    $$(\alpha_1+i\alpha_2)(v_1,v_2) = (\alpha_1v_1-\alpha_2v_2,\alpha_2v_1+\alpha_1v_2)$$ as by Def.~\ref{def: complexification_vector_spaces}, together with properties of $\otimes$. 
    The second one is:
    \begin{equation}
        \begin{tikzcd}[column sep=huge, row sep=large]
        \mathds{C}(V)\otimes \mathds{C}
          \arrow[r,"\;\mathrm{id}\otimes\varepsilon\;"]
          \arrow[d,"r^{\mathds{C}}_{\mathds{C}(V)}"']
        &
        \mathds{C}(V)\otimes \mathds{C}(\R)
          \arrow[d,"\mu_{V,\R}"]
        \\
        \mathds{C}(V)
        &
        \mathds{C}(V\otimes \R)
          \arrow[l,"{\mathds{C}(\,r^{\R}_{V}\,)}"']
        \end{tikzcd}
    \end{equation}
    and we do not repeat the argument as it follows trivially the same steps as for the left unitor, changing the order of the tensor products. 
    
    Diagrammatically, we conclude that we can represent the action of the complexification functor on the empty wire as another empty wire.

In the proof above, we have showed that $\mathds{C}$ is a faithful strong monoidal functor. This implies that, as a map between different categories, this is (up to an extension as explained in the main text) a diagram-preserving map. 

\end{appendix}

\end{document}